\documentclass[10pt, journal]{IEEEtran}

\IEEEoverridecommandlockouts  
\usepackage{amsmath, amsthm, amssymb}
\usepackage{listings}
\usepackage{indentfirst}
\usepackage{bbm}
\usepackage{graphicx}
\usepackage{subfigure}
\usepackage{float}
\usepackage{extarrows}
\usepackage{epstopdf}
\usepackage{caption}
\usepackage{cite,color}
\usepackage{todonotes}
\usepackage{tikz,pgfplots}
\usepackage{xcolor}
\usetikzlibrary{automata,positioning}
\usetikzlibrary{arrows,shapes,chains}
\usepgflibrary{patterns}
\usepackage{algorithm}
\usepackage{algorithmic}
\usepackage{tikz,pgfplots}
\usepackage{subcaption}
\usepackage{caption}
\usetikzlibrary{arrows, intersections}
\pgfplotsset{plot coordinates/math parser=false}

\newtheorem{lemma}{Lemma}

\newtheorem{theorem}{Theorem}
\newtheorem{definition}{Definition}

\newtheorem{proposition}{Proposition}
\newtheorem{remark}{Remark}

\newtheorem{fact}{Fact}
\usepackage{enumitem}
\setlist[itemize]{leftmargin=*}

\newcommand{\set}[1]{\left\lbrace#1\right\rbrace}
\newcommand{\setIn}[1]{\mathbbm{1}_{\set{#1}}}

\begin{document}
	
\title{\huge Timely Requesting for Time-Critical Content Users in Decentralized F-RANs}
\date{}
\author{Xingran Chen, {\it Member}, IEEE, \IEEEmembership{}
	Kai Li,\IEEEmembership{}
	and Kun Yang, {\it Fellow}, IEEE\IEEEmembership{}
	\IEEEcompsocitemizethanks 
	{
	\IEEEcompsocthanksitem	This work was accepted for publication in the IEEE Transactions on Networking in August 2025.
		
	\IEEEcompsocthanksitem This work was supported by Young Scientists Fund of the National Natural Science Foundation of China (Grant No. 62401111), Natural Science Foundation of China (Grant No. 62132004), the Jiangsu Major Project on Basic Researches (Grant No. BK20243059), and Gusu Innovation Project for Talented People (Grant No. ZXL2024360), High-Tech District of Suzhou City (Grant No. RC2025001), and Quzhou Goveronment (Grant No. 2024D007, 2023D005). \textit{(Corresponding Author: Kun Yang)}
				
	\IEEEcompsocthanksitem Xingran Chen and Kai Li are with School of Information and Communication Engineering,  University of Electronic Science and Technology of China, Chengdu, 611731, China (E-mail: xingranc@ieee.org, 202221011122@std.uestc.edu.cn).
		\IEEEcompsocthanksitem  Kun Yang is with the State Key Laboratory of Novel Software Technology, Nanjing University, Nanjing, 210008, China,  School of Intelligent Software and Engineering, Nanjing University (Suzhou Campus), Suzhou, 215163, China, and School of Information and Communication Engineering, University of Electronic Science and Technology of China, Chengdu, 611731, China (E-mail: kunyang@nju.edu.cn).}
}

\maketitle
\begin{abstract}
With the rising demand for high-rate and timely communications, fog radio access networks (F-RANs) offer a promising solution. This work investigates age of information (AoI) performance in F-RANs, consisting of multiple content users (CUs), enhanced remote radio heads (eRRHs), and content providers (CPs). Time-critical CUs need rapid content updates from CPs but cannot communicate directly with them; instead, eRRHs act as intermediaries. CUs decide whether to request content from a CP and which eRRH to send the request to, while eRRHs decide whether to command CPs to update content or use cached content. We study two broad classes of policies: (i) oblivious policies, where decision-making is independent of historical information, and (ii) non-oblivious policies, where decisions are influenced by historical information. We first derive closed-form expressions for the average AoI of eRRHs under both policy types. Due to the complexity of calculating closed-form expressions for CUs, we then derive general upper bounds for their average AoI. Next, we identify optimal policies for both types. Under both optimal policies, each CU requests content from each CP at an equal rate. When demand is low or resources are limited, all requests are consolidated to a single eRRH; when demand is high and resources are ample, requests are evenly distributed among eRRHs. eRRHs command content from each CP at an equal rate under an optimal oblivious policy, while prioritize the CP with the highest age under an optimal non-oblivious policy. Our numerical results validate these theoretical findings.  We further extend our analytical framework to two generalized scenarios, and simulations confirm the validity of our conclusions.
\end{abstract}

\section{Introduction}\label{sec: Introduction}

The proliferation of multimedia services has led to massive volumes of data-intensive and time-sensitive applications generated by end-users in wireless networks. This has led to a surge in the demand for high-rate and timely communications. Recently, fog radio access networks (F-RANs) have emerged as a promising solution to meet this growing demand. Fog computing, a new paradigm, provides a platform for local computing, distribution, and storage on end-user devices, as opposed to relying solely on centralized data centers   \cite{fogcomputing}. As an evolution of cloud radio access networks (C-RANs) \cite{C-RAN}, F-RANs feature fog access points to a centralized baseband signal processing unit via fronthaul links.  Each fog access point is equipped with caching capabilities to proactively store content \cite{FRAN}, enabling  end-users to access data locally without needing to reach the baseband unit (BBU) pool, thus reducing the load on both the fronthaul links and the BBU pool.

\begin{figure}[htbp]
	\centering
	\includegraphics[height=4.5cm, width=6cm]{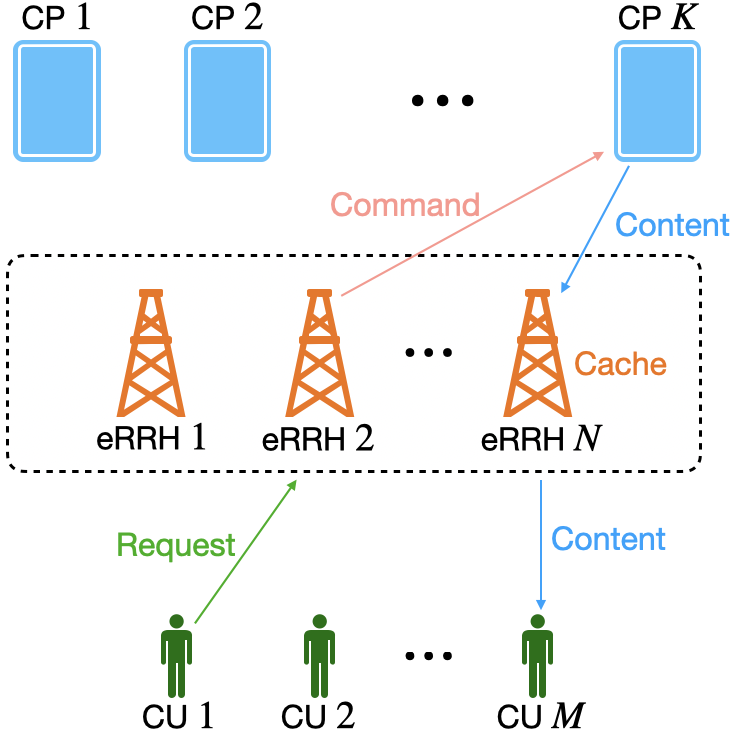}
	\caption{An example of a decentralized F-RAN.}
	\label{F-RAN0}
\end{figure}

This paper focuses on ensuring timely content delivery for latency-sensitive end-users, referred to content users (CUs), in decentralized F-RANs (see Fig.~\ref{F-RAN0}). In an F-RAN, multiple content providers (CPs), enhanced remote radio heads (eRRHs), and CUs interact in a network. CUs, which require timely content from CPs, cannot communicate directly with CPs. Instead, eRRHs act as gateways. Upon receiving requests from CUs, eRRHs can either command CPs to provide new content or retrieve cached content \cite{on-demandAoI}. Given that all CUs are time-critical,  the objective is to optimize the information freshness for CUs by proposing optimal policies for both CUs and eRRHs.  We quantify freshness using the {\it Age of Information (AoI)} metric \cite{AoIconcept}, which reflects the staleness of data at the receiver.  This metric is crucial in Internet of Things (IoT) applications where the timeliness of information is essential, such as in system status monitoring \cite{AoIconcept}. Addressing this problem involves two main sequential steps: (i) CUs making request decisions based on demand and resource availability, and (ii) eRRHs making command decisions based on available resources. 
Optimizing AoI in F-RANs has significantly practical implications, particularly for real-time applications that rely on timely data updates such as smart cities, autonomous vehicles, industrial automation, and health monitoring \cite{cxrTIT, cxrinfocom, Timelycache}. Improved AoI management ensures that these systems operate with up-to-date information, reducing latency and enhancing overall performance and reliability.

This problem presents the following key challenges:
\begin{itemize}
	\item [](i) Decentralization among CUs and eRRHs:  In a decentralized F-RAN, each CU and eRRH makes decisions based solely on local information, making coordination among them inherently difficult.
	\item [](ii)  Inter-dependence within each CU: During one time slot, a CU can send at most one request per CP to eRRHs,  leading to complex inter-dependencies that significantly complicate analysis, and the {\it age of version} metric introduced in previous works \cite{Freshcaching, Cacheupdating, Timelycache, cacherecommendation, timelyproactivecache} cannot be applied. 
	\item [] (iii) Complexity from the network topology: An F-RAN operates as a two-hop network with multiple receiver nodes at each hop.  Designing optimal policies is challenging due to the limited communication bandwidth.
	\item [] (iv) Historical information: Since CUs require the freshest content, any outdated content in the eRRHs' local cache must be discarded. This requirement leads to an objective function involving multiple nested {\it min} operations, which complicates the analytical derivation.
\end{itemize}
By addressing these challenges, our objective becomes clearer: to develop and analyze an optimal decentralized strategy that effectively utilizes historical information to optimize freshness for time-critical CUs in F-RANs.

\subsection{Related Work}\label{subsec: Related Work}
The AoI has been proposed as a metric to quantify the freshness of information at the receiver side \cite{AoIconcept}. By definition, AoI measures the time elapsed since the generation of the most recently received update. It depends jointly on the update frequency and the transmission delay experienced in the communication network \cite{cxrTIT}. This work intersects three major research directions: (i) AoI in multi-hop networks, (ii) AoI in edge/fog computing networks, and (iii) AoI in cache updating systems. These areas are not mutually exclusive and share conceptual overlaps. We review relevant studies in each of these directions below.

\subsubsection{AoI in Multi-hop Networks} 
A significant body of research has focused on analyzing AoI in single-hop networks. However, studies on AoI in multi-hop networks are relatively limited due to the added complexity from network topologies and inter-node interference.

One of the earliest works addressing AoI optimization in general networks is \cite{MultisourceMultihop}. This study considered multi-source multi-hop wireless networks where nodes function as both senders and receivers. It derived lower and upper bounds on average and peak AoI and proposed an optimal scheduling policy to minimize these metrics. However, the framework is limited to connected networks with nine or fewer nodes. In contrast, \cite{FundamentalScheduling} proposed near-optimal periodic status update scheduling policies using graph properties and obtained lower bounds for average and peak AoI. These policies can be applied to any connected topology. Another early work, \cite{AoIMN}, investigated the minimization of AoI for a single information flow in interference-free multi-hop networks. It proved that a preemptive last-generated, first-served policy results in smaller age processes across all nodes in the network compared to any other causal policy (in a stochastic ordering sense) if the packet transmission times are exponentially distributed.

Subsequent studies have explored optimal policies in general multi-hop networks. \cite{Multihopwirelessnetworks} proposed a near-optimal network control policy, achieving an average AoI close to a theoretical lower bound. \cite{Modiano-1, Modiano-2} investigated optimal stationary randomized policies and optimal scheduling policies, termed age difference policies and age debt policies. The optimal age debt policy can be applied to general multi-hop networks with unicast, multicast, and broadcast flows. Additionally, \cite{BoostingHindering} explored the trade-off between AoI and throughput in general multi-hop wireless networks, identified Pareto-optimal points and provided insights into balancing AoI and throughput. \cite{ODFMmultihop} examined AoI optimization using OFDM for multi-channel spectrum access, focusing on the impacts of real-world factors such as orthogonal channel access and wireless interference on AoI.
\cite{Non-Poissonmultihop} investigated AoI in multi-hop multicast cache-enabled networks with inter-update times that are not necessarily exponentially distributed. The study demonstrated that the expected AoI has an additive structure and is directly proportional to the variance of inter-update times across all links. Finally, \cite{cxrgnn} proposed transferable decentralized policies for minimizing AoI and estimation error, making them applicable to networks of arbitrary size.

\subsubsection{AoI in Edge/Fog Computing Networks}
In mobile edge computing (MEC) and fog computing networks, information  typically undergoes  two phases: transmission and processing. Consequently, mathematical models for these networks are often established as two-hop networks and tandem queues.

The first study to focus on the AoI for edge computing applications is \cite{AoIedge1}, which primarily calculated the average AoI. \cite{MPeakAoI} investigated the average AoI in a network with multiple sensors and one destination. The study derived analytical expressions for the average peak AoI for different sensors and designed a derivative-free algorithm to obtain the optimal updating frequency. A generic tandem model with a first-come-first-serve discipline was considered in \cite{PAoIedge}. The authors obtained the distribution of peak AoI for both M/M/1-M/D/1 and M/M/1-M/M/1 tandems. Building on this, \cite{ComputationTransmission} took a step further by considering both average and peak AoI in general tandems with packet management. This work illustrated the tradeoff between computation and transmission when optimizing AoI. In \cite{AoIedgecomputing}, both local computing and edge computing schemes under a first-come-first-serve discipline were considered in an MEC system with multiple users and a single base station (BS). The average AoI of these two computing schemes was derived.

\cite{MultiaccessEdge} and \cite{NOMAedge} investigated information freshness in MEC networks from a multi-access perspective.
\cite{MultiaccessEdge} studied a MEC scenario in which a base station serves multiple IoT traffic streams, proposing an optimal scheduling algorithm via integer linear programming. Along the same lines, \cite{NOMAedge} investigated a NOMA-based offloading system where devices share uplink channels to offload tasks to an MEC-enabled access point (AP), and derived a closed-form AoI solution using queuing theory.

\cite{decentralizedMEC} and \cite{RLEdge} explored the average AoI in MEC networks using other mathematical tools. \cite{decentralizedMEC} considered MEC-enabled IoT networks with multiple source-destination pairs and heterogeneous edge servers. Using game-theoretical analysis, an age-optimal computation-intensive update scheduling strategy was proposed based on Nash equilibrium. Reinforcement learning is also a powerful tool in this context. \cite{RLEdge} proposed a computation offloading method based on a directed acyclic graph task model, which models task dependencies. The algorithm integrates deep Q-learning techniques—including standard, double, and dueling architectures—to effectively minimize AoI in edge computing environments.

\subsubsection{AoI in Cache Updating Systems}
In cache updating systems, maintaining up-to-date cached data is essential to avoid serving stale information to users,  which may lead to incorrect decisions. In \cite{CacheYates}, a system is examined where a local cache maintains multiple content items. The authors introduced a popularity-weighted AoI metric for updating dynamic content in a local cache, demonstrating that the optimal policy involves updating the items in the cache at a rate proportional to the square root of their popularity.

Several studies, including \cite{Freshcaching, Cacheupdating, cacherecommendation, timelyproactivecache, Timelycache}, introduced a new metric called the age of version to measure information freshness in cache updating systems. In \cite{Freshcaching}, a caching scenario for dynamically changing content was considered, and the authors showed that the optimal caching strategy allocates cache space based solely on item popularity.  In \cite{Cacheupdating} examined three types of cache updating systems: (i) one source, one cache, and one user; (ii) one source, a sequence of cascading caches, and one user; and (iii) one user, one cache, and multiple sources. In all cases, the optimal rate allocation policies are threshold-based. \cite{Timelycache} studied timely cache updating policies in parallel multi-relay networks, deriving an upper bound on the average AoI for users and proposing a sub-optimal policy using a stochastic hybrid system approach. \cite{cacherecommendation} investigated the optimal scheduling of cache updates while accounting for content recommendation and AoI. Although the problem was proven NP-hard, efficient algorithms were proposed for the sub-optimization problem. The framework was extended in \cite{timelyproactivecache} to a cache updating system with randomly located caches following a Poisson point process, where the distribution of user-perceived version age of files was derived.

Another key direction in cache updating systems is optimizing AoI for energy-harvesting sensors. The average AoI for an energy-harvesting sensor with caching capability was studied in \cite{CachedEH}, where a probabilistic model was used to determine decisions, and the closed form was derived. A request-based scenario was considered in \cite{requestoriented}, where a cache-enabled base station stores the most recent status observed by energy-harvesting sensors and delivers the cached status to applications upon request. The proposed optimal scheduling policy achieved a 16\% performance gain compared to the traditional greedy policy.

The work most closely related to ours is \cite{Timelycache}, which investigated timely cache updating policies in parallel multi-relay networks. Although both works consider similar network topologies, our setting exhibits several critical differences:  (i) Applications: \cite{Timelycache} addressed cache updating systems for $N$ files in a time-continuous setting, while our work focused on decentralized F-RANs in a time-discrete setting. (ii) Number of senders: \cite{Timelycache} involved only one sender, whereas our work included multiple senders, making our scenario more general since the decisions of one sender affect others. (iii) Decisions within senders: In \cite{Timelycache}, decisions within the sender are independent, whereas in our work, decisions are inter-dependent, significantly increasing the analysis complexity and making our scenario more general. Our theoretical results encompass those of \cite{Timelycache} (see the discussions below Lemma~\ref{lem: optimal H h OS}). (iv) Policies: \cite{Timelycache} used oblivious policies, where decisions are independent of historical information. In contrast, our work includes both oblivious and non-oblivious policies, with the latter depending on historical information. Therefore, our work investigated a more general case and achieved more comprehensive results.

\subsection{Contributions}\label{subsec: contributions}
This paper addresses the issue of optimizing the information freshness for latency-sensitive CUs in F-RANs.  At each time slot, every CU independently makes decisions based on its own information, specifically: (i) whether to send a request to eRRHs for each CP, and (ii) which eRRH to send it to. Upon receiving requests, eRRHs decide whether to command CPs to send new content or to retrieve content from their local cache. The objective is to minimize the average AoI for CUs by designing decentralized stationary randomized policies. 
We explore two policy classes: (i) oblivious policies, where decisions are made independently of historical information, and (ii) non-oblivious policies, where decisions are influenced by historical information \cite{cxrinfocom, cxrgnn}. 
We begin by focusing on an F-RAN with $2$ CUs, $2$ eRRHs, and $2$ CPs (see Section~\ref{sec: optimal OSR policies} and Section~\ref{sec: optimal NSR policies}). Building on this foundation, we extend the framework in Appendix~\ref{App: Generalizations} to accommodate two more general scenarios: (i) collisions occurring over CU-to-eRRH links, and (ii) large-scale networks comprising $M$ CUs, $N$ eRRHs, and $K$ CPs.

In an F-RAN, CUs receive content from CPs via eRRHs, making their AoI directly dependent on that of the eRRHs. To understand this relationship, we establish recursive formulas for the AoI of both CUs and eRRHs and examine their average AoI. Under oblivious policies, where decisions are made without using historical context, the AoI at an eRRH (associated with a CP) follows a geometric distribution. We derive closed-form expressions for the average AoI of eRRHs in this scenario (see Theorem~\ref{thm: closed form Je}). For non-oblivious policies, where eRRHs utilize historical information, the analysis becomes more complex. To address this, we introduce an ergodic two-dimensional Markov chain to derive closed-form expressions for the average AoI (see Theorem~\ref{thm: expectation of Xnk}).

Deriving closed-form expressions for the average AoI of CUs is challenging due to the discarding of outdated content from eRRHs by the CUs. This introduces a set of {\it min} functions into the AoI calculations for CUs. As a result, we derive two upper bounds for the average AoI of CUs: see Theorem~\ref{thm: upper bounds of Jc} for the oblivious case and Theorem~\ref{thm: upper bounds of Jc N2S} for the non-oblivious case.  The first upper bound corresponds to an extreme scenario where eRRHs only send new packets and do not transmit previously cached ones. The second corresponds to an extreme case where CUs replace the most recently received packets with the currently delivered ones from eRRHs, even if they are older. The first bound tends to be more accurate in oblivious settings, while the second is more accurate under non-oblivious policies.

Even in the absence of closed-form expressions for AoI, we are able to theoretically characterize the optimal policies in both oblivious and non-oblivious scenarios. 
An optimal oblivious policy has the following characteristics (see Theorem~\ref{thm: b r optimal}):  (i) Each eRRH commands content from each CP at an equal rate.  (ii) Each CU sends requests to eRRHs for each CP at equal rates.  (iii) To request content from a CP, each CU consolidates all rates to a single eRRH when the demand is low or communication resources are limited, while requests are distributed evenly among eRRHs when the demand is high and communication resources are ample.  An optimal non-oblivious policy shares similar features with the optimal oblivious one, except for the first (see Theorem~\ref{thm: b r optimal NS}). In non-oblivious policies, each eRRH consolidates all rates to the CP with the highest current age in each time slot. This is because the command rates depend on the instantaneous age values of the CPs, making it optimal to focus all resources on the CP with the highest age.
These theoretical results are validated by our numerical simulations.

In the general scenario where collisions occur over CU-to-eRRH links (see Appendix~\ref{subApp: Collision Channels}), we impose a simplifying assumption for analytical tractability: each CU sends a request to every CP in every time slot. Under this assumption, we derive closed-form expressions for the average AoI at the eRRHs under both oblivious and non-oblivious policies (see Proposition~\ref{pro: closed form Je, oblivious, C} and Proposition~\ref{pro: expectation of Xnk, non-oblivious, C}).  Given the request constraint, the optimal policy for each CU—regardless of policy type—is twofold: (i) send requests to all CPs at equal rates; and (ii) for each CP, distribute requests evenly among the available eRRHs. On the eRRH side, an optimal choice under oblivious setting is to distribute requests uniformly across CPs, while an optimal choice under non-oblivious policy is to consolidate all rates to the CP with the highest age in a time slot (see Proposition~\ref{pro: b r optimal, Oblivious, C} and Proposition~\ref{thm: b r optimal NS, non-oblivious, c}). 
Next, we consider a large-scale network consisting of $M$ CUs, $N$ eRRHs, and $K$ CPs (see Appendix~\ref{subApp: Larger Networks}). For analytical tractability, we restrict attention to the oblivious policy regime. We derive closed-form expressions for the average AoI at the eRRHs (see Proposition~\ref{pro: closed form Je, oblivious, L}) and characterize the optimal strategy, which generalizes Theorem~\ref{thm: b r optimal} (see Proposition~\ref{pro: b r optimal, Oblivious, L}). The result reveals that the optimal request distribution adapts dynamically to the system conditions. Specifically, when demand is high and communication resources are abundant, requests should be evenly distributed across all eRRHs. However, as communication resources or demand decrease, the optimal approach shifts towards consolidating requests to a subset of eRRHs. Moreover, the greater the reduction in resources or demand, the smaller the subset of eRRHs that should be utilized. 

The paper is organized as follows. Section~\ref{sec: systemModel} introduces the system model. Sections~\ref{sec: optimal OSR policies} and~\ref{sec: optimal NSR policies} present the theoretical derivation of optimal policies under oblivious and non-oblivious scenarios, respectively. Section~\ref{sec: numerical results} presents simulation results, which validate our theoretical findings. Finally, Section~\ref{sec: Conclusions and Future Directions} concludes the paper.

\section{System Model}\label{sec: systemModel}

We consider a F-RAN consisting of $M$ CUs, $N$  eRRHs, and $K$ CPs. CUs, eRRHs and CPs are assumed to be statistically identical within their respective groups. We denote the sets of CUs, eRRHs, and CPs as $[M] = \{1,2,\cdots, M\}$, $ [N]= \{1,2,\cdots, N\}$, and $[K] = \{1,2,\cdots, K\}$, respectively. An example of a F-RAN is depicted in  Fig.~\ref{outline}.  CUs require time-critical content from CPs but cannot communicate with them directly. Instead, eRRHs in the network act as gateways \cite{on-demandAoI}. Content from CPs is transmitted through packets; therefore, we use the terms \textit{content} and \textit{packets} interchangeably. These eRRHs have the capability to cache and transmit packets. If CU $m$ needs the latest time-critical packet from CP $k$, it sends a request to eRRHs. The eRRH that receives the request is denoted as eRRH $n$. After receiving the request, eRRH $n$ can either retrieve a cached packet locally or command CP $k$ to update a new packet to serve the request. In this network, all CUs are time-sensitive and require packet updates as quickly as possible. 
\begin{figure}[htbp]
	\centering
	\includegraphics[height=4.5cm, width=7cm]{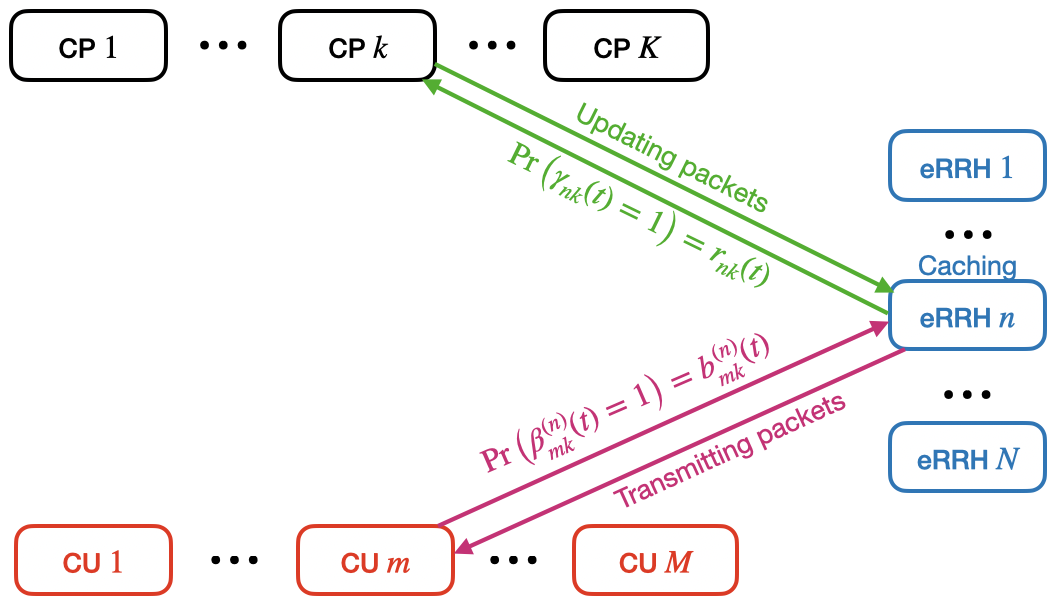}
	\caption{An example of a fog computing-based radio access network.}
	\label{outline}
\end{figure}

We consider a slotted-time system, where each time slot may span seconds, minutes, or longer depending on practical settings.
 At the beginning of every time slot, CU $m$ decides whether to send a request for packets from CP $k$ to eRRH $n$ or not. Denote $\beta_{mk}^{(n)}(t)\in\{0, 1\}$, $m\in [M]$, $n\in [N]$, $k\in [K]$; $\beta_{mk}^{(n)}(t)=1$ indicates that CU $m$ sends a request for CP $k$ to eRRH $n$ and $\beta_{mk}^{(n)}(t)=0$ otherwise. Furthermore, for every $k$,  CU $m$ sends at most one request, i.e., for all $m, k, t$,
\begin{align}\label{eq: at most one request}
\sum_{n\in[N]}\beta_{mk}^{(n)}(t)\leq1.
\end{align}
It is worth noting that $\{\beta_{mk}^{(n)}(t)\}_n$ are {\it not independent over $n$}. This setting differs from those in cache updating systems as discussed in \cite{Freshcaching, Cacheupdating, Timelycache, cacherecommendation, timelyproactivecache}, where the counterpart variables are independent over $n$. This inter-dependence in our model introduces additional complexity compared to the aforementioned studies. Denote the probability/rate of each action as 
\begin{align}\label{eq: indicator user m sensor k edge n}
\Pr\big(\beta_{mk}^{(n)}(t)=1\big)\triangleq b_{mk}^{(n)}(t)\in[0, 1].
\end{align} 
From \eqref{eq: at most one request} and \eqref{eq: indicator user m sensor k edge n}, we have 
\begin{align}\label{eq: at most one request probability}
\sum_{n\in[N]}b_{mk}^{(n)}(t)\leq 1,
\end{align}
for all $m, k, t$. CU $m$ sends a request in time slot $t$ with rate $\sum_{n\in[N]}b_{mk}^{(n)}(t)$.

Upon receiving requests from CUs, eRRH $n$ can either command CP $k$ to update with a new packet or retrieve a previously cached packet. Let $\gamma_{nk}(t) \in \{0,1\}$ represent the command of eRRH $n$ for CP $k$ at time slot $t$, given that $\sum_{m \in [M]} \beta_{mk}^{(n)}(t) \geq 1$. If $\gamma_{nk}(t) = 1$, eRRH $n$ commands CP $k$ to send a new packet; otherwise, $\gamma_{nk}(t) = 0$, indicating that the eRRH uses the cached packet. $\{\gamma_{nk}(t)\}_k$ are {\it independent over $k$}. Let  \begin{align}\label{eq: indicator edge n sensor k}
\Pr(\gamma_{nk}(t)=1) \triangleq r_{nk}(t)\in[0, 1]. 
\end{align}
Since we consider decentralized policies, the decision variables ${\beta_{mk}^{(n)}(t)}$ are {\it independent across $m$}.
Then, eRRH $n$ gets a new packet from CP $k$ in time slot $t$ with probability 
\begin{align}\label{eq: indicator edge n sensor k-1}
\big(1-\prod_{m\in[M]}(1-b_{mk}^{(n)}(t))\big) r_{nk}(t),
\end{align}
which reflects the decentralized nature of the decision-making process in the network.

Next, we display several key assumptions in our setting: (i) Requests from CUs and commands from eRRHs are small in size, rendering the transmission delays negligible, (ii) The transmission delay from CPs to eRRHs is fixed at {\it one time} slot, while the transmission delays of packets from eRRHs to CUs are negligible\footnote{For example, in a dense urban deployment of an F-RAN, multiple eRRHs are connected to the centralized cloud via a shared millimeter-wave wireless fronthaul link. During peak traffic hours, such as in the evening when many users stream high-definition videos or engage in online gaming, the fronthaul link experiences significant congestion. The limited bandwidth and variable link quality of the mmWave connection result in increased queuing delays. As a consequence, the CP-to-eRRH transmission delay becomes a bottleneck in the system, while the eRRH-to-CU link remains fast and stable due to the short-range high-quality wireless access.}.  (iii) We do not assume interference in transmissions, following the precedents  \cite{Freshcaching, Cacheupdating, cacherecommendation, timelyproactivecache, Timelycache}. In fact, interference can be avoided utilizing advanced techniques, such as power-domain non-orthogonal multiple access. 
Let $B$ represent the constraints on the actions of each CU, which include the limitations of radio resources (e.g., bandwidth) and the demand from external environments.  Given that CUs are statistically identical within their groups, we impose the following constraint:
\begin{align}
\sum_{n\in[N]}\sum_{k\in[K]}b_{mk}^{(n)}(t)\leq B,\,\,\text{for } m\in [M]\label{eq: bounds for b}.
\end{align}
A small $B$ indicates low content demand or limited radio resources, whereas a large $B$ corresponds to high demand and ample resources. Similarly, let $R$ represent the constraints on the actions of each eRRH, reflecting the resources available on the eRRHs' side,
\begin{align}
\sum_{k\in[K]}r_{nk}(t)\leq R,\,\,\text{for }n\in [N]\label{eq: bounds for r}.
\end{align}

\subsection{Age of Information}\label{subsec: AoI}
Packets can be received by both CUs and eRRHs. Accordingly, we define two types of AoI: one at the eRRHs and the other at the CUs.  By convention, the AoI evolves at the {\it end} of each time slot \cite{cxrTIT, cxrinfocom, cxrgnn}.

At time slot $t$, the most recently received packet from CP $k$ at eRRH $n$ has generation time $\tau_{nk}$. The AoI of CP $k$ at eRRH $n$ is defined as
\begin{align*}
g_{nk}(t)=t-\tau_{nk}.
\end{align*}
Since the transmission delay from CPs to eRRHs is fixed at one time slot,  $g_{nk}(t)$ evolves as
\begin{align}\label{eq: recursion of edge node AoI0}
g_{nk}(t+1) = & \setIn{\gamma_{nk}(t)\sum_{m\in[M]}\beta_{mk}^{(n)}(t)>0}\nonumber\\ 
+&\left(g_{nk}(t)+1\right)\setIn{\gamma_{nk}(t)\sum_{m\in[M]}\beta_{mk}^{(n)}(t)=0} ,
\end{align}
with $g_{nk}(0)=1$.

In time slot $t$, the most recently received packet from CP $k$ has the generation time $\tau'_{mk}$. The AoI of CP $k$ at CU $m$ is defined as 
\begin{align*}
h_{mk}(t)=t-\tau_{mk}'.
\end{align*}
To ensure freshness, any delivered packet that is older than the most recently received one is discarded.  Based on the constraint in \eqref{eq: at most one request}, the evolution of $h_{mk}(t)$ is given by:

\begin{align}\label{eq: recursion of user AoI}
h_{mk}(t+1) =&  \sum_{n\in[N]}\Big(\setIn{\beta_{mk}^{(n)}(t)\gamma_{nk}(t)=1}\nonumber\\
+&\big(\tilde{h}_{mk}^{(n)}(t)+1\big)\setIn{\beta_{mk}^{(n)}(t)\big(1-\gamma_{nk}(t)\big)=1}\Big)\nonumber\\ 
 +& \left(h_{mk}(t)+1\right)\setIn{\sum_{n\in[N]}\beta_{mk}^{(n)}(t)=0}.
\end{align}
with $\tilde{h}_{mk}^{(n)}(t) = \min\{h_{mk}(t), g_{nk}(t)\}$ and $h_{mk}(0)=1$.  The interdependence of decisions within each CU, as captured by \eqref{eq: at most one request} and \eqref{eq: recursion of user AoI}, complicates the analysis. As a result, the age of version metric introduced in \cite{Freshcaching, Cacheupdating, Timelycache, cacherecommendation, timelyproactivecache} is no longer applicable.

We calculate the average AoI of eRRHs and CUs, respectively, using the recursive expressions in \eqref{eq: recursion of edge node AoI0} and  \eqref{eq: recursion of user AoI}. The average AoI of eRRHs is defined as
\begin{align}\label{eq: eRRHs average AoI}
L_e=\lim_{T\to\infty}\frac{1}{T}\sum_{t=1}^{T}\frac{1}{NK}\sum_{n\in[N]}\sum_{k\in[K]}\mathbb{E}[g_{nk}(t)].
\end{align}
Similarly, the average AoI of CUs is defined as
\begin{align}\label{eq: CUs average AoI}
J_c=\lim_{T\to\infty}\frac{1}{T}\sum_{t=1}^{T}\frac{1}{MK}\sum_{m\in[M]}\sum_{k\in[K]}\mathbb{E}[h_{mk}(t)].
\end{align}

\subsection{Stationary Randomized Policies}\label{subsec: Stationary Randomized Policies}
We consider decentralized policies for CUs and eRRHs, where each CU and eRRH makes decisions based on its local information.  This decentralization structure implies that $\big\{\{\beta_{mk}^{(n)}(t)\}_{n, k}\big\}_m$ are independent across $m$, and $\big\{\{\gamma_{nk}(t)\}_{k}\big\}_n$ are independent across $n$.  We define a decentralized policy $\pi$ as a sequence of actions over time:
\begin{align}\label{eq: sequential strategy}
\pi =\Big\{\big\{\{b_{mk}^{(n)}(t)\}_{n, k}\big\}_m, \big\{\{r_{nk}(t)\}_{k}\big\}_n\Big\}_t.
\end{align} 
For analysis tractability, this paper focuses on a widely used class of decentralized policies, known as stationary randomized policies, where each action is taken with a fixed probability \cite{Modiano-1, Modiano-2}. We first introduce {\it oblivious} stationary randomized (OSR) policies, where decisions are independent of AoI information. Then, we define {\it non-oblivious} stationary randomized (NSR) policies, where decisions by eRRHs depend on the AoI information \cite{cxrinfocom, cxrgnn}.

\subsubsection{Oblivious Stationary Randomized Policies}\label{subsubsec: Oblivious Stationary Randomized Policies}
The idea behind OSR policies is that each node, including CUs and eRRHs, makes decisions based on predetermined probabilities, independent of AoI information. Specifically, each CU's decision reduces to 
\begin{align}\label{eq: decisinos of users}
b_{mk}^{(n)}(t)=b_{k}^{(n)},\text{ for all }m, t.
\end{align}
Each eRRH $n$ adopts a rate vector $\{r_1, r_2, \cdots, r_K\}$, with $r_1\geq r_2\geq\cdots\geq r_K$, specifying the probabilities of commanding different CPs. When eRRH $n$ receives requests for CP $k$ in time slot $t$ (i.e., $\sum_{m\in[M]}\beta_{mk}^{(n)}(t)\geq1$), it then  commands CP $k$ with probability $r_k$:
\begin{align}\label{eq: oblivious r}
r_{nk}(t)=&r_k \setIn{\sum_{m\in[M]}\beta_{mk}^{(n)}(t) \geq 1}.
\end{align}

\begin{definition}\label{def: OSR}
In a network, an O2S policy is defined by \eqref{eq: sequential strategy}, where $\{b_{mk}^{(n)}(t)\}_{m,n,k}$ and $\{r_{nk}(t)\}_{n,k}$ are specified by \eqref{eq: decisinos of users} and  \eqref{eq: oblivious r}, respectively.
\end{definition}

\subsubsection{Non-oblivious Stationary Randomized Policies}\label{subsubsec: Non-Oblivious Stationary Randomized Policies}
Under an NSR policy, eRRH $n$ utilizes the AoI information to make commanding decisions. A CP with a higher age is given higher priority in decision-making. Specifically, in time slot $t$, eRRH $n$ arranges $\{g_{nk}(t)\}_{k\in[K]}$ in descending order: $g_{nl_1}(t)\geq g_{nl_2}(t)\geq\cdots\geq g_{nl_K}(t)$. If eRRH $n$ receives one or more requests for CP $l_k$, it commands CP $l_k$ with rate:
\begin{align*}
r_{nl_k}(t) =& r_{k} \setIn{\sum_{m\in[M]}\beta_{ml_k}^{(n)}(t)\geq1}.
\end{align*}

However, if all CPs have large ages at eRRH $n$, then prioritization among them is unnecessary. In such cases, a predefined threshold $z$ is introduced.  If $g_{nl_K}(t)\geq z$, then each CP will be commanded at an equal rate\footnote{One can extend to other cases, for example, if $K-1$ CPs have large ages at eRRH $n$, there is no need to set priority. Specifically, if $g_{nl_{K-1}}(t)\geq z$, then every CP will be commanded with an equal probability $r_0$. The analytical framework closely resembles that discussed in Section~\ref{sec: optimal NSR policies}.}. The expression of $r_{nl_k}(t)$ is then:
\begin{align}\label{eq: nonoblivious r}
r_{nl_k}(t) =\left\{
\begin{aligned}
r_{k} \setIn{\sum_{m\in[M]}\beta_{ml_k}^{(n)}(t)\geq1}&&g_{nl_K}(t) < z\\
r_0 \setIn{\sum_{m\in[M]}\beta_{ml_k}^{(n)}(t)\geq1}&&g_{nl_K}(t)\geq z,
\end{aligned}
\right.
\end{align}
where $r_0 = \frac{1}{K}\sum_{k\in[K]}r_k$.

\begin{definition}\label{def: non-oblivious stationary strategy for all}
In a network, a NSR policy is defined by \eqref{eq: sequential strategy}, where $\{b_{mk}^{(n)}(t)\}_{m,n,k}$ and $\{r_{nk}(t)\}_{n,k}$  are provided in \eqref{eq: decisinos of users}  and  \eqref{eq: nonoblivious r}, respectively.
\end{definition}
The policy in Definition~\ref{def: non-oblivious stationary strategy for all} allows eRRHs to dynamically adjust commanding rates based on the ages of CPs, promoting efficient resource allocation while maintaining flexibility in handling varying network conditions.

\subsubsection{Feasible Region}\label{subsubsec: Feasible Region}
We now characterize the feasible region of stationary randomized policies. Let  $b_k\triangleq\sum_{n\in[N]}b_{k}^{(n)}$ represent the rate of a CU sending requests for CP $k$ to eRRHs. Define $b\triangleq\sum_{n\in[N]}\sum_{k\in[K]}b_{k}$ and  $r\triangleq\sum_{k\in[K]}r_k$. 
Based on \eqref{eq: at most one request probability},  \eqref{eq: bounds for b},  \eqref{eq: bounds for r}, \eqref{eq: oblivious r}, and \eqref{eq: nonoblivious r}, 
the feasible region of a stationary randomized policy $\pi$ is given by
\begin{align*}
\mathcal{F}\triangleq\set{(b, r, \set{b_k}_{k\in[K]})| b\leq B,  r\leq R,b_k\leq1,k\in[K]}.
\end{align*}

\subsection{Optimization Problems}\label{subsec: Optimization Problems}
Our goal is to deliver new packets to CUs with minimal delay. Mathematically, this goal can be formulated as minimizing the average AoI of CUs under a stationary randomized policy, i.e., 
\begin{align}\label{eq: optimal pi}
\pi^* = \min_{\pi} J_c,
\end{align}	
where $\pi$ can be either an OSR or an NSR policy. Since CUs communicate with CPs through eRRHs, the value of $J_c$ depends on $L_e$. 
To solve the optimization \eqref{eq: optimal pi}, we follow these $3$ steps: (i) analyze $L_e$ theoretically, (ii) analyze $J_c$ theoretically, and (iii) determine an optimal strategy $\pi^*$ as defined in \eqref{eq: optimal pi}. Throughout the rest of the paper, we assume $(b, r, \{b_k\}_{k\in[K]})\in\mathcal{F}$.

We provide a detailed analysis of $L_e$ and $J_c$ for a representative system 
consisting of $2$ CUs, $2$ eRRHs, and $2$ CPs (see Section~\ref{sec: optimal OSR policies} and Section~\ref{sec: optimal NSR policies}). Although our primary analysis focuses on this specific setting, the developed analytical methods generalize to broader scenarios.  To further broaden the applicability of our framework, we present two generalizations in Appendix~\ref{App: Generalizations}.  First, we consider the case of collisions over CU-to-eRRH links, where multiple CUs may simultaneously attempt to access the same eRRH, leading to  transmission collisions.  Second, we extend our analysis to large-scale networks involving arbitrary numbers of CUs ($M$), eRRHs ($N$), and CPs ($K$).

\section{Optimal OSR Policies}\label{sec: optimal OSR policies}
In this section, we investigate an optimal OSR policy in a network with $N=M=K=2$. To achieve this, we start by calculating $L_e$ as defined in \eqref{eq: eRRHs average AoI} and derive its closed form for $L_e$ (see Theorem~\ref{thm: closed form Je} in Section~\ref{subsubsec: Closed From of edge node}).  Due to the complexity of calculating $J_c$, its closed-form expression is intractable. Instead, we derive two upper bounds for $J_c$ (see Theorem~\ref{thm: upper bounds of Jc} in Section~\ref{subsec: Upper bounds of Jc O2S}). Finally, we solve the optimization \eqref{eq: optimal pi}  and obtain an optimal O2S strategy $\pi^*$ (see Theorem~\ref{thm: b r optimal} in Section~\ref{subsec: optimal strategy pic in oblivious network}).\footnote{
One may wonder why the closed-form expression for $L_e$ and the upper bounds of $J_c$ are necessary. The key reason lies in the architecture of  two-hop F-RANs, where the AoI performance depends on coordinated decisions made at both the eRRHs and the CUs. The decisions at these two layers are inherently coupled, making the optimization of the overall AoI particularly challenging. To address this, we adopt a layered and constructive analytical approach. We begin by analyzing the eRRH layer in isolation, where we are able to derive closed-form expressions for $L_e$, offering  insights into how local parameters influence the edge AoI. However, due to the coupling between layers, deriving a closed-form solution for $J_c$ is highly intractable. As an alternative, we derive upper bounds on $J_c$ to enhance analytical tractability. These bounds also provide meaningful guidance for policy design.}

\subsection{Closed Form of $J_e$}\label{subsubsec: Closed From of edge node}
Based on \eqref{eq: indicator edge n sensor k-1}, \eqref{eq: recursion of edge node AoI0}, and \eqref{eq: oblivious r}, the value of $g_{nk}(t)$ resets to $1$ with probability $r_k  \left(1 - (1 - b_k^{(n)})^2\right)$. Denote
\begin{align}\label{eq: zeta}
\zeta_{nk} \triangleq  1-(1-b_k^{(n)})^2,
\end{align}
and $\zeta_{nk}$ is a concave function with respect to  $b_k^{(n)}$. The update rule for $g_{nk}(t)$ in \eqref{eq: recursion of edge node AoI0} is then given by
\begin{align}\label{eq: oblivious recursion g-1}
g_{nk}(t+1) = \left\{
\begin{aligned}
&1&& \text{w.p. }r_k\zeta_{nk}\\
&g_{nk}(t)+1&&\text{w.p. }1-r_k\zeta_{nk}.
\end{aligned}
\right.
\end{align}

We use a graphical argument to calculate the AoI $L_e$ \cite{AgeSKaul}, as illustrated in Fig.~\ref{gnk}. 
\begin{figure}[htbp]
	\centering
	\includegraphics[height=5.25cm, width=6.25cm]{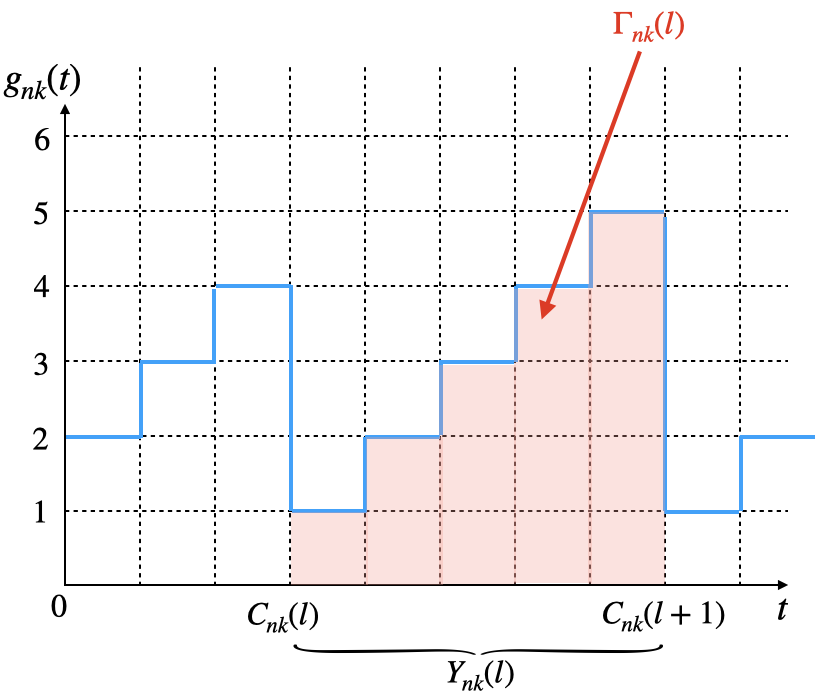}
	\caption{A trajectory of $g_{nk}(t)$.}
	\label{gnk}
\end{figure}
Let $C_{nk}(l)$ denote the timestamp when eRRH $n$ receives a new packet from CP $k$ for the $l$-th time. Let 
\begin{align*}
Y_{nk}(l) = C_{nk}(l+1) - C_{nk}(l).
\end{align*}
Since the transmission delays from CPs and eRRHs are fixed at $1$, we have $g_{nk}(C_{nk}(l))\equiv1$
for all $l$. Let $\Gamma_{nk}(l)$ denote the sum of $g_{nk}(t)$ over the interval $[C_{nk}(l), C_{nk}(l+1))$:
\begin{align*}
\Gamma_{nk}(l) = \sum_{t=C_{nk}(l)}^{C_{nk}(l+1)-1}g_{nk}(t) = \frac{1}{2}Y_{nk}^2(l) + \frac{1}{2}Y_{nk}(l).
\end{align*}

\begin{theorem}\label{thm: closed form Je}
Let $\zeta_{nk}$ be given in \eqref{eq: zeta}. The closed-form expression for the average AoI of eRRHs is provided by 
\begin{align}\label{eq: closed form Je}
L_e=\frac{1}{4}\sum_{n\in[2]}\sum_{k\in[k]}\frac{\mathbb{E}[\Gamma_{nk}(l)]}{\mathbb{E}[Y_{nk}(l)]}=\frac{1}{4}\sum_{n\in[2]}\sum_{k\in[k]}\frac{1}{r_k\zeta_{nk}}.
\end{align}
\end{theorem}
\begin{remark}	
It is straightforward to verify that \eqref{eq: closed form Je} is convex with respect to $r_1$ and $r_2$. Consequently, given decisions of CUs $\{b_k^{(n)}\}_{n, k}$, there exists a unique optimal  $\{r_k\}_{k\in[2]}$ \cite{convexoptimization}.
\end{remark}
\begin{proof}
The proof is given in Appendix~\ref{App: closed form Je}. 
\end{proof}
From the recursion of $g_{nk}(t)$ in \eqref{eq: oblivious recursion g-1}, it follows that $g_{nk}(t)$ has a geometric distribution with parameter $r_k\zeta_{nk}$. This parameter represents the probability of a new content sucessfully arriving at the eRRH.

\subsection{Upper Bounds of $J_c$}\label{subsec: Upper bounds of Jc O2S}

From \eqref{eq: recursion of user AoI}, \eqref{eq: decisinos of users}, and \eqref{eq: oblivious r}, we simplify the recursion of $h_{mk}(t)$ as follows:
\begin{equation}\label{eq: recursion of h-1}
\begin{aligned}
h_{mk}(t+1)= \left\{
\begin{aligned}
&1&&b_kr_k\\
&\min\{h_{mk}(t), g_{1k}(t)\}+1&& b_{k}^{(1)}(1-r_k)\\
&\min\{h_{mk}(t), g_{2k}(t)\}+1&& b_{k}^{(2)}(1-r_k)\\
&h_{mk}(t)+1&&1-b_k. 
\end{aligned}
\right.
\end{aligned}
\end{equation}
Due to the dependency between $\beta_{mk}^{(1)}(t)$ and $\beta_{mk}^{(1)}(t)$ (defined in Section~\ref{sec: systemModel}), the $\min\{\cdot\}$ function in \eqref{eq: recursion of h-1} complicates the analysis of $J_c$, making it difficult to derive a closed-form expression. Therefore, we aim to obtain two upper bounds for $J_c$. 

\begin{theorem}\label{thm: upper bounds of Jc}
Let $\zeta_{nk}$ be given in \eqref{eq: zeta}. Then, $J_c$ has two upper bounds,
\begin{align}
J_c\leq&\frac{1}{2}\sum_{k\in[2]}\frac{1}{b_kr_k}\label{eq: Jc O2S upper bound 1},\\
J_c\leq&\frac{1}{2}\sum_{k\in[2]}\Big(\frac{1-b_k+b_kr_k}{b_kr_k}+\sum_{n=1}^{2}\frac{b_k^{(n)}(1-r_k)}{r_k\zeta_{nk}}\Big)\label{eq: Jc O2S upper bound 2}.
	\end{align}
\end{theorem}
\begin{proof}
	The proof is given in Appendix~\ref{App: upper bounds of Jc}.
\end{proof}

The upper bound in \eqref{eq: Jc O2S upper bound 1} represents the average AoI in an extreme scenario where eRRHs only send new packets and do not transmit previously cached ones. The upper bound in \eqref{eq: Jc O2S upper bound 2} reflects the average AoI in an extreme case where CUs replace the most recently received packets with the newly delivered ones from eRRHs, even when those are older. Under oblivious policies, since the commanding rate $r_k$ remains fixed over time, the first extreme case results in more new content being received compared to the second. Consequently, the upper bound in \eqref{eq: Jc O2S upper bound 1} provides a more accurate estimate (see Fig.~\ref{network2_ub}(a) in Section~\ref{sec: numerical results}).

\subsection{Optimal OSR Strategy}\label{subsec: optimal strategy pic in oblivious network}
Now, we aim to obtain an optimal policy $\pi^*$ for $J_c$ in \eqref{eq: CUs average AoI}. To achieve this, we follow two steps: (i) Fix $b_k$ and $r_k$, and determine an optimal $\big((b_k^{(1)})^*, (b_k^{(2)})^*\big)$ that minimize $$\lim_{T\to\infty}\frac{1}{T}\sum_{t=1}^{T}\mathbb{E}[h_{mk}(t)];$$ (ii) Find an optimal policy $\pi^*$ to minimize $$J_c=\frac{1}{4}\sum_{m\in[2]}\sum_{k\in[2]}\lim_{T\to\infty}\frac{1}{T}\sum_{t=1}^{T}\mathbb{E}[h_{mk}(t)].$$
By systematically addressing these steps, we identify a policy that optimally minimizes the average AoI for the CUs.

Before deriving the optimal decisions theoretically, we first present an important observation regarding $h_{mk}(t)$ \big(from \eqref{eq: recursion of h-1}\big) that will help us to understand the optimal pair $\big((b_k^{(1)})^*, (b_k^{(2)})^*\big)$: The parameter $b_k$ represents the demand for CP $k$ from CUs. As $b_k$ increases,  the demand for new content from CUs also increases, which leads to smaller values for both $h_{mk}(t)$ and $g_{nk}(t)$.  According to \eqref{eq: recursion of h-1}, since $h_{mk}(t)$ is determined by $g_{nk}(t)$, these random variables are comonotonic \cite{nelsen2006introduction}. Due to the presence of the {\it min} function in  \eqref{eq: recursion of h-1}, $h_{mk}(t)$ changes more rapidly than $g_{nk}(t)$ as $b_k$ decreases, the influence of the min function in \eqref{eq: recursion of h-1} becomes more pronounced.  This suggests that the optimal decisions may differ significantly when $b_k$ is small compared to when it is large.
We define two stochastic processes,
\begin{align}\label{eq: upper bound H1}
	H_{mk, 1}(t+1) = \left\{
	\begin{aligned}
		&1&& b_k r_k\\
		&h_{mk}(t)+1&&1-b_k r_k;
	\end{aligned}
	\right.
\end{align}
and
\begin{align}\label{eq: upper bound H2}
	H_{mk, 2}(t+1) = \left\{
	\begin{aligned}
		&1&&b_kr_k\\
		&g_{1k}(t)+1&& b_k^{(1)}(1-r_k)\\
		&g_{2k}(t)+1&&b_k^{(2)}(1-r_k)\\
		&h_{mk}(t)+1&&1-b_k,
	\end{aligned}
	\right.
\end{align}
with $H_{mk, 1}(0) = H_{mk, 2}(0)=1$. From \eqref{eq: recursion of h-1}, \eqref{eq: upper bound H1}, and \eqref{eq: upper bound H2}, 
\begin{align}\label{eq: hH1H2}
h_{mk}(t) \overset{d}{=} \min\{H_{mk, 1}(t), H_{mk, 2}(t)\}
\end{align}
where $\overset{d}{=}$ represents equality {\it in distribution}.
By utilizing \eqref{eq: hH1H2}, we derive the following analytical results.

\begin{lemma}\label{lem: optimal H h OS}
Fix $b_k$ and $r_k$.  For $t\geq2$, there exists a threshold depending on $r_k$, denoted as  $w(r_k; t)$ such that,
\begin{itemize}
	\item [] (i) when $b>w(r_k; t)$, $(b_k^{(1)}, b_k^{(2)})=(b_k/2, b_k/2)$ is the global minimizer of $\mathbb{E}[h_{mk}(t)]$;
	\item [] (ii) when $b_k < w(r_k; t)$, either $(b_k^{(1)}, b_k^{(2)})=(b_k, 0)$ or $(0, b_k)$ is the global minimizer of $\mathbb{E}[h_{mk}(t)]$.
	\item [] (iii) when $b_k = w(r_k; t)$, either $(b_k^{(1)}, b_k^{(2)})=(b_k/2, b_k/2)$, $(b_k, 0)$, or $(0, b_k)$ is the  global minimizers of $\mathbb{E}[h_{mk}(t)]$.
	\item [] (iv) $w(r_k; t)$ decreases with $t$ and  converges as $t \to \infty$.
\end{itemize}
\end{lemma}
\begin{proof}
The roadmap of the proof is as follows: First, we demonstrate that $\mathbb{E}[h_{mk}(t)]$ is symmetric with respect to $(b_k^{(1)}, b_k^{(2)})$, and $\mathbb{E}[h_{mk}(t)]$ is continuous in $b_k^{(1)}$ and $b_k^{(2)}$, which implies that the minimum point can be found at $(0, b_k)$, $(b_k, 0)$, or $(b_k/2, b_k/2)$. Then, we prove the lemma by mathematical induction. Specifically, since $h_{mk}(t)$ becomes asymptotically stationary under stationary randomized policies, $h_{mk}(t+1) = h_{mk}(t)$ almost surely as $t\to\infty$, which implies that $\lim_{t\to\infty}w(r_k; t)$ converges. The detailed proof is provided in Appendix~\ref{App: lem: optimal H h OS}.
\end{proof}

\begin{lemma}\label{lem: optimal H1H2h OS}
Let $b, r$ be fixed and $(b, r)\in\mathcal{F}$,  $(b_k, r_k)=(b/2, r/2)$ is the global minimizer of $\sum_{k\in[2]}\mathbb{E}[h_{mk}(t)]$ for $t\geq1$.
\end{lemma}
\begin{proof}
The roadmap of the proof is as follows. We begin by verifying that if both CUs have the optimal allocation $\big((b_k^{(1)})^*, (b_k^{(2)})^*\big) = (b_k, 0)$ or $(0, b_k)$, the minimum choice for $(b_1, r_1)$ is $(b/2, r/2)$. Next, we show that if both CUs choose the optimal split $\big((b_k^{(1)})^*, (b_k^{(2)})^*\big)=(b_k/2, b_k/2)$, $(b/2, r/2)$ again yields the minimum. Finally, we demonstrate that the mixed case--where one CU chooses $\big((b_k^{(1)})^*, (b_k^{(2)})^*\big)=(b_k, 0)$ or $(0, b_k)$, and the other chooses $\big((b_k^{(1)})^*, (b_k^{(2)})^*\big)=(b_k/2, b_k/2)$--cannot occur. The full proof is provided in Appendix~\ref{App: optimal H1H2h OS}.
\end{proof}

\begin{theorem}\label{thm: b r optimal}
Let $b, r$ be fixed, $(b, r)\in\mathcal{F}$, and $ m,n,k\in[2]$. there exists a threshold (depending on $r$), $w(r)$, such that if $b>w(r)$, the optimal policy $\pi^*$ is given by
\begin{align}\label{eq: optimal OSR-1}
\pi^* = \{b_{mk}^{(n)}=b/4, r_k=r/2\};
\end{align}
if $b \leq w(r)$, an optimal policy
\begin{align}\label{eq: optimal OSR-2}
\pi^* = \{&(b_{mk}^{(n)}, b_{mk}^{(3-n)}) = (b/2, 0),  r_k=r/2\}.
\end{align}
\end{theorem}
\begin{proof}
From Lemma~\ref{lem: optimal H h OS}~(iii), we know that $\lim_{t\to\infty}w(r_{k}; t)$ exists and denote it by $w(r_k)$. The proof then follows directly from  Lemma~\ref{lem: optimal H h OS}~(i), (ii), (iii), and Lemma~\ref{lem: optimal H1H2h OS}. 
\end{proof}
\begin{remark}
If $b  < w(r)$, an alternative optimal policy is given by 
$\pi^* = \{(b_{mk}^{(n)}, b_{mk}^{(3-n)}) = (0, b/2),  r_k=r/2\}$.
If $b=w(r)$, then two alternative optimal policies exist:
$\pi^* = \{(b_{mk}^{(n)}, b_{mk}^{(3-n)}) = (0, b/2),  r_k=r/2\}$.
and  $\pi^* = \{b_{mk}^{(n)}=b/4, r_k=r/2\}$.
\end{remark}
It is worth noting that each CU's decision-making process involves two layers: first, determining how to allocate the requested data rates among the CPs; and second, for each selected CP, deciding how to distribute the requested rates among the available eRRHs. Based on Theorem~\ref{thm: b r optimal}, we derive the following insights into the optimal strategies for both CUs and eRRHs:
\begin{itemize}
\item [] (i) For each CU, the optimal decision for allocating rates among CPs (decisions in the first layer) is an equal distribution, i.e., $b_k=b/2$, regardless of how rates are allocated within the eRRHs associated with each CP. Similarly, for each eRRH, the optimal decision for allocating rates among CPs is to distribute them equally, $r_k=r/2$. 
\item [](ii)  For each selected CP, when the rate of a CU sending a request for CP $k$ to eRRHs is relatively small, i.e., $b\leq w(r)$, $(b_k^{(1)}, b_k^{(2)})=(b/2, 0)$ is an optimal strategy for the CU. In this case, an optimal choice is to consolidate all rates to a single eRRH (decisions in the second layer).  
\item [](iii) For each selected CP, when the rate of a CU sending a request for CP $k$ to eRRHs is relatively large, i.e., $w(r)<b$, $(b_k^{(1)}, b_k^{(2)})=(b/2, b/2)$ is an optimal strategy. In this case, an optimal choice is to allocate the rate $b$ equally to both eRRHs (decisions in the second layer). 
\item [](iv) If we relax the condition \eqref{eq: at most one request} such that $\beta_{mk}^{(n)}(t)$ are independent across $n$, then using a simplified proof framework in Appendix~\ref{App: lem: optimal H h OS} and Appendix~\ref{App: optimal H1H2h OS}, we obtain the same conclusions as in   \cite{Cacheupdating, cacherecommendation, Timelycache}. That is, $(b_k^{(1)}, b_k^{(2)})=(b/2, 0)$, i.e., the optimal choice is to consolidate all rates to a single eRRH.  
\end{itemize}
In Theorem~\ref{thm: b r optimal}, deriving the closed form of  $w(r)$ is extremely challenging. 
However, since the sequence $\{h_{mk}(t)\}_t$ is asymptotically stationary, we can estimate $w(r)$ using $w(r; t)$ for small values of $t$. For example, as discussed in {\bf Step 2} of Appendix~\ref{App: lem: optimal H h OS}, setting $t=2$ yields the following approximation:
\begin{align*}
w(r)\approx \frac{2}{3 - r/2}.
\end{align*}
Based on this approximation, we characterize the optimal policy as follows:
\begin{itemize}
	\item [] (i) when $\frac{b}{2}$ is significantly larger than $\frac{2}{3 - r/2}$, an optimal policy is given by \eqref{eq: optimal OSR-1};
	\item [] (ii) when $\frac{b}{2}$ is significantly smaller than $\frac{2}{3 - r/2}$, an optimal policy is given by \eqref{eq: optimal OSR-2};
	\item [] (iii) when $\frac{b}{2}$ is approximately equal to $\frac{2}{3 - r/2}$, both policies in \eqref{eq: optimal OSR-1} and \eqref{eq: optimal OSR-2} yield similar average AoI for the CUs.
\end{itemize}

\section{Optimal NSR Policies}\label{sec: optimal NSR policies}

In this section, we investigate an optimal NSR policy in a network with $N=M=K=2$, following a process similar to the one described in Section~\ref{sec: optimal OSR policies}. We begin by calculating $L_e$ as defined in \eqref{eq: eRRHs average AoI} and derive the closed form of $L_e$ (see Theorem~\ref{thm: expectation of Xnk} in Section~\ref{subsubsec: Closed From of Je nonoblivious}) by introducing a  two-dimensional truncated Markov chain \big(see \eqref{eq: truncated G}\big). Next, we derive two upper bounds for $J_c$ (see Theorem~\ref{thm: upper bounds of Jc N2S} in Section~\ref{subsec: Upper bounds of Jc N2S}).  Finally, we theoretically determine an optimal strategy $\pi^*$ (see Theorem~\ref{thm: b r optimal NS}) in Section~\ref{subsec: optimal strategy pic in nonoblivious network}.

\subsection{Closed Form of $J_e$}\label{subsubsec: Closed From of Je nonoblivious}
\begin{figure}[htbp]
	\centering
	\includegraphics[height=4.75cm, width=5.75cm]{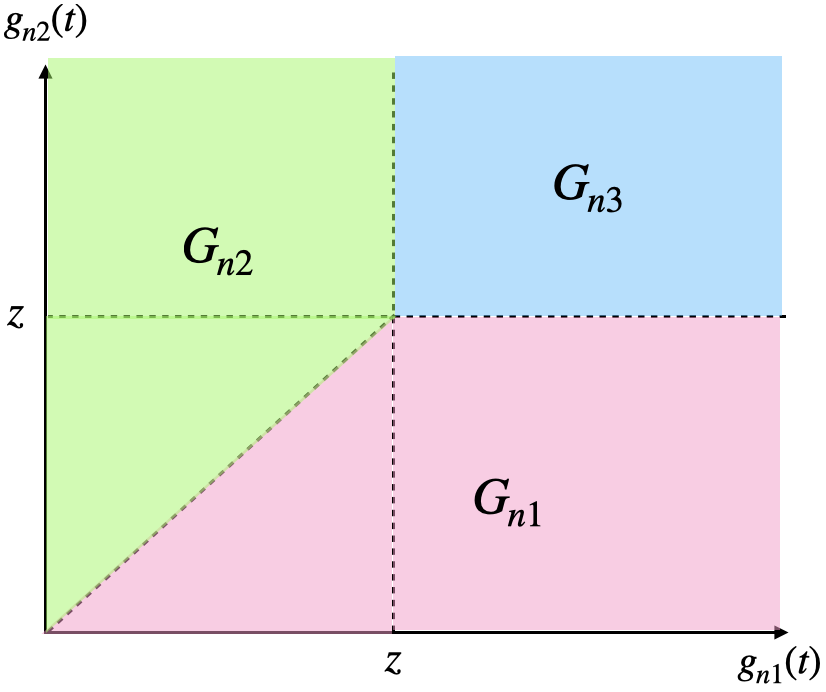}
	\caption{The partition of $\mathcal{G}_n$.}
	\label{Gpartition}
\end{figure} 
From \eqref{eq: nonoblivious r}, the commanding rate $\{r_{nk}(t)\}_{n, k}$ depend on the current ages of eRRHs, so the recursions of $g_{n1}(t)$ and $g_{n2}(t)$ are correlated with each other in every time slot $t$. Let $G_n(t) = (g_{n1}(t), g_{n2}(t))$, and denote its state space by $\mathcal{G}_n$. The space
$\mathcal{G}_n$ can be partitioned in the following $3$ disjoint subsets (see Fig.~\ref{Gpartition}),
\begin{align*}
\mathcal{G}_{n1}=&\Big\{G_n(t)|g_{n1}(t) \geq g_{n2}(t), g_{n2}(t)< z, t\geq 0\Big\},\\
\mathcal{G}_{n2}=&\Big\{G_n(t)|g_{n1}(t) < g_{n2}(t), g_{n1}(t)< z, t\geq 0\Big\},\\
\mathcal{G}_{n3}=&\Big\{G_n(t)|\{g_{n1}(t)\geq z, g_{n2}(t)\geq z, t\geq 0 \Big\}.
\end{align*}
For $n, k\in[2]$, according to \eqref{eq: indicator edge n sensor k}, the rate $r_{nk}(t)=\Pr\{\gamma_{nk}(t)=1\}$ depends on the values of $G_n(t)$ and varies over time as specified in \eqref{eq: nonoblivious r}:
\begin{align}
r_{nk}(t)=&1_{\{G_{n}(t)\in \mathcal{G}_{n1}\}}r_k +1_{\{G_{n}(t)\in \mathcal{G}_{n2}\}}(r-r_k)\nonumber \\
+&1_{\{G_{n}(t)\in \mathcal{G}_{n3}\}}r/2.\label{eq: taukt}
\end{align}
Then, the recursion of $g_{nk}(t)$ is given by
\begin{align}\label{eq: nonoblivious recursion g}
g_{nk}(t+1) = \left\{
\begin{aligned}
&1&& \text{w.p. }r_{nk}(t)\zeta_{nk}\\
&g_{nk}(t)+1&&\text{w.p. }1-r_{nk}(t)\zeta_{nk},
\end{aligned}
\right.
\end{align}
where $\zeta_{nk}$ is given in \eqref{eq: zeta}. 

Similar to \cite[Eqn.~(13)]{LAZY}, we define for any time slot $t$,
\begin{align}\label{eq: eRRH n CP k}
X_{nk}(t) = \min\{C_{nk}(l)|C_{nk}(l) > t\}-t,
\end{align}
where $C_{nk}(l)$ represents the timestamp when eRRH $n$ receives a new packet from CP $k$ for the $l$-th time, as defined in Section~\ref{subsubsec: Closed From of edge node}. In other words, $X_{nk}(t)$ in \eqref{eq: eRRH n CP k} denotes the time period from time slot $t$ to the next update. According to \cite[Eqn.~(15)]{LAZY}, 
\begin{align*}
\lim_{T\to\infty}\frac{1}{T}\sum_{t=1}^{T}\mathbb{E}[g_{nk}(t)] = \lim_{t\to\infty}\mathbb{E}[X_{nk}(t)]\quad\text{with probability }1.
\end{align*}
Thus, $J_e$ in \eqref{eq: eRRHs average AoI} can be re-written as
\begin{align*}
L_e=\frac{1}{4}\sum_{n\in[2]}\sum_{k\in[2]}\lim_{t\to\infty}\mathbb{E}[X_{nk}(t)]\quad\text{with probability }1.
\end{align*}
The remainder of this section focuses on calculating $\mathbb{E}[X_{nk}(t)]$.

Note that the state space $\mathcal{G}_n$ is infinite, 
to obtain the closed form of $L_e$, we introduce another vector $G_n^z(t)$, representing the $z$-truncated age state of eRRHs in time slot $t$ \cite{LAZY}:
\begin{align}\label{eq: truncated G}
G_n^z(t) = \big(g_{n1}^z(t), g_{n2}^z(t)\big),
\end{align} 
where $g_{nk}^z(t)$ is truncated at $z$. The evolution of $g_{nk}^z(t)$ is defined as follows
\begin{align*}
g_{nk}^z(t+1)  = \left\{
\begin{aligned}
&1&&\text{w.p. } r_{nk}(t)\zeta_{nk}\\
&\min\{g_{nk}^z(t)+1, z\}&&\text{w.p. } 1-r_{nk}(t)\zeta_{nk}.
\end{aligned}
\right.
\end{align*}
From \cite[Proposition~1]{LAZY} and  \cite[Proposition~2]{LAZY}, we know that $G_n^z(t)$, as defined in \eqref{eq: truncated G}, is a Markov process and possesses a unique steady-state distribution. Its state space is given by
\begin{align*}
\mathcal{G}_n^z=\{\phi|G_n^z(t)=\phi, t\geq0\},
\end{align*} 
with the cardinality $|\mathcal{G}_n^z| = z^2$. The steady-state distribution, denoted as $\pi = \{\pi_\phi\}_{\phi\in\mathcal{G}_n^z}$, is uniquely determined by
\begin{align}\label{eq: steady-state distribution phi}
\pi= {\bf P}\pi,
\end{align}
where ${\bf P}$ is the transition probability matrix. The specific form of  ${\bf P}$ is provided in \eqref{eq: transition probability matrix-1}, \eqref{eq: transition probability matrix-2}, \eqref{eq: transition probability matrix-3}, and \eqref{eq: transition probability matrix-4} in Appendix~\ref{App: steady-state distribution phi}. 

We proceed to derive the closed form of $L_e$ using $G_{n}^z(t)$ defined in \eqref{eq: truncated G}. Following a similar approach as in \cite[Eqn. (25)]{LAZY}, we express the expected age $\mathbb{E}[X_{nk}(t)]$ as
\begin{align}\label{eq: expectation X-1}
	\mathbb{E}[X_{nk}(t)] =& \sum_{\phi\in\mathcal{G}_n^z}\mathbb{E}[X_{nk}(t)|G_n^z(t)=\phi]\Pr(G_n^z(t)=\phi)\nonumber\\
	\triangleq&\sum_{\phi\in\mathcal{G}_n^z}F_{nk}(\phi)\pi_\phi,
\end{align}
where $F_{nk}(\phi) = \mathbb{E}[X_{nk}(t)|G_n^z(t)=\phi]$
 denotes the expected time until the next new packet from  CP $k$, given the AoI state $\phi = (\phi_1, \phi_2)$.
Define the set $\mathcal{G}_{n, down, k}$ as
\begin{align}\label{eq: Gup}
\mathcal{G}_{n, down, k} = \{\phi | \phi\in\mathcal{G}_{n}^z, \phi_k=1\},\,\, k\in[2].
\end{align}
Here, $F_{nk}(\phi)$ represents the expected first passage time from state $\phi$ to the set $\mathcal{G}_{n, down, k}$. The transition probabilities $P_{\phi\phi'}$ \big(see \eqref{eq: transition probability matrix-1}, \eqref{eq: transition probability matrix-2}, \eqref{eq: transition probability matrix-3}, and \eqref{eq: transition probability matrix-4} in Appendix~\ref{App: steady-state distribution phi}\big) indicate the probability of transitioning from state $\phi$ to state $\phi'$. Similar to the derivations in \cite[Eqns. (28) and (29)]{LAZY}, the mean first passage times starting from each state $\phi$, can be determined by using one-step analysis, i.e., by solving the system of equations,
\begin{align}
F_{nk}(\phi) = 1 + \sum_{\phi'\in\mathcal{G}_n^z\backslash\mathcal{G}_{n, down, k}}{\bf P}_{\phi\phi'}F_{nk}(\phi'),\,\,\phi\in\mathcal{G}_n^z.\label{eq: systems of Fnk-1}
\end{align}
By leveraging \eqref{eq: steady-state distribution phi}, \eqref{eq: expectation X-1}, and \eqref{eq: systems of Fnk-1}, we directly establish the following theorem.
\begin{theorem}\label{thm: expectation of Xnk}
The closed-form expression for the average AoI of eRRHs is provided by 
\begin{align}\label{eq: closed form Je N2S}
L_e=\frac{1}{4}\sum_{n\in[2]}\sum_{k\in[2]}\sum_{\phi\in\mathcal{G}_{n}^z}F_{nk}(\phi)\pi_\phi,
\end{align}
where $\pi_\phi$ is provided by \eqref{eq: steady-state distribution phi} and $F_{nk}(\phi)$ is provided by \eqref{eq: systems of Fnk-1}.
\end{theorem}
In Theorem~\ref{thm: expectation of Xnk}, we adopt an approach similar to  \cite{LAZY} to derive the closed-form expression of $L_e$.  However, while \cite{LAZY} employs a one-dimensional truncated Markov chain, our analysis requires the construction of a two-dimensional truncated Markov chain to accommodate the more complex structure of strategies.

\subsection{Upper Bounds of $J_c$}\label{subsec: Upper bounds of Jc N2S}
From \eqref{eq: recursion of user AoI} and \eqref{eq: taukt}, under a NSR policy, $h_{mk}(t)$ has the following recursion, 
\begin{equation}\label{eq: recursion of h-2}
\begin{aligned}
h_{mk}(t+1)= \left\{
\begin{aligned}
&1,\quad \sum_{n\in[2]}b_k^{(n)}r_{nk}(t)\\
&\min\{h_{mk}(t), g_{1k}(t)\}+1,\,\, b_{k}^{(1)}(1-r_{1k}(t))\\
&\min\{h_{mk}(t), g_{2k}(t)\}+1,\,\, b_{k}^{(2)}(1-r_{2k}(t))\\
&h_{mk}(t)+1,\quad1-b_k, 
\end{aligned}
\right.
\end{aligned}
\end{equation}
where $g_{1k}(t)$ and $g_{2k}(t)$ are given in \eqref{eq: nonoblivious recursion g}. Similar to Section~\ref{subsec: Upper bounds of Jc O2S}, we provide two upper bounds for $J_c$. 
\begin{theorem}\label{thm: upper bounds of Jc N2S}
	$J_c$, defined in \eqref{eq: CUs average AoI}, has two upper bounds,
	\begin{align}
		J_c\leq&\frac{1}{2}\sum_{k\in[2]}\frac{1}{b_kr_2}\label{eq: Jc N2S upper bound 1},\\
		J_c\leq&\frac{1}{2}\sum_{k\in[2]}\Big(\frac{1-b_k+b_kr_2}{b_kr_2}\nonumber\\
		&+(1-r_2)\sum_{n\in[2]}b_k^{(n)}\sum_{\phi\in\mathcal{G}_n^z}F_{nk}(\phi)\pi_\phi\Big)\label{eq: Jc N2S upper bound 2}.
	\end{align}
\end{theorem}
\begin{proof}
The proof is given in Appendix~\ref{App: thm: upper bounds of Jc N2S}.
\end{proof}
Both upper bounds \eqref{eq: Jc N2S upper bound 1} and \eqref{eq: Jc N2S upper bound 2} represent extreme scenarios in which the eRRHs always command the CPs to update with new packets with probability $r_2$. Specifically, in the case associated with \eqref{eq: Jc N2S upper bound 1}, the CUs exclusively deliver new generated packets, without transmitting previously cached ones. In contrast, under \eqref{eq: Jc N2S upper bound 2}, the CUs replace the most recently received packets with the ones currently delivered by the eRRHs even if the latter are older. 
	
Under non-oblivious policies, the opposite phenomenon to that described in Theorem~\ref{thm: upper bounds of Jc} is observed \big(see Fig.~\ref{network2_ub}(b) in Section~\ref{sec: numerical results}\big): the upper bound given in \eqref{eq: Jc N2S upper bound 2} is more accurate. This is because, in each time slot, eRRHs prioritize commanding the CP with the higher AoI, thereby reducing the likelihood that the delivered content is older than the CUs' previously received content. As a result, in the scenario corresponding to \eqref{eq: Jc N2S upper bound 2}, the CUs are more likely to receive fresher content consistently.

\subsection{Optimal NSR Strategy}\label{subsec: optimal strategy pic in nonoblivious network}
Following the approach outlined in Section~\ref{subsec: optimal strategy pic in oblivious network}, we derive an optimal policy $\pi^*$ for $J_c$ in \eqref{eq: CUs average AoI} via the following two steps: (i) Fix $b_k$ and $r_k$, determine an optimal $\big((b_k^{(1)})^*, (b_k^{(2)})^*\big)$ that minimize $$\lim_{T\to\infty}\frac{1}{T}\sum_{t=1}^{T}\mathbb{E}[h_{mk}(t)];$$ (ii) Find an optimal policy $\pi^*$ to minimize $$J_c=\frac{1}{4}\sum_{m\in[2]}\sum_{k\in[2]}\lim_{T\to\infty}\frac{1}{T}\sum_{t=1}^{T}\mathbb{E}[h_{mk}(t)].$$
Similar to \eqref{eq: upper bound H1} and \eqref{eq: upper bound H2}, we define two stochastic processes,
\begin{align}\label{eq: upper bound H1 NS}
H_{mk, 1}'(t+1) = \left\{
\begin{aligned}
&1&& \sum_{n\in[2]}b_k^{(n)}r_{nk}(t)\\
&h_{mk}(t)+1&&1-\sum_{n\in[2]}b_k^{(n)}r_{nk}(t)
\end{aligned}
\right.
\end{align}
and
\begin{align}\label{eq: upper bound H2 NS}
	H_{mk, 2}'(t+1) = \left\{
	\begin{aligned}
		&1&&\sum_{n\in[2]}b_k^{(n)}r_{nk}(t)\\
		&g_{1k}(t)+1&& b_k^{(1)}\big(1-r_{1k}(t)\big)\\
		&g_{2k}(t)+1&&b_k^{(2)}\big(1-r_{2k}(t)\big)\\
		&h_{mk}(t)+1&&1-b_k,
	\end{aligned}
	\right.
\end{align}
where $g_{1k}(t)$ and $g_{2k}(t)$ are given in \eqref{eq: nonoblivious recursion g}. By convention, let $H_{mk, 1}'(0) = H_{mk, 2}'(0)=1$. From \eqref{eq: recursion of h-2}, \eqref{eq: upper bound H1 NS}, and \eqref{eq: upper bound H2 NS}, we have
\begin{align}\label{eq: hH1H2 NS}
h_{mk}(t) \overset{d}{=} \min\{H_{mk, 1}'(t), H_{mk, 2}'(t)\},
\end{align}
where $\overset{d}{=}$ represents equality {\it in distribution}. By utilizing \eqref{eq: hH1H2 NS}, we derive the following analytical results.

\begin{lemma}\label{lem: optimal H h NS}
Fix $b_k, r_k$ with $k\in[2]$. For each $t\geq 2$, there exists  a threshold depending on $r_k$ , denoted as  $w'(r_k, r_{3-k}; t)$, such that, 
\begin{itemize}
	\item [] (i) when $b_k>w'(r_k, r_{3-k}; t)$, $(b_1^{(1)}, b_2^{(1)})=(b_1/2, b_2/2)$ is the global minimizer of $\mathbb{E}[h_{mk}(t)]$; 
	\item [] (ii) when $b_k < w'(r_k, r_{3-k}; t)$, either $(b_1^{(1)}, b_2^{(1)})=(b_1, b_2)$ or $(0, 0)$ is the global minimizer of $\mathbb{E}[h_{mk}(t)]$;
	\item [] (iii) when $b_k = w'(r_k, r_{3-k}; t)$, either $(b_1^{(1)}, b_2^{(1)})=(b_1, b_2)$, $(0, 0)$, or $(b_1/2, b_2/2)$ is the global minimizer of $\mathbb{E}[h_{mk}(t)]$;
	\item [] (iv) $w'(r_k, r_{3-k}; t)$ decreases with $t$ and converges as $t \to \infty$.
\end{itemize}
\end{lemma}
\begin{proof}
The proof follows a similar approach to that of Lemma~\ref{lem: optimal H h OS}, with a key difference: $r_{nk}(t)$ depends on the values of $\{g_{nk}(t)\}_{k\in[2]}$, making the optimization of $b_k^{(1)}$ and $b_k^{(2)}$ interdependent.  As a result,  $b_k^{(1)}$ and $b_k^{(2)}$ cannot be treated independently; they must be considered together as an integral. The structure of the proof is outlined as follows: First, we demonstrate that $\mathbb{E}[h_{mk}(t)]$ is symmetric with respect to $(b_1^{(1)}, b_2^{(1)})=(b_1/2, b_2/2)$, and $\mathbb{E}[h_{mk}(t)]$ is continuous in $b_1^{(1)}$ and $b_2^{(1)}$, which implies that the minimum attained at $(b_1, b_2)$, $(0, 0)$, or $(b_1/2, b_2/2)$. Next, we prove the lemma by mathematical induction. Specifically, since $h_{mk}(t)$ becomes asymptotically stationary under stationary randomized policies, $h_{mk}(t+1) = h_{mk}(t)$ almost surely as $t\to\infty$, which implies that $\lim_{t\to\infty}w'(r_k, r_{3-k}; t)$ converges. The detailed proof is provided in Appendix~\ref{App: lem: optimal H h NS}.
\end{proof}
\begin{remark}
Let $w'(r_k, r_{3-k}) = \lim_{t\to\infty}w'(r_k, r_{3-k}; t)$. One may wonder what happens in a mixed case -- for example, when $b_1>w'(r_1, r_2)$ but $b_2\leq w'(r_2, r_1)$.  From the condition $b_1>w'(r_1, r_2)$, one may infer that an optimal choice is $(b_1^{(1)}, b_2^{(1)})=(b_1/2, b_2/2)$. However, based on  $b_2\leq w'(r_2, r_1)$, an optimal decision would instead suggest  $(b_1^{(1)}, b_2^{(1)})=(b_1, b_2)$ or $(0, 0)$. In Lemma~\ref{lem: optimal H1H2h NS}, we will show that such a mixed case does {\it not} exist. In other words, since $b_k^{(1)}$ and $b_k^{(2)}$ cannot be treated independently, the thresholds $w'(r_1, r_2)$ and $w'(r_2, r_1)$ must also be coordinated rather than treated independently. 
\end{remark}

\begin{lemma}\label{lem: optimal H1H2h NS}
Let $b, r$ be fixed and $(b, r)\in\mathcal{F}$: (i) if $r\leq 1$, $(b_1, b_2) = (b/2, b/2)$ and $(r_1, r_2)=(r, 0)$ is the global minimizer of $\sum_{k\in[2]}\mathbb{E}[h_{mk}(t)]$ for $t\geq 1$; (ii) if $r > 1$, $(b_1, b_2) = (b/2, b/2)$ and $(r_1, r_2)=(1, r-1)$ is the global minimizer of $\sum_{k\in[2]}\mathbb{E}[h_{mk}(t)]$ for $t\geq 1$.
\end{lemma}
\begin{proof}
The roadmap of the proof is similar to that of Lemma~\ref{lem: optimal H1H2h OS}, and is given as follows. We begin by verifying that if both CUs have the optimal allocation $\big((b_1^{(1)})^*,  b_2^{(1)})^*\big) = (b_1, b_2)$ or $(0, 0)$, the minimum choice for $(b_1, r_1)$ is $(b/2, r/2)$. Next, we show that if both CUs choose the optimal split $\big((b_1^{(1)})^*,  b_2^{(1)})^*\big) = (b_1/2, b_2/2)$, $(b/2, r/2)$ again yields the minimum. Finally, we demonstrate that the mixed case--where one CU chooses $(b_1, b_2)$ or $(0, 0)$, and the other chooses $(b_1/2, b_2/2)$--cannot occur. The full proof is provided in Appendix~\ref{App: optimal H1H2h NS}.
\end{proof}

\begin{theorem}\label{thm: b r optimal NS}
Let $b, r$ be fixed, $(b, r)\in\mathcal{F}$, and $ m,n,k\in[2]$. Let
\begin{align}\label{eq: optimal r}
(r_1, r_2)=\left\{
\begin{aligned}
	&(r, 0),\quad r\leq 1\\
	&(1, r-1),\quad r>1.
\end{aligned}
\right.
\end{align}
There exists a threshold $w'(\min\{r, 1\}, r-\min\{r, 1\})$, such that if $b>w'(\min\{r, 1\}, r-\min\{r, 1\})$, the optimal policy $\pi^*$ is given by
\begin{align}\label{eq: optimal NSR-1}
\pi^* = \Big\{
b_{mk}^{(n)}=b/4, (r_1, r_2)\text{ is given in \eqref{eq: optimal r}}
\Big\};
\end{align}
if $b \leq w'(\min\{r, 1\}, r-\min\{r, 1\})$, an optimal policy
\begin{align}\label{eq: optimal NSR-2}
\pi^* = 
\Big\{
(b_{mk}^{(n)}, b_{mk}^{(3-n)}) = (b/2, 0), (r_1, r_2)\text{ is given in \eqref{eq: optimal r}}\Big\}.
\end{align}
\end{theorem}
\begin{proof}
From Lemma~\ref{lem: optimal H h NS}~(iv), $\lim_{t\to\infty}w'(r_k, r_{3-k}; t)$ exists, denoted as $w'(r_k, r_{3-k})$. 
Next, the proof follows directly from  Lemma~\ref{lem: optimal H h NS}~(i), (ii), (iii), and Lemma~\ref{lem: optimal H1H2h NS}. 
\end{proof}
\begin{remark}
If $b  < w'(\min\{r, 1\}, r-\min\{r, 1\})$, an alternative optimal policy is 
$\pi^* = \{(b_{mk}^{(n)}, b_{mk}^{(3-n)}) = (0, b/2),  (r_1, r_2)\text{ is given in \eqref{eq: optimal r}}\}$.
If $b = w'(\min\{r, 1\}, r-\min\{r, 1\})$, then two alternative optimal policies exist:
$\pi^* = \big\{(b_{mk}^{(n)}, b_{mk}^{(3-n)}) = (0, b/2), (r_1, r_2)\text{ is given in \eqref{eq: optimal r}}\big\}$
and  $\pi^* = \big\{b_{mk}^{(n)}=b/4, (r_1, r_2)\text{ is given in \eqref{eq: optimal r}}\big\}$.
\end{remark}

The optimal non-oblivious policies share similarities to optimal oblivious policies, but also exhibit key differences. Based on Theorem~\ref{thm: b r optimal NS}, we can draw the following insights regarding the optimal strategies for both CUs and eRRHs:
\begin{itemize} 
\item[] (i) CU Policies: The optimal strategies for each CU share similarities with those under oblivious policies. Specifically, in the first decision layer—allocating rates among CPs—the optimal policy is to distribute the rate equally, i.e., $b_k = b/2$, regardless of how the rate is further allocated among the eRRHs associated with each CP. In the second decision layer—allocating the CU’s request rate to eRRHs for a selected CP—the strategy depends on the magnitude of $b$. When the rate is relatively small, i.e., $b \leq w'(\min\{r, 1\}, r-\min\{r, 1\})$, the optimal choice is to consolidate the entire rate to a single eRRH. In contrast, when the rate is relatively large, i.e., $w'(\min\{r, 1\}, r-\min\{r, 1\}) < b$, it is optimal to allocate the rate equally between the two eRRHs.
\item[] (ii) eRRH Policies: In contrast to the optimal OSR policy, which distributes the rates $(r_1, r_2)$ equally, the optimal policy here changes to: $(r_1, r_2) = (r, 0)$ when $r\leq 1$ and $(r_1, r_2) = (1, r-1)$ when $r>1$. Under the optimal policy, priority is given to transmit commands to the CP with the higher AoI. This consolidation is optimal because the urgency of delivering commands varies with the CPs’ current ages, and prioritizing the most outdated CP ensures the most efficient use of the available rate. 
\end{itemize}

\section{Numerical Results}\label{sec: numerical results}
In this section, we verify our findings through simulations, focusing on networks with $M=N=K=2$. For non-oblivious policies, we let pre-determined threshold $z=3$. 

\begin{figure}[htbp]
	\centering
	\includegraphics[height=4.5cm, width=6cm]{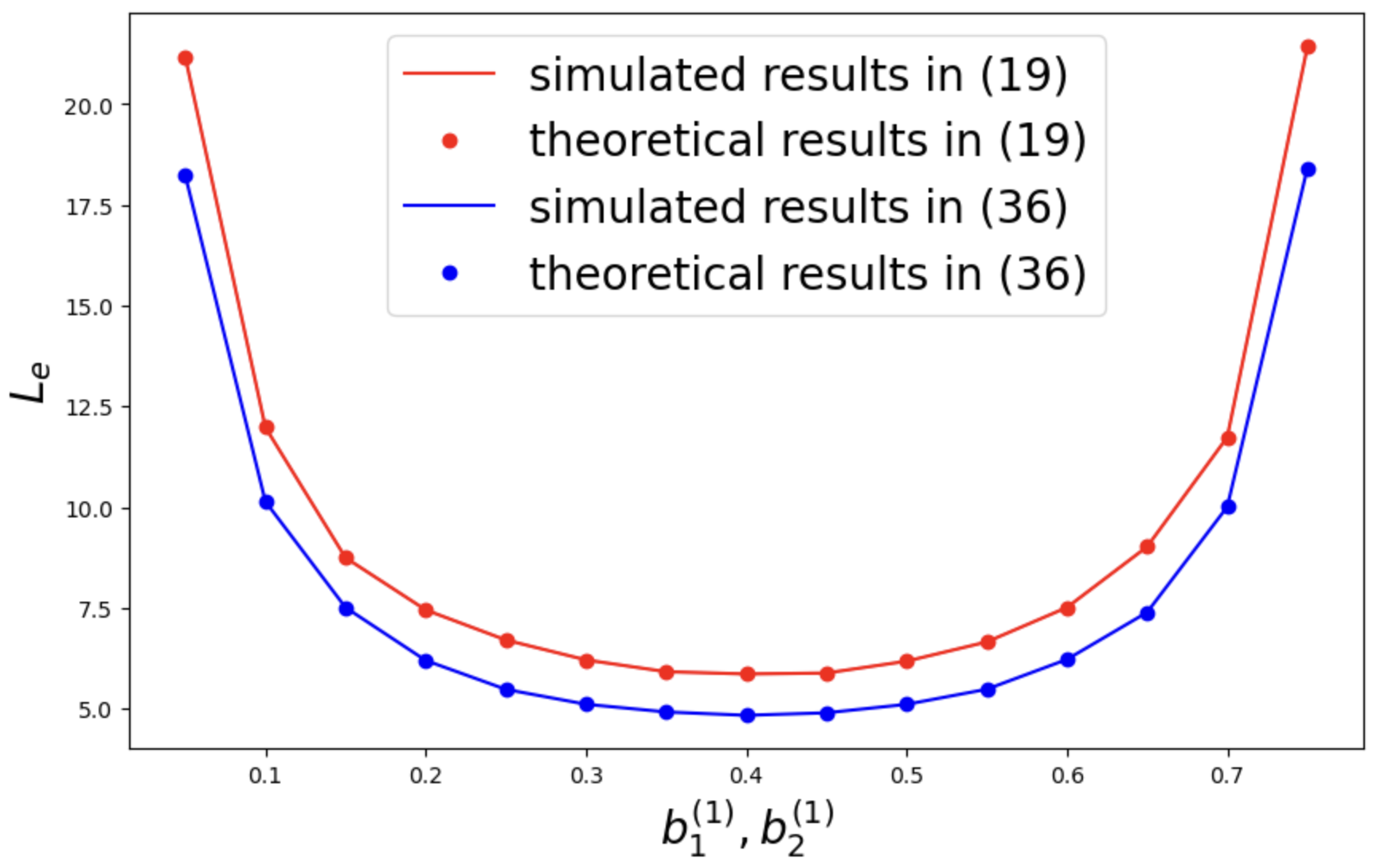}
	\caption{The theoretical and numerical $L_e$ under both oblivious and non-oblivious policies.}
	\label{Je}
\end{figure} 
First of all, we verify our theoretical results for $L_e$ under both oblivious (see Theorem~\ref{thm: closed form Je}) and non-oblivious (see Theorem~\ref{thm: expectation of Xnk}) policies. In Fig.~\ref{Je},  we let $r_1=0.4$, $r_2=0.2$, and $b_1=b_2=0.8$, while varying $b_1^{(1)}=b_2^{(1)}$ within the range $(0, 0.8)$. The closed-form expressions for $L_e$ align perfectly with the simulation results, confirming the accuracy of Theorem~\ref{thm: closed form Je} and Theorem~\ref{thm: expectation of Xnk}. 
\begin{figure}[htbp]
	\centering
	\begin{minipage}[b]{0.241\textwidth}
		\centering
		\includegraphics[width=\textwidth]{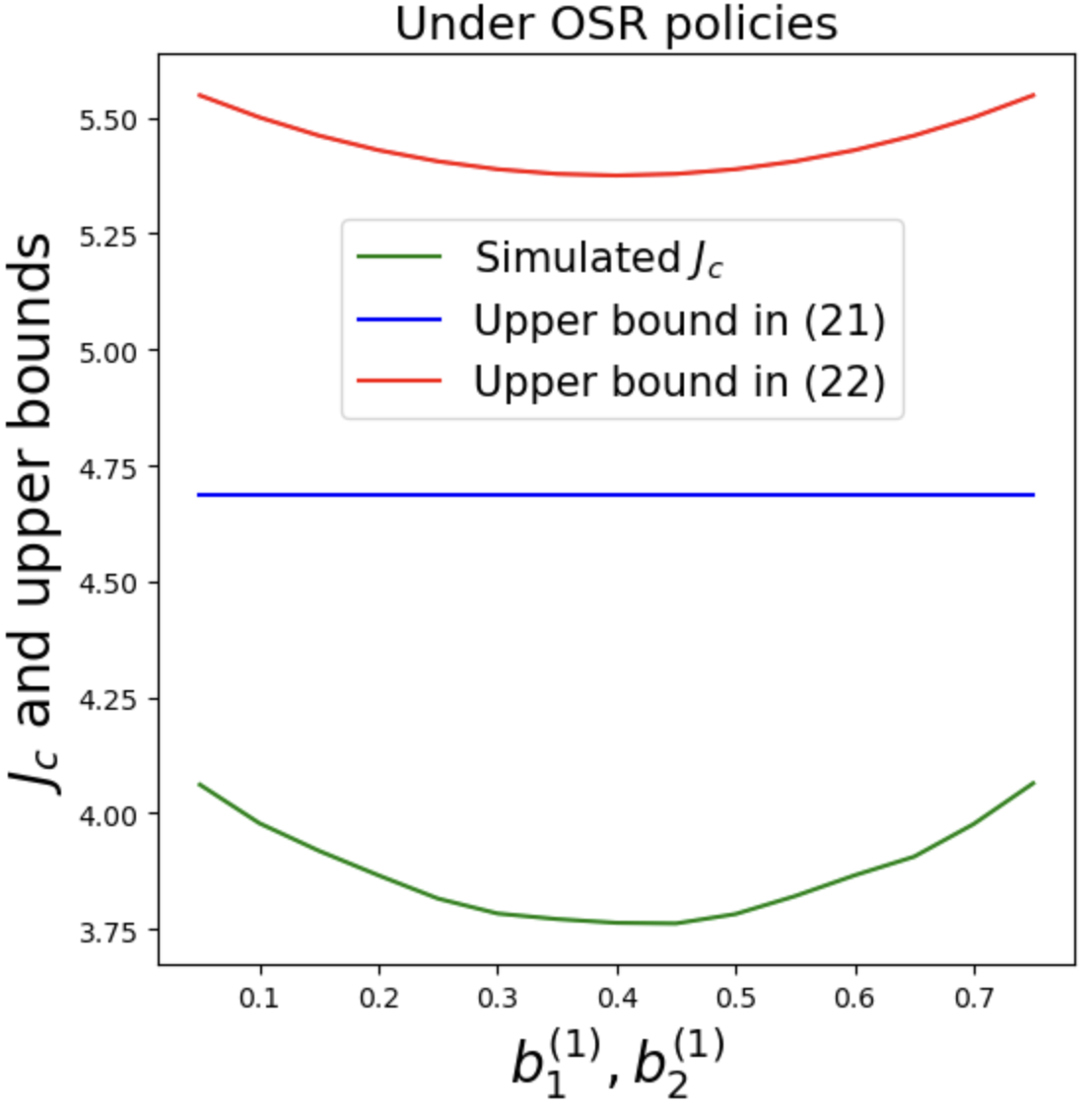}
		\caption*{(a) }
	\end{minipage}
	\begin{minipage}[b]{0.241\textwidth}
		\centering
		\includegraphics[width=\textwidth]{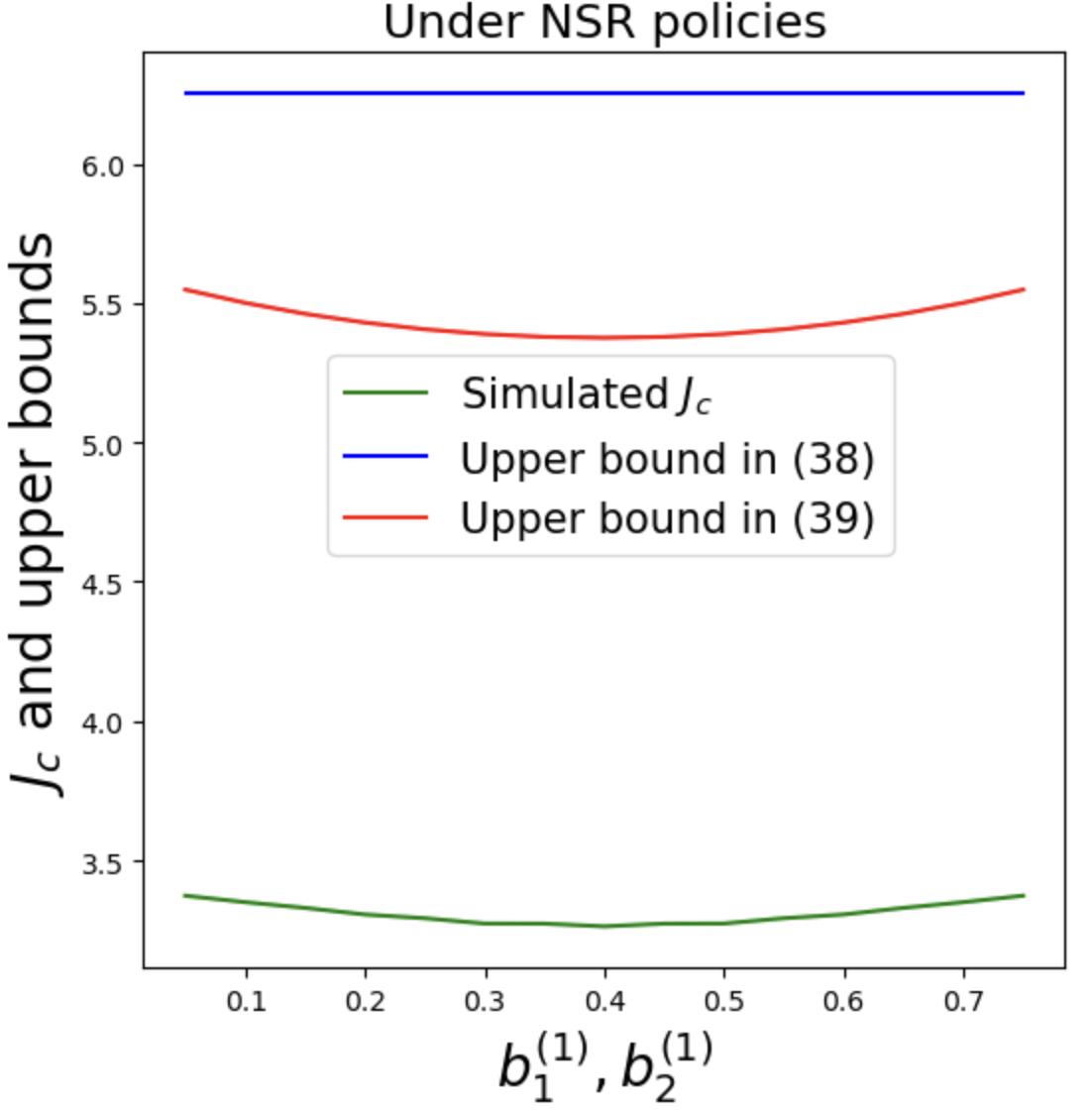}
		\caption*{(b)}
	\end{minipage}
	\caption{$J_c$ and upper bounds under both oblivious and non-oblivious policies.}
	\label{network2_ub}
\end{figure}

In Fig.~\ref{network2_ub}, we verify our theoretical results for upper bounds (see Theorem~\ref{thm: upper bounds of Jc} and Theorem~\ref{thm: upper bounds of Jc N2S}). For both cases, we set $r_1=0.4$, $r_2=0.2$, and $b_1^{(1)}=b_2^{(1)}$ varies within the range $[0, 0.8]$. We observe that \eqref{eq: Jc O2S upper bound 1}, \eqref{eq: Jc O2S upper bound 2}, \eqref{eq: Jc N2S upper bound 1}, and \eqref{eq: Jc N2S upper bound 2} indeed serve as upper bounds. 
Under oblivious policies, the upper bound in \eqref{eq: Jc O2S upper bound 1} offers a more accurate estimate compared to the one in \eqref{eq: Jc O2S upper bound 2}. This is because \eqref{eq: Jc O2S upper bound 1} corresponds to an extreme scenario where eRRHs only send new packets and never retransmit previously cached ones. In contrast, the upper bound in \eqref{eq: Jc O2S upper bound 2} reflects another extreme case where the CUs overwrite their most recently received packets with those currently delivered by the eRRHs, even if the new packets are older. Since the commanding rates for CPs remain constant over time, the first scenario results in a higher rate of new content reception, making its corresponding bound a tighter and more representative estimate under oblivious behavior.

Under non-oblivious policies, we observe the opposite: the upper bound  in \eqref{eq: Jc N2S upper bound 2} turns out to be more accurate. While under the upper bound \eqref{eq: Jc N2S upper bound 1} the CUs transmit only new packets and disregard previously cached ones, the scenario in \eqref{eq: Jc N2S upper bound 2} allows CUs to replace their most recently received packets with the content currently delivered by the eRRHs—even if it is older. However, since the eRRHs prioritizes the CP that has a higher age in each time slot, the likelihood that they deliver stale content is significantly reduced. Consequently, in the setting of \eqref{eq: Jc N2S upper bound 2}, the CUs are more likely to consistently receive fresher content.

\begin{figure}[htbp]
\centering
\includegraphics[height=4.8cm, width=8.4cm]{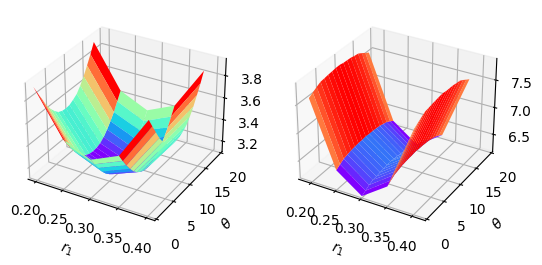}
\caption{Performance of $J_c$ under oblivious policies for $r=0.6$, with $b=1.8$ (left) and $b=0.8$ (right).}
\label{network2_O_1}
\end{figure} 

\begin{figure}[htbp]
	\centering
	\includegraphics[height=4.6cm, width=8.5cm]{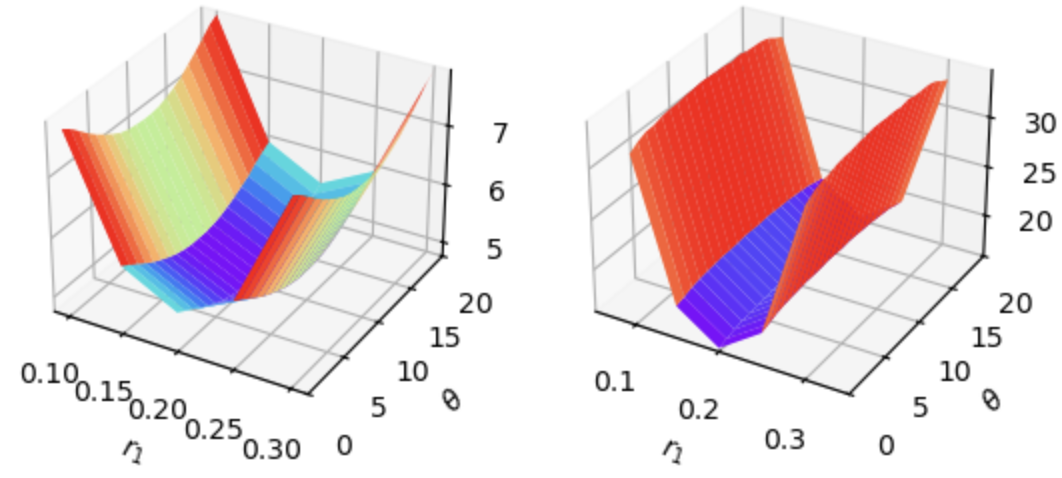}
	\caption{Performance of $J_c$ under oblivious policies for $r=0.4$, with $b=1.6$ (left) and $b=0.4$ (right).}
	\label{network2_O_2}
\end{figure}

We now verify the optimal oblivious and non-oblivious policies  stated in Theorem~\ref{thm: b r optimal} and Theorem~\ref{thm: b r optimal NS}, using the results illustrated in Fig.~\ref{network2_O_1} through Fig.~\ref{network2_N_2}, under different pairs of parameters $(b, r)$. In these figures: (i) The $x$-axis represents $r_1$. (ii) The $y$-axis shows the index $\theta\in [20]$, which determines the value of $b_k^{(1)}$ for $k\in[2]$. Specifically, $b_k^{(1)} = \frac{b_k}{20}\cdot\theta$. (iii) The $z$-axis represents the value of $J_c$. 
This visualization provides a comprehensive view of how $J_c$ varies with different combinations of $(r_1, r_2)$ and $(b_k^{(1)}, b_k^{(2)})$ with $k\in[2]$. 

The performance of $J_c$ under oblivious policies is illustrated in  Fig.~\ref{network2_O_1} and Fig.~\ref{network2_O_2}.  In Fig.~\ref{network2_O_1}, we fix the total rate $r=r_1+r_2 = 0.6$, and let  $(r_1, r_2)$ to vary within $[0.2, 0.4]$.  On the left side of Fig.~\ref{network2_O_1}, $(b_1, b_2)$ varies within the range $[0.8, 1]$ with $b=b_1+b_2=1.8$; on the right side, $(b_1, b_2)$ varies within the range $[0.3, 0.5]$ with $b=b_1+b_2=0.8$.  Similarly,  in Fig.~\ref{network2_O_2}, we set $r=0.4$ with  $(r_1, r_2)$ varying between $[0.1, 0.3]$.  The left side of Fig.~\ref{network2_O_2} corresponds to $(b_1, b_2)$ varying within the range $[0.6, 1]$ with $b=b_1+b_2=1.8$, while the right side considers $(b_1, b_2)$ varying within the range $[0.1, 0.3]$ with $b=b_1+b_2=0.4$.

In the left subplots of both Fig.~\ref{network2_O_1} ($b=1.8$) and Fig.~\ref{network2_O_2} $(b=1.6)$, 
the total request probability is relatively high. In this setting, $J_c$ attains its minimum at $(b_1^{(1)}, b_1^{(2)}, b_2^{(1)}, b_2^{(2)}) = (b/4, b/4, b/4, b/4)$ and $r_1=r_2=r/2$, which aligns with the optimal solution described in \eqref{eq: optimal OSR-1} of Theorem~\ref{thm: b r optimal}. In contrast, on the right sides of  Fig.~\ref{network2_O_1} ($b=0.8$) and Fig.~\ref{network2_O_2} $(b=0.4)$, where the total request probability for CPs is  relatively low,  $J_c$ is minimized at $(b_1^{(1)}, b_1^{(2)}, b_2^{(1)}, b_2^{(2)}) = (b/2, 0, b/2, 0)$ and $r_1=r_2=r/2$,  consistent with \eqref{eq: optimal OSR-2} in  Theorem~\ref{thm: b r optimal}.  

Unlike previous studies \cite{Freshcaching, Cacheupdating, cacherecommendation, timelyproactivecache, Timelycache}, where the optimal choice is always to consolidate all rates to a single eRRH (or source), our findings include a different choice. In those precedents, decisions within an eRRH are independent of each other, leading to straightforward rate consolidation.  However, in our model, the interdependence of decisions within an eRRH results in a different optimal choice. This interdependence decreases the average AoI for CUs by adopting this new optimal choice, resulting in better performance compared to the scenarios considered in the aforementioned studies.

\begin{figure}[htbp]
	\centering
	\includegraphics[height=5cm, width=8.5cm]{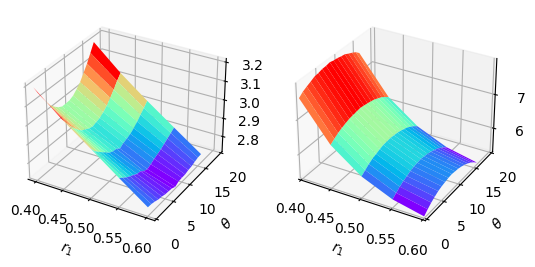}
	\caption{Performance of $J_c$ under non-oblivious policies for $r=0.6$, with $b=1.8$ (left) and $b=0.8$ (right).}
	\label{network2_N_1}
\end{figure} 

\begin{figure}[htbp]
	\centering
	\includegraphics[height=4.6cm, width=8.5cm]{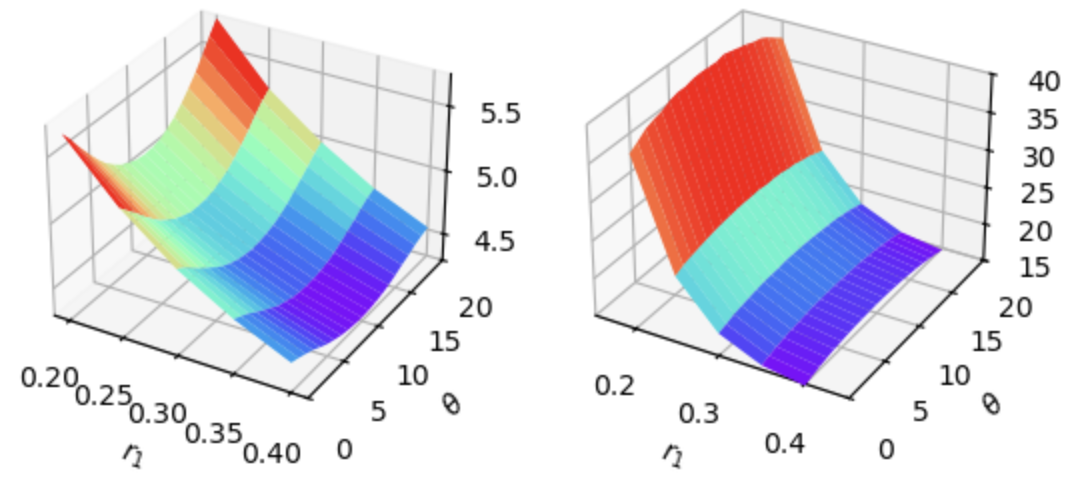}
	\caption{Performance of $J_c$ under non-oblivious policies for $r=0.4$, with $b=1.6$ (left) and $b=0.4$ (right).}
	\label{network2_N_2}
\end{figure} 

The performance of $J_c$ under non-oblivious policies is shown in  Fig.~\ref{network2_N_1} and Fig.~\ref{network2_N_2}.  The experimental setup in Fig.~\ref{network2_N_1} matches that of Fig.~\ref{network2_O_1}:  we fix $r=0.6$, and let  $(r_1, r_2)$ to vary within $[0.2, 0.4]$.  On the left, $(b_1, b_2)$ varies within the range $[0.8, 1]$ with $b=1.8$; on the right, $(b_1, b_2)$ varies within the range $[0.3, 0.5]$ with $b=0.8$.  Similarly,  Fig.~\ref{network2_N_2} uses the same setup as Fig.~\ref{network2_O_2}: we set $r=0.4$ with  $(r_1, r_2)$ varying between $[0.1, 0.3]$.  The left panel corresponds to $(b_1, b_2)$ varying within the range $[0.6, 1]$ with $b=b_1+b_2=1.8$, while the right panel uses $(b_1, b_2)$ varying within the range $[0.1, 0.3]$ with $b=b_1+b_2=0.4$.

In the left subplots of both Fig.~\ref{network2_N_1} ($b=1.8$) and Fig.~\ref{network2_N_2} $(b=1.6)$, 
the total request probability of requests for CPs is relatively high. Under this setting, $J_c$ achieves its minimum at $(b_1^{(1)}, b_1^{(2)}, b_2^{(1)}, b_2^{(2)}) = (b/4, b/4, b/4, b/4)$ and $(r_1, r_2)=(r, 0)$, which aligns with the optimal solution described in \eqref{eq: optimal NSR-1} of Theorem~\ref{thm: b r optimal NS}. In contrast, on the right sides of  Fig.~\ref{network2_N_1} ($b=0.8$) and Fig.~\ref{network2_N_2} $(b=0.4)$, where the total request probability of requests for CPs is  relatively low,  $J_c$ is minimized at $(b_1^{(1)}, b_1^{(2)}, b_2^{(1)}, b_2^{(2)}) = (b/2, 0, b/2, 0)$ and $(r_1, r_2)=(r, 0)$,  consistent with \eqref{eq: optimal NSR-2} in  Theorem~\ref{thm: b r optimal NS}.  

By comparing the numerical results in Figs.~\ref{network2_O_1},~\ref{network2_O_2} and Figs.~\ref{network2_N_1},~\ref{network2_N_2}, we observe that the average AoI under an optimal NSR strategy is smaller than under an optimal OSR strategy. Specifically: (i) For $(b, r)=(0.8, 0.6)$, $J_c=3.54$ under optimal OSR strategy, and $J_c=2.73$ under optimal NSR strategy. (ii) For $(b, r)=(1.8, 0.6)$, $J_c=6.21$ under optimal OSR strategy, and $J_c=5.34$ under optimal NSR strategy.  (iii) For $(b, r)=(1.6, 0.4)$, $J_c=4.82$ under optimal OSR strategy, and $J_c=4.32$ under optimal NSR strategy.  (iv) For $(b, r)=(0.4, 0.4)$, $J_c=16.11$ under optimal OSR strategy, and $J_c=14.74$ under optimal NSR strategy.
The performance gain stems from the fact that non-oblivious policies enable eRRHs to utilize historical AoI information, allowing them to make more informed and adaptive decisions. By leveraging this historical data, eRRHs can prioritize updates from CPs with the highest current AoI, thereby improving the freshness of information more effectively. In contrast, oblivious policies operate without knowledge of past AoI values and thus lack this adaptive capability. As a result, the optimal non-oblivious policy demonstrates superior performance in minimizing the average AoI.

\section{Conclusions and Future Research}\label{sec: Conclusions and Future Directions}

In this work, we analyzed the AoI performance in a decentralized F-RAN consisting $M$ CUs, $N$ eRRHs, and $K$ CPs, under two general classes of policies: oblivious and non-oblivious. We first focused on a baseline configuration with $M=N=K=2$, which enables analytical tractability while capturing essential system behaviors.
For this setting, we began by deriving closed-form expressions for the average AoI of eRRHs under both policy types. Given the analytical intractability of the average AoI at CUs, we provided two general upper bounds instead. We then identified the optimal policies for both eRRHs and CUs. Building upon the results obtained for the baseline configuration, we subsequently extended our framework to accommodate two general scenarios. Finally, our numerical results validate the theoretical findings.

Future research directions include the following generalizations: 1) Dynamic rates: the probabilities of actions change over time. Addressing this requires the integration of powerful machine learning tools, such as multi-armed bandit algorithms, to adaptively optimize the AoI performance. 2) Heterogeneous networks: Here, CUs, eRRHs, and CPs are no longer statistically identical within their respective groups. A promising approach is to categorize CUs, eRRHs, and CPs into several groups based on their characteristics. We can then apply a similar analytical framework to each group, aggregating the results to identify optimal policies across the heterogeneous network.

\bibliographystyle{unsrt}
\bibliography{references.bib}

\vspace{-3 em}
\begin{IEEEbiography}[{\includegraphics[width=1in,height=1.25in,clip,keepaspectratio]{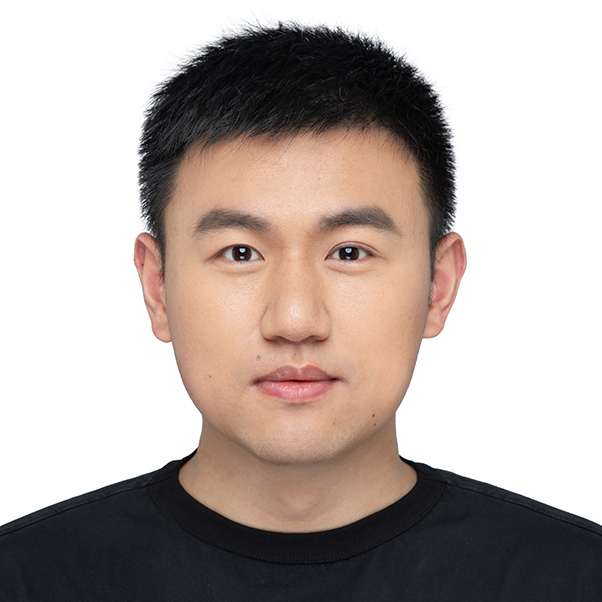}}]{Xingran Chen} received the M.A. in Applied Mathematics and Computational Science and the Ph.D. degree in Electrical and Systems Engineering from the University of Pennsylvania. He is currently a postdoctoral researcher at Rutgers University and an Assistant Professor (currently on leave) in the School of Information and Communication Engineering at University of Electronic Science and Technology of China. His research focuses on the theoretical foundations and collaborative learning algorithms for decentralized systems. He received the IEEE Communications Society \& Information Theory Society Joint Paper Award in 2023. He served as a guest editor for China Communications in 2024, and Entropy in 2025.
\end{IEEEbiography}

\vspace{-3 em}

\begin{IEEEbiography}[{\includegraphics[width=1in,height=1.25in,clip,keepaspectratio]{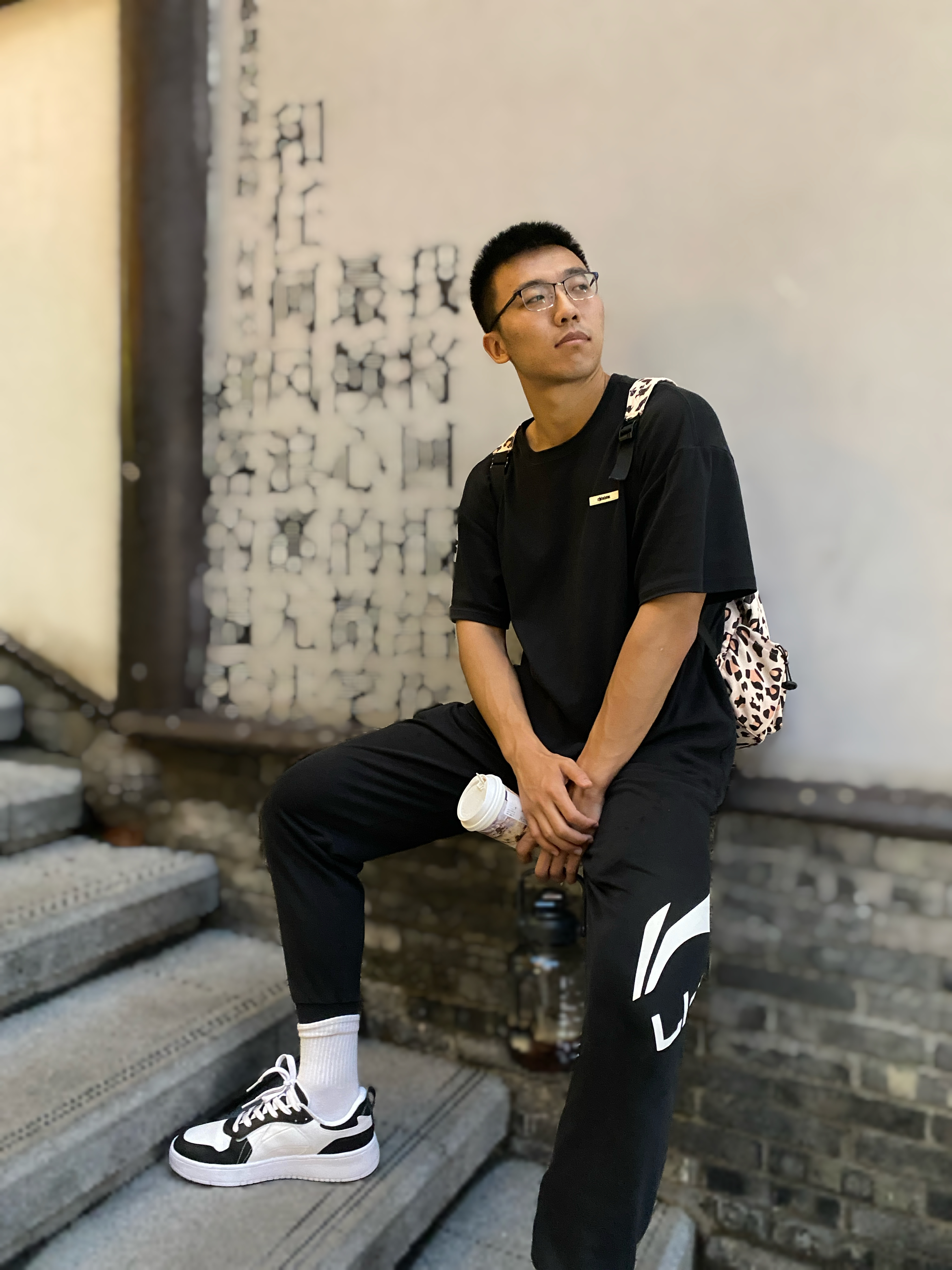}}]{Kai Li} 
	received the M.S. degree in information and communication engineering at the University of Electronic Science and Technology of China, in 2025. His research interests include distributed network systems, wireless caching, and reinforcement learning.
\end{IEEEbiography}

\vspace{-3 em}
\begin{IEEEbiography}[{\includegraphics[width=1in,height=1.25in,clip,keepaspectratio]{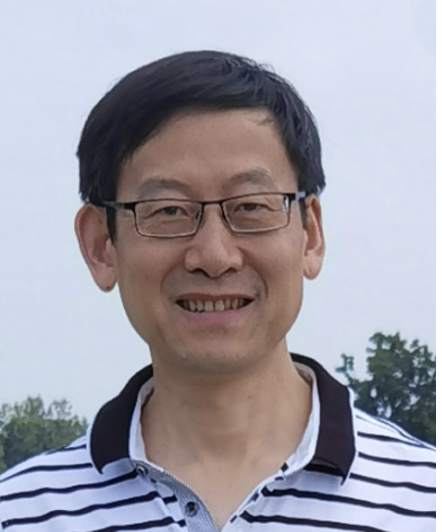}}]{Kun Yang}
	received his PhD from the Department of Electronic \& Electrical Engineering of University College London (UCL), UK. He is currently a Chair Professor of Nanjing University and also an affiliated professor at University of Essex and UESTC. His main research interests include wireless networks and communications, communication-computing cooperation, and new AI (artificial intelligence) for wireless. He has published 500+ papers and filed 50 patents. He serves on the editorial boards of a number of IEEE journals (e.g., IEEE WCM, TVT, TNB). He is a Deputy Editor-in-Chief of IET Smart Cities Journal. He has been a Judge of GSMA GLOMO Award at World Mobile Congress – Barcelona since 2019. He was a Distinguished Lecturer of IEEE ComSoc (2020-2021), a Recipient of the 2024 IET Achievement Medals and the Recipient of 2024 IEEE CommSoft TC’s Technical Achievement Award. He is a Member of Academia Europaea (MAE), a Fellow of IEEE, a Fellow of IET and a Distinguished Member of ACM.
\end{IEEEbiography}

\appendices

\section{Proof of Theorem~\ref{thm: closed form Je}}\label{App: closed form Je}
From the definition of $Y_{nk}(l)$,  we know that it follows a geometric distribution with parameter $r_k\zeta_{nk}$. The first and second moments of $Y_{nk}(l)$ are $\mathbb{E}[Y_{nk}(l)]=\frac{1}{r_k\zeta_{nk}}$ and $\mathbb{E}[Y_{nk}^2(l)]=\frac{2-r_k\zeta_{nk} }{(r_k\zeta_{nk})^2}$, respectively \cite{stochasticprocess}.

By the definition of  $\Gamma_{nk}(l)$, the sequence $\{\Gamma_{nk}(l)\}_l$ is i.i.d over $l$, since the sequence $\{Y_{nk}(l)\}_l$ is i.i.d over $l$. Let $Q(T)$ denote the number of commands sent by eRRH $n$ to CP $k$ in the interval  $[0, T]$. The time-average AoI of CP $k$ at eRRH $n$ can then be expressed as:
\begin{align}
	&\lim_{T\to\infty}\frac{1}{T}\sum_{t=1}^{T}g_{nk}(t)=\lim_{T\to\infty}\frac{1}{T}\sum_{l=1}^{Q(T)}\Gamma_{nk}(l)\nonumber\\
	=&\lim_{T\to\infty}\frac{Q(T)}{T}\cdot\frac{1}{Q(T)}\sum_{l=1}^{Q(T)}\Gamma_{nk}(l)=\frac{\mathbb{E}[\Gamma_{nk}(l)]}{\mathbb{E}[Y_{nk}(l)]},\label{eq: average nk}
\end{align}
where the last equality follows from the Renewal Reward Theorem \cite{stochasticprocess}.

By substituting $\mathbb{E}[Y_{nk}(l)]=\frac{1}{r_k\zeta_{nk}}$ and $\mathbb{E}[Y_{nk}^2(l)]=\frac{2-r_k\zeta_{nk} }{(r_k\zeta_{nk})^2}$ into \eqref{eq: average nk}, and then applying the result to \eqref{eq: eRRHs average AoI}, we obtain:
\begin{align*}
L_e =& \frac{1}{NK}\sum_{n\in[N]}\sum_{k\in[K]}\big(\frac{\mathbb{E}[Y_{nk}^2(l)]}{2\mathbb{E}[Y_{nk}(l)]}+\frac{1}{2}\big)\\
=&\frac{1}{NK}\sum_{n\in[N]}\sum_{k\in[K]}\frac{1}{r_k\zeta_{nk}}.
\end{align*}
By setting $N=K=2$, we recover the desired results.

\section{Proof of Theorem~\ref{thm: upper bounds of Jc}}\label{App: upper bounds of Jc}

We define two {\it independent}  stochastic processes, $\{u_{mk}(t)\}_t$ and $\{v_{mk}(t)\}_t$. The process $\{u_{mk}(t)\}_t$ has the recursion 
\begin{align}\label{eq: upper bound 1 O2S}
	u_{mk}(t+1) = \left\{
	\begin{aligned}
		&1&& b_k r_k\\
		&u_{mk}(t)+1&&1-b_k r_k;
	\end{aligned}
	\right.
\end{align}
and $\{v_{mk}(t)\}_t$ has the recursion
\begin{align}\label{eq: upper bound 2 O2S}
	v_{mk}(t+1) = \left\{
	\begin{aligned}
		&1&&b_kr_k\\
		&g_{1k}(t)+1&& b_k^{(1)}(1-r_k)\\
		&g_{2k}(t)+1&&b_k^{(2)}(1-r_k)\\
		&u_{mk}(t)+1&&1-b_k.
	\end{aligned}
	\right.
\end{align}
By convention, we set $u_{mk}(0)=v_{mk}(0)=1$. 
\begin{lemma}\label{lem: upper bounds O2S}
	In every time slot $t$, 
	\begin{align}
		&\mathbb{E}[h_{mk}(t)]\leq \mathbb{E}[u_{mk}(t)]\label{eq: upper bounds O2S-1},\\
		&\mathbb{E}[h_{mk}(t)]\leq \mathbb{E}[v_{mk}(t)]\label{eq: upper bounds O2S-2}.
	\end{align}
\end{lemma}
\begin{proof}
We prove \eqref{eq: upper bounds O2S-1} and \eqref{eq: upper bounds O2S-2} by mathematical induction. When $t=0$, $\mathbb{E}[h_{mk}(0)]=\mathbb{E}[u_{mk}(0)]=\mathbb{E}[v_{mk}(0)]=1$. We assume that $\mathbb{E}[h_{mk}(t)]\leq \mathbb{E}[u_{mk}(t)]$ and $\mathbb{E}[h_{mk}(t)]\leq \mathbb{E}[v_{mk}(t)]$ for $t\leq l$. 

Now, we consider $t=l+1$. 
From \eqref{eq: recursion of h-1}, $\mathbb{E}[h_{mk}(l+1)]$ can be written as follows:
\begin{align}
	&\mathbb{E}[h_{mk}(l+1)]=b_kr_k +\big(\mathbb{E}[h_{mk}(l)]+1\big)(1-b_k)\nonumber\\ &+\big(\mathbb{E}[\min\{h_{mk}(l), g_{1k}(l)\}]+1\big)\cdot b_k^{(1)}(1-r_k)\nonumber\\
	&+\big(\mathbb{E}[\min\{h_{mk}(l), g_{2k}(l)\}]+1\big)\cdot b_k^{(2)}(1-r_k)\label{eq: proof hmk}.
\end{align}
For $n\in\{1, 2\}$,
\begin{align}
&\mathbb{E}[\min\{h_{mk}(l), g_{nk}(l)\}]\leq\mathbb{E}[h_{mk}(l)],\label{eq: leq gnk}\\
&\mathbb{E}[\min\{h_{mk}(l), g_{nk}(l)\}]\leq\mathbb{E}[g_{nk}(l)]\label{eq: leq hmk}.
\end{align}
Substituting \eqref{eq: leq gnk} into \eqref{eq: proof hmk}, we have
\begin{align*}
\mathbb{E}[h_{mk}(l+1)]&\leq b_kr_k +\big(\mathbb{E}[h_{mk}(l)]+1\big)(1-b_k)\nonumber\\ &+\big(\mathbb{E}[h_{mk}(l)]+1\big) b_k^{(1)}(1-r_k)\nonumber\\
	&+\big(\mathbb{E}[h_{mk}(l)]+1\big) b_k^{(2)}(1-r_k)\\
	&=b_kr_k + \big(\mathbb{E}[h_{mk}(l)]+1\big)(1-b_kr_k).
\end{align*}
By assumption, $\mathbb{E}[h_{mk}(l)]\leq \mathbb{E}[u_{mk}(l)]$. Then,
\begin{align*}
\mathbb{E}[h_{mk}(l+1)]\leq&b_kr_k + \big(\mathbb{E}[u_{mk}(l)]+1\big)(1-b_kr_k)\\
	=&\mathbb{E}[u_{mk}(l+1)].
\end{align*}
The final equality follows from \eqref{eq: upper bound 1 O2S}.
Substituting \eqref{eq: leq hmk} into \eqref{eq: proof hmk}, we have
\begin{align*}
\mathbb{E}[h_{mk}(l+1)]&\leq\sum_{n\in[2]}\big(\mathbb{E}[g_{2k}(l)]+1\big)b_k^{(n)}(1-r_k)\\ &+\big(\mathbb{E}[h_{mk}(l)]+1\big)(1-b_k)+b_kr_k\\
&\leq\sum_{n\in[2]}\big(\mathbb{E}[g_{2k}(l)]+1\big)b_k^{(n)}(1-r_k)\\ &+\big(\mathbb{E}[u_{mk}(l)]+1\big)(1-b_k)=\mathbb{E}[v_{mk}(l+1)].
\end{align*}
The last equality holds due to \eqref{eq: upper bound 2 O2S}. 
\end{proof}
\noindent Substituting \eqref{eq: upper bounds O2S-1} into \eqref{eq: CUs average AoI}, we have
\begin{align}\label{eq: proof upper bound 1 O2S-1}
J_c\leq\lim_{T\to\infty}\frac{1}{T}\sum_{t=1}^{T}\frac{1}{4}\sum_{m\in[2]}\sum_{k\in[2]}\mathbb{E}[u_{mk}(t)].
\end{align}
From \eqref{eq: upper bound 1 O2S}, $\{u_{mk}(t)\}_t$ is stationary and has a geometric distribution with parameter $b_kr_k$ over $t$, then
\begin{align}\label{eq: proof upper bound 1 O2S-2}
\lim_{T\to\infty}\frac{1}{T}\sum_{t=1}^{T}\mathbb{E}[u_{mk}(t)]=\lim_{t\to\infty}\mathbb{E}[u_{mk}(t)]=\frac{1}{b_kr_k}.
\end{align}
Substituting \eqref{eq: proof upper bound 1 O2S-2} into \eqref{eq: proof upper bound 1 O2S-1}, we obtain \eqref{eq: Jc O2S upper bound 1}.
Substituting \eqref{eq: upper bounds O2S-2} into \eqref{eq: CUs average AoI}, we have
\begin{align}\label{eq: proof upper bound 2 O2S-1}
J_c\leq\lim_{T\to\infty}\frac{1}{T}\sum_{t=1}^{T}\frac{1}{4}\sum_{m\in[2]}\sum_{k\in[2]}\mathbb{E}[v_{mk}(t)].
\end{align}
From \eqref{eq: upper bound 2 O2S}, $\{v_{mk}(t)\}_t$ is stationary. Taking the expectation on both sides of \eqref{eq: upper bound 2 O2S} and letting $t\to\infty$, we obtain:
\begin{align}\label{eq: proof upper bound 2 O2S-2}
\lim_{t\to\infty}\mathbb{E}[v_{mk}(t)] =& \sum_{k\in[2]}b_k^{(1)}(1-r_k)\lim_{t\to\infty}\mathbb{E}[g_{nk}(t)] + 1\nonumber \\
+& (1-b_k)\lim_{t\to\infty}\mathbb{E}[u_{mk}(t)].
\end{align}
We know that $g_{nk}(t)$ follows a geometric distribution with parameter $r_k\zeta_{nk}$,
\begin{align}
\lim_{t\to\infty}\mathbb{E}[g_{nk}(t)]=\frac{1}{r_k\zeta_{nk}}\label{eq: average g1k}.
\end{align}
Substituting \eqref{eq: proof upper bound 1 O2S-2}, \eqref{eq: proof upper bound 2 O2S-2}, and \eqref{eq: average g1k} into \eqref{eq: proof upper bound 2 O2S-1},  we obtain \eqref{eq: Jc O2S upper bound 2}.

\section{Proof of Lemma~\ref{lem: optimal H h OS}}\label{App: lem: optimal H h OS}

In {\bf Step 1}, we demonstrate that the minimizer of $\mathbb{E}[h_{mk}(l)]$ is one of $(0, b_k)$, $(b_k, 0)$, or $(b_k/2, b_k/2)$, based on symmetry. In {\bf Step 2} and {\bf Step 3}, we apply mathematical induction to prove Lemma~\ref{lem: optimal H h OS}. Specifically, {\bf Step 2} covers the base cases, while {\bf Step 3} completes the inductive proof.

Note that $b_k$ and $r_k$ are fixed. By substituting $b_k^{(2)} = b_k-b_k^{(1)}$ into $\mathbb{E}[h_{mk}(t)]$, $\mathbb{E}[H_{mk, 1}(t)]$ and  $\mathbb{E}[H_{mk, 2}(t)]$, we observe that they can be viewed as functions of $b_k^{(1)}$. 
For simplicity, we provide the following notations: 
\begin{align*}
&\mathbb{E}[h_{mk}(t)]\triangleq f_1(b_k^{(1)}; t)\\
&\mathbb{E}[H_{mk, 1}(t)]\triangleq f_2(b_k^{(1)}; t)\\
&\mathbb{E}[H_{mk, 2}(t)]\triangleq f_3(b_k^{(1)}; t).
\end{align*}

\noindent{\bf Step 1} We show that a minimizer can be obtained at $(0, b_k)$, $(b_k, 0)$, or $(b_k/2, b_k/2)$.

\begin{lemma}\label{lem: proof symmetric OS}
Fix $b_k$ and $r_k$, we have:
\begin{itemize}
	\item [] (i) $h_{mk}(l+1)$, $H_{mk, 1}(l+1)$ and  $H_{mk, 2}(l+1)$ are symmetric with respect to $(b_k^{(1)}, b_k^{(2)})$.
	\item [] (ii) 	$ f_1(b_k^{(1)}; l+1)$, $ f_2(b_k^{(1)}; l+1)$ and  $ f_3(b_k^{(1)}; l+1)$ have a single symmetric point $b_k^{(1)} = b_k/2$.
\end{itemize}

\end{lemma}
\begin{proof}
	We prove (i) using mathematical induction.  When $t=0$, $H_{mk,1}(0) = H_{mk,2}(0) = h_{mk}(0)=1$, so they are symmetric with respect to $(b_k^{(1)}, b_k^{(2)})$. We assume that $h_{mk}(t)$, $H_{mk, 1}(t)$ and  $H_{mk, 2}(t)$ are symmetric with respect to $(b_k^{(1)}, b_k^{(2)})$ for all $t\leq l$. Now, we consider $t=l+1$. 
	
	Since $h_{mk}(l)$ is symmetric with respect to $(b_k^{(1)}, b_k^{(2)})$, from \eqref{eq: upper bound H1}, $H_{mk, 1}(l+1)$ remains unchanged under the exchange of $b_k^{(1)}$ and $b_k^{(2)}$. Thus, $H_{mk, 1}(l+1)$ is symmetric with respect to $(b_k^{(1)}, b_k^{(2)})$. 
	
	Fix $b_k$ and $r_k$, the stochastic process $\{g_{nk}(t)\}_t$ depends only on $b_k^{(n)}$ for $n\in[2]$. Recall that $h_{mk}(l)$ is symmetric with  respect to $(b_k^{(1)}, b_k^{(2)})$. From \eqref{eq: upper bound H2}, we see that $H_{mk, 2}(l+1)$ remains unchanged if we exchange $b_k^{(1)}$ and $b_k^{(2)}$. Thus, $H_{mk, 2}(l+1)$ is symmetric with respect to $(b_k^{(1)}, b_k^{(2)})$. 
	
	Since both $H_{mk, 1}(l+1)$ and $H_{mk, 2}(l+1)$ are symmetric with respect to $(b_k^{(1)}, b_k^{(2)})$, then from \eqref{eq: hH1H2}, $h_{mk}(l+1) = \min\{H_{mk, 1}(l+1), H_{mk, 2}(l+1)\}$ is symmetric with respect to $(b_k^{(1)}, b_k^{(2)})$.  

We prove (ii) directly from (i). Since $h_{mk}(l+1)$, $H_{mk, 1}(l+1)$ and  $H_{mk, 2}(l+1)$ are symmetric with respect to $(b_k^{(1)}, b_k^{(2)})$. Note that $b_k^{(1)} + b_k^{(2)} = b_k$, and $b_k$ is fixed. Thus, 
$f_1(b^{(1)}_k; l+1)$, $f_2(b^{(1)}_k; l+1)$ and  $f_3(b^{(1)}_k; l+1)$ have a {\it single} symmetric point at $b_k^{(1)}=b_k/2$.
\end{proof}
Since $f_2(b^{(1)}_k; l+1)$ and  $f_3(b^{(1)}_k; l+1)$ are continuous with respect to $b^{(1)}_k$, and $0\leq b_k^{(1)}\leq b_k$. Therefore, a minimizer can be obtained at either the end points or the middle point, i.e., $(0, b_k)$, $(b_k, 0)$, or $(b_k/2, b_k/2)$.

 {\bf Step 2}. We verify initial cases, i.e., $t=1$ and $t=2$. 
 
 When $t=0$, $h_{mk}(0)=1$, which satisfies Lemma~\ref{lem: optimal H h OS}. When $t=1$, from \eqref{eq: recursion of h-1}, 
 \begin{align*}
 	h_{mk}(1)= \left\{
 	\begin{aligned}
 		&1&&b_kr_k\\
 		&2&&1-b_kr_k,
 	\end{aligned}
 	\right.
 \end{align*}
 $\mathbb{E}[h_{mk}(1)]=2-b_kr_k$ remains unchanged when we vary $b_k^{(1)}$.

 When $t=2$, we recurse $h_{mk}(1)$ one more round. Without loss of generality, we calculate $h_{11}(2)$.  From the definition of $h_{mk}(t)$ in  \eqref{eq: recursion of h-1}, we know that $h_{11}(t)$ depends on $g_{11}(t)$, $g_{21}(t)$, and $h_{21}(t)$. Denote $H\triangleq \big(h_{11}(1), g_{11}(1), h_{21}(1), g_{21}(1)\big)$, we have the following $8$ cases:
 \begin{align*}
 	&(1,1,1,1),\,\, (1, 1, 1, 2),\,\,(1,2,1,1), \,\, (1,1,2,2),\\
 	&(1, 2, 2, 1),\,\, (2,1,1,2),\,\, (2,2,1,1),\,\, (2, 2, 2, 2).
 \end{align*}
 First, we consider the subcase $\big(h_{11}(1), h_{21}(1)\big) = (1, 1)$. By calculation, we have
 \begin{align*}
 	q_1 \triangleq & \Pr\{H = (1,1,1,1)\} = 2r_1^2b_1^{(1)}b_1^{(2)}\\
 	q_2 \triangleq & \Pr\{H = (1,1,1,2)\} = r_1(b_1^{(1)})^2\\
 	q_3 \triangleq & \Pr\{H = (1,2,1,1)\} = r_1(b_1^{(2)})^2.
 \end{align*}
 Next, we consider $\big(h_{11}(1), h_{21}(1)\big) = (1, 2)$. By calculation, we have
 \begin{align*}
 	q_4 \triangleq & \Pr\{H = (1,1,2,2)\} = r_1b_1^{(1)}\big(1-b_1+b_1^{(2)}(1-r_1)\big)\\
 	q_5 \triangleq & \Pr\{H = (1,2,2,1)\} = r_1b_1^{(2)}\big(1-b_1+b_1^{(1)}(1-r_1)\big).
 \end{align*}
 Then, we consider $\big(h_{11}(1), h_{21}(1)\big) = (2, 1)$. By calculation, we have
 \begin{align*}
 	q_6 \triangleq & \Pr\{H = (2,1,1,2)\} = r_1b_1^{(1)}\big(1-b_1+b_1^{(2)}(1-r_1)\big)\\
 	q_7 \triangleq & \Pr\{H = (2,2,1,1)\} = r_1b_1^{(2)}\big(1-b_1+b_1^{(1)}(1-r_1)\big).
 \end{align*}
 Finally, we consier $\big(h_{11}(1), h_{21}(1)\big) = (2, 2)$. By calculation, we have
 \begin{align*}
 	q_8 \triangleq & \Pr\{H = (2,2,2,2)\} = (1-r_1)\big(b_k-2r_1b_1^{(1)}b_1^{(2)}\big)\\
 	+&(1-b_k)\big(1-r_1b_1^{(1)}-r_1b_1^{(2)}\big).
 \end{align*}
 Denote $\tilde{H}\triangleq \big(h_{11}(2), h_{21}(2)\big)$. Combining all cases above, when $h_{11}(2)=1$, we have
 \begin{align*}
 	\Pr\{\tilde{H} = (1,1)\} =& q_1+q_2+q_3 \\
 	\Pr\{\tilde{H} = (1,2)\} =&(q_1+q_2+q_3)(q_4+q_5) + (q_4+q_5)^2\\
 	+&(q_4+q_5)(1-r_1)r_1b_1^{(1)}b_1^{(2)} \\
 	\Pr\{\tilde{H} = (1,3)\} =& q_4\big(q_4+r_1b_1^{(2)}(1-b_1)\big)+ q_8(q_4+q_5)\\
 	+&q_5(q_5+r_1(1-b_1)b_1^{(1)}).
 \end{align*}
 When $h_{11}(2)=2$, we have
 \begin{align*}
 	\Pr\{\tilde{H} = (2,1)\} =& (q_1+q_2+q_3)(q_4+q_5) + (q_4+q_5)^2\\
 	\Pr\{\tilde{H} = (2,2)\} =& q_8(q_1+q_2+q_3)\\
 	+& 2q_4\big((1-r_1)(b_1^{(1)}-r_1b_1^{(1)}b_1^{(2)})\big)\\
 	+&2q_5\big((1-r_1)(b_1^{(2)}-r_1b_1^{(1)}b_1^{(2)})\big)\\
 	\Pr\{\tilde{H} = (2,3)\} =& q_4(1-r_1)(b_1^{(2)}-r_1b_1^{(1)}b_1^{(2)})\\
 	+& q_4(1-b_1)(1-r_1b_1^{(1)}-r_1b_1^{(2)})\\
 	+&q_5(1-r_1)(b_1^{(1)}-r_1b_1^{(1)}b_1^{(2)})\\
 	+& q_5(1-b_1)(1-r_1b_1^{(1)}-r_1b_1^{(2)}).
 \end{align*}
 When $h_{11}(2)=3$, we have
 \begin{align*}
 	\Pr\{\tilde{H} = (3,1)\} =& q_4\big(q_4+r_1b_1^{(2)}(1-b_1)\big)+ q_8(q_4+q_5)\\
 	+&q_5(q_5+r_1(1-b_1)b_1^{(1)})\\
 	\Pr\{\tilde{H} = (3,2)\} =& q_4(1-r_1)(b_1^{(2)}-r_1b_1^{(1)}b_1^{(2)})\\
 	+& q_4(1-b_1)(1-r_1b_1^{(1)}-r_1b_1^{(2)})\\
 	+&q_5(1-r_1)(b_1^{(1)}-r_1b_1^{(1)}b_1^{(2)})\\
 	+& q_5(1-b_1)(1-r_1b_1^{(1)}-r_1b_1^{(2)})\\
 	\Pr\{\tilde{H} = (3,3)\} =& q_8^2.
 \end{align*}
 Therefore, $\mathbb{E}[h_{11}(2)]$ can be calculated by
 \begin{align*}
 	\mathbb{E}[h_{11}(2)] =&\Pr\{\tilde{H} = (1,1)\}  + \Pr\{\tilde{H} = (1,2)\} \\
 	+& \Pr\{\tilde{H} = (1,3)\} +2\Pr\{\tilde{H} = (2,1)\}  \\
 	+& 2 \Pr\{\tilde{H} = (2,2)\}+2\Pr\{\tilde{H} = (2,3)\} \\
 	+& 3\Pr\{\tilde{H} = (3,1)\}+3\Pr\{\tilde{H} = (3,2)\}\\
 	+&3\Pr\{\tilde{H} = (3,3)\}\\
 	=&\big(b_1^{(1)}\big)^2(r_1^3b_1-4r_1^2b_1+2r_1^2+3r_1b_1-2r_1)\\
 	+&b_1^{(1)}(-r_1^3b_1^2+4r_1^2b_1^2-2r_1^2b_1-3r_1b_1^2+2r_1b_1)\\
 	+&(-r_1^2b_1^3+2r_1^2b_1^2+r_1b_1^3-r_1b_1^2-3r_1b_1+3).
 \end{align*}
 Taking the second derivative of $\mathbb{E}[h_{11}(2)]$ with respect to $b_1^{(1)}$, we have
 \begin{align*}
 	\frac{\partial^2 (\mathbb{E}[h_{11}(2)])}{\partial (b_1^{(1)})^2} = 6r_1b_1-4r_1-8r_1^2b_1+2r_1^3b_1+4r_1^2.
 \end{align*}
 Therefore,  
 \begin{itemize}
 	\item [] (i) $\mathbb{E}[h_{11}(2)]$ is convex with respect to $b_1^{(1)}$ when $b_1>\frac{2}{3-r_1}$, and it is minimized when $b_1^{(1)} = b_1/2$.
 	\item [] (ii) Conversely, $\mathbb{E}[h_{11}(2)]$ is concave with respect to $b_1^{(1)}$ when $b_1<\frac{2}{3-r_1}$, and it is minimized when $b_1^{(1)} = 0$ or $b_1$.
 	\item [] (iii) When $b_1^{(1)} = \frac{2}{3-r_1}$, $\mathbb{E}[h_{11}(2)]$ remains constant.
 \end{itemize}
The same proof process holds for $h_{mk}(2)$  with $(m, k) = (1, 2)$, $(2, 1)$, and $(2, 2)$, we can obtain similar results.

\noindent{\bf Step 3}. We prove Lemma~\ref{lem: optimal H h OS} by mathematical induction.  We begin by presenting a useful lemma before proceeding with the proof.

\begin{lemma}\label{pro: minimum symmetry}
Let $x\in[0, c]$. If $A(x) = a(x) + a(c-x)$ and $B(x) = b(x) + b(c-x)$, both $A(x)$ and $B(x)$ are continuous and have the global minimizer at $x=c/2$. Then
\begin{align*}
D(x) = a(x)b(x) + a(c-x)b(c-x)
\end{align*}
has the global minimizer at $x=c/2$.
\end{lemma}
\begin{proof}
From the definition of $A(x)$,  $B(x)$ and $D(x)$, all these functions are symmetric at the point $x=c/2$. Since $A(x)$,  $B(x)$ and $D(x)$ are continuous, then the minimizer of each function can be obtained at either $x=0$, $x=c$, or $x=c/2$. By assumption, $x=c/2$ is a minimizer of $A(x)$ and  $B(x)$, then
\begin{align*}
 &a(c) = A(c) > A(c/2) = 2a(c/2)\\
 &b(c) = B(c) > B(c/2) = 2b(c/2),
\end{align*}
which implies
\begin{align*}
a(c)\cdot b(c) > 4 a(c/2)\cdot b(c/2).
\end{align*}
Based on the inequalities above, we can compare $D(c)$ and $D(c/2)$:
\begin{align*}
&D(c/2) = 2 a(c/2)b(c/2) < 4 a(c/2)\cdot b(c/2) \\
& < a(c)\cdot b(c) =D(c).
\end{align*}
This implies $x=c/2$ is a minimizer of $D(x)$.
\end{proof}

From {\bf Step 2}, the initial case is verified. Next, we assume Lemma~\ref{lem: optimal H h OS} holds for all $t\leq l$. Next, we consider $t=l+1$.

\begin{lemma}\label{pro: decreasing threshold w}
 If  $f_i(b_k^{(1)}; l)$ with $i\in\{1,2,3\}$ has the global minimizer $(b_k^{(1)}, b_k^{(2)}) = (b_k/2, b_k/2)$, respectively, then   $f_i(b_k^{(1)}; l+1)$ with $i\in\{1,2,3\}$ has the global minimizer $(b_k^{(1)}, b_k^{(2)}) = (b_k/2, b_k/2)$, respectively.
\end{lemma}
\begin{proof}
\noindent{\bf Step (1)}. Taking the expectation on both sides of \eqref{eq: upper bound H1}, we obtain
\begin{align*}
f_2(b_k^{(1)}; l+1) =  (1-b_k r_k)f_1(b_k^{(1)}; l)+1.
\end{align*}
Note that $(1-b_k r_k)$ is constant with respect to $b_k^{(1)}$ since $b_k$ is fixed. By assumption, $f_1(b_k^{(1)}; l)$ has the global minimizer  $b_k^{(1)}=b_k/2$. Consequently, $f_2(b_k^{(1)}; l+1)$ has the global minimizer $b_k^{(1)}=b_k/2$.

\noindent{\bf Step (2)}. From \eqref{eq: upper bound H2}, $H_{mk, 2}(t+1)$ has an equivalent expression,
\begin{align}\label{eq: proof upper bound H2}
	H_{mk, 2}(t+1) = \left\{
	\begin{aligned}
		&g_{1k}(t+1)&& b_k^{(1)}\\
		&g_{2k}(t+1)&&b_k^{(2)}\\
		&h_{mk}(t)+1&&1-b_k.
	\end{aligned}
	\right.
\end{align} 
Taking the expectation on both sides of \eqref{eq: proof upper bound H2}, we obtain
\begin{align*}
	&f_3(b_k^{(1)}; l+1) = 
	\mathbb{E}[g_{1k}(l+1)]b_k^{(1)}\\
	&+ \mathbb{E}[g_{2k}(l+1)]b_k^{(2)}+f_1(b_k^{(1)}; l)(1-b_k)+1.
\end{align*}
By assumption, $f_1(b_k^{(1)};l)(1-b_k)$ has the global minimizer  $b_k^{(1)}=b_k/2$.  Next, we need to prove 
\begin{fact}\label{fact: sum g}
$\mathbb{E}[g_{1k}(l+1)]b_k^{(1)}+
\mathbb{E}[g_{2k}(l+1)]b_k^{(2)}$ has the global minimizer at $b_k^{(1)} = b_k/2$.
\end{fact} 
\begin{proof}
We prove this statement by mathematic induction.  When $l=0$, it is straightforward to verity Fact~\ref{fact: sum g} is true. We assume Fact~\ref{fact: sum g} is true for all $l\leq j$, now we consider $l = j+1$. By calculation, we have:
\begin{align*}
&\mathbb{E}[g_{1k}(j+1)]b_k^{(1)}+
\mathbb{E}[g_{2k}(j+1)]b_k^{(2)} \\
&= b_k + \sum_{n=1}^{2}(1-r_k\zeta_{nk})b_k^{(n)}\mathbb{E}[g_{nk}(j)].
\end{align*}
It is straightforward to check $\sum_{n=1}^{2}(1-r_k\zeta_{nk})b_k^{(n)}$ has the global minimizer at $b_k^{(1)} = b_k/2$. By assumption,  $\sum_{n=1}^{2}\mathbb{E}[g_{nk}(j)]$ has the global minimizer $b_k^{(1)} = b_k/2$. Utilizing Lemma~\ref{pro: minimum symmetry}, we complete the proof.
\end{proof}
From Fact~\ref{fact: sum g},  $\mathbb{E}[g_{1k}(l+1)]b_k^{(1)}+
\mathbb{E}[g_{2k}(l+1)]b_k^{(2)}$ has the global minimizer at $b_k^{(1)} = b_k/2$. Consequently, $f_3(b_k^{(1)}; l+1)$ has the global minimizer at  $b_k^{(1)}=b_k/2$.

\noindent{\bf Step (3)}.
We prove that $f_1(b_k^{(1)};l+1)$ has the global minimizer at $b_k^{(1)}=b_k/2$ by contradiction. Assume that $b_k^{(1)}=b_k/2$ is not the global minimizer, but $b_k^{(1)}=b'$ is the global minimizer. 
Then, from \eqref{eq: hH1H2},
\begin{align*}
	&\mathbb{E}[\min\{H_{mk, 1}(l+1), H_{mk, 2}(l+1)\}]_{b_k^{(1)}=b'}\\
	&<\mathbb{E}[\min\{H_{mk, 1}(l+1), H_{mk, 2}(l+1)\}]_{b_k^{(1)}=b_k/2}
\end{align*}
Recall that $\mathbb{E}[h_{mk}(t)]$, $\mathbb{E}[H_{mk, 1}(t)]$ and  $\mathbb{E}[H_{mk, 2}(t)]$ have the same single symmetric point $b_k^{(1)}=b_k/2$. Thus, at least one of the following inequalities holds,
\begin{align*}
	&\mathbb{E}[H_{mk, 1}(l+1)]_{b_k^{(1)}=b'}<\mathbb{E}[H_{mk, 1}(l+1)]_{b_k^{(1)}=b_k/2}\\
	&\mathbb{E}[H_{mk, 2}(l+1)]_{b_k^{(1)}=b'}<\mathbb{E}[H_{mk, 2}(l+1)]_{b_k^{(1)}=b_k/2},
\end{align*}
which contradicts the conclusion of either {\bf Step (1)} or {\bf Step (2)}.
Therefore, $b_k^{(1)}=b_k/2$ is the global minimizer of $f_1(b_k^{(1)}; l+1)$. 

From {\bf Step (1)}, {\bf Step (2)}, and {\bf Step (3)}, we complete the proof of Lemma~\ref{pro: decreasing threshold w}. 
\end{proof}
\noindent Next, we define 
\begin{align*}
w(r_k; l+1) = & \sup\Big\{b_k: \mathbb{E}[h_{mk}(l+1)]\\
&\text{ can not be minimized at }b_k^{(1)} = b_k/2\Big\}.
\end{align*}
From the definition of $w(r_k; l+1)$, it represents the supremum of $b_k$ such that, $\mathbb{E}[h_{mk}(l+1)]$ can only be minimized at the end points $b_k^{(1)} = 0$ or $b_k$. In other words, for all $b_k > w(r_k; l+1)$, $b_k^{(1)} = b_k/2$ is the global minimizer of $\mathbb{E}[h_{mk}(l+1)]$. Since $\mathbb{E}[h_{mk}(l+1)]$ is a continuous function of $b_k$,  then if $b_k = w(r_k; l+1)$, either $b_k^{(1)} = 0$, $b_k$ or $b_k/2$ is a minimum point. By Proposition~\ref{pro: decreasing threshold w}, for each $b_k > w(r_k; l)$, since $b_k^{(1)} = b_k/2$ is the global minimizer of $\mathbb{E}[h_{mk}(l)]$, then it is the global minimizer of $\mathbb{E}[h_{mk}(l+1)]$.  Thus, 
\begin{align*}
w(r_k; l+1) \leq w(r_k; l),
\end{align*}
thereby implying that the sequence ${w(r_k; j)}_{j=1}^{l+1}$ is non-increasing.

Finally, we will show the following lemma.
\begin{lemma}\label{pro: b<w is not a minimum point}
For $\forall b_k < w(r_k; l+1)$, $b_k^{(1)}=b_k/2$ can not be the only global minimizer of $\mathbb{E}[h_{mk}(l+1)]$. 
\end{lemma}
\begin{proof}
To complete the proof, we proceed by contradiction. Assume there exists a $b_k = c_0<w(r_k; l+1)$ such that, $b_k^{(1)}=c_0/2$ is the global minimizer of $\mathbb{E}[h_{mk}(l+1)] = f_1(b_k^{(1)}; l+1)$. As discussed in {\bf Step (3)} in the proof of Lemma~\ref{pro: decreasing threshold w}, $b_k^{(1)}=c_0/2$ is the global minimizer of $f_2(b_k^{(1)}; l+1)$ and $f_2(b_k^{(1)}; l+1)$. Recall that
\begin{align*}
f_2(b_k^{(1)}; l+1) =  (1-b_k r_k)f_1(b_k^{(1)}; l)+1,
\end{align*}
which implies $f_1(b_k^{(1)}; l)$ has the global minimizer $b_k^{(1)}=c_0/2$. However, since $c_0<w(r_k; l+1)\leq w(r_k; l)$, by assumption, $b_k^{(1)} = c_0$ is the global minimizer of $f_1(b_k^{(1)}; l)$.  Consequently, we obtain a contradiction, and the proof is complete.
\end{proof}

From Lemma~\ref{pro: decreasing threshold w} and Lemma~\ref{pro: b<w is not a minimum point}, we prove that (i)  $(b_k^{(1)}, b_k^{(2)}) = (b_k, 0)$ is the global minimizer when $b_k<w(r_k; l+1)$, (ii) $(b_k^{(1)}, b_k^{(2)}) = (b_k/2, b_k/2)$ is the global minimizer when $b_k > w(r_k; l+1)$, (iii) $(b_k^{(1)}, b_k^{(2)}) = (b_k/2, b_k/2)$, $(b_k, 0)$, or $(b_k, 0)$ is the global minimizer when $b_k = w(r_k; l+1)$, and (iii) $w(r_{k+1}; l)\leq w(r_k; l)$.  This completes the mathematical induction for Lemma~\ref{lem: optimal H h OS}.

Finally, under stationary randomized policies, both $h_{mk}(l)$ and $g_{nk}(l)$ converge to asymptotically stationary distributions \cite{Modiano-1}. As such, the initial points $h_{mk}(0)$ and $g_{nk}(0)$ have a negligible impact on $h_{mk}(l)$ and $g_{nk}(l)$ when $l$ becomes large, respectively. Thus, $h_{mk}(l+1) = h_{mk}(l)$ almost surely when $l\to\infty$, which implies $\lim_{l\to\infty}w(r_k; l)$ converges.

\section{Proof of Lemma~\ref{lem: optimal H1H2h OS}}\label{App: optimal H1H2h OS} 
We prove the result using mathematical induction. To this end, we first present a useful lemma.
\begin{lemma}\label{pro: minimum symmetry 2d}
	Let $x\in[0, c]$ and $y\in[0, d]$. If $A(x, y) = a(x, y) + a(c-x, d-y)$ and $B(x, y) = b(x, y) + b(c-x, d-y)$, both $A(x, y)$ and $B(x, y)$ are continuous and attain their global minimum at $(x, y)=(c/2, d/2)$. Then
	\begin{align*}
		D(x) = a(x, y)b(x, y) + a(c-x, d-y)b(c-x, d-y)
	\end{align*}
	has the global minimizer at $(x, y)=(c/2, d/2)$.
\end{lemma}
\begin{proof}
This proposition extends Lemma~\ref{pro: minimum symmetry} to functions in two dimensions. The result follows directly by extending the argument used in the proof of Lemma~\ref{pro: minimum symmetry} to two dimensions.
\end{proof}

\noindent{\bf Step 1}. We verify initial cases. When $t=0$, $\sum_{k\in[2]}h_{mk}(0)=2$, which is consistent with the statement of Lemma~\ref{lem: optimal H1H2h OS}. When $t=1$, from \eqref{eq: recursion of h-1}, 
\begin{align*}
h_{mk}(1)= \left\{
\begin{aligned}
&1&& b_kr_k\\
&2&&1-b_kr_k, 
\end{aligned}
\right.
\end{align*}
which implies
\begin{align*}
\sum_{k\in[2]}\mathbb{E}[h_{mk}(1)] = 4 - \sum_{k\in[2]}b_kr_k.
\end{align*}
Recall that $b_2=b-b_1$ and $r_2=r-r_1$, then it is straightforward to verify that $\sum_{k\in[2]}\mathbb{E}[h_{mk}(1)]$ attains its global minimum at $(b_1, r_1)=(b/2, r/2)$.

Assume that Lemma~\ref{lem: optimal H1H2h OS} holds for all $t\leq l$, we now proceed to examine $t=l+1$. 

\noindent{\bf Step 2}. In this step, we prove that Lemma~\ref{lem: optimal H1H2h OS} holds under the condition that, for both $k\in[2]$, an optimal choice of $\big((b_k^{(1)})^*, (b_k^{(2)})^*\big)$ takes the form $\big((b_k^{(1)})^*, (b_k^{(2)})^*\big) = (b_k, 0)$. 

When $\big((b_k^{(1)})^*, (b_k^{(2)})^*\big) = (b_k, 0)$, CU $m$ only sends request to eRRH $1$ (with rate $b_k$). From \eqref{eq: recursion of h-1}, $h_{mk}(l)\geq g_{1k}(l)$ for all $l$. The recursion of $h_{mk}(l)$ is therefore reduced to 
\begin{equation}\label{eq: proof hmkg}
\begin{aligned}
h_{mk}(l+1)= \left\{
\begin{aligned}
&g_{1k}(l+1)&& b_{k}\\
&h_{mk}(l)+1&&1-b_k,
\end{aligned}
\right.
\end{aligned}
\end{equation}
which implies 
\begin{align*}
\sum_{k\in[2]}\mathbb{E}[h_{mk}(l+1)] &= \sum_{k\in[2]}b_k\mathbb{E}[g_{1k}(l+1)] \\
& + \sum_{k\in[2]}(1-b_k)\mathbb{E}[h_{mk}(l)] + 2-b.
\end{align*}
Based on a very similar proof of Fact~\ref{fact: sum g} in Appendix~\ref{App: lem: optimal H h OS}, we can show that
\begin{align*}
\sum_{k\in[2]}b_k\mathbb{E}[g_{1k}(l+1)]
\end{align*}
attains its global minimum at $(b_1, r_1) = (b/2, r/2)$.

By assumption, $\sum_{k\in[2]}\mathbb{E}[h_{mk}(l)]$ has the global minimizer at $(b_1, r_1) = (b/2, r/2)$. Since $\sum_{k\in[2]}(1-b_k) = 2-b$ is a constant. From Lemma~\ref{pro: minimum symmetry 2d}, $\sum_{k\in[2]}(1-b_k)\mathbb{E}[h_{mk}(l)]$ has the global minimizer at  $(b_1, r_1) = (b/2, r/2)$. Therefore, $\sum_{k\in[2]}\mathbb{E}[h_{mk}(l+1)]$ has the global minimizer at $(b_1, r_1)=(b/2, r/2)$.

\noindent{\bf Step 3}. In this step, we prove that $\sum_{k\in[2]}\mathbb{E}[h_{mk}(l+1)]$ has the global minimizer at $(b_1, r_1) = (b/2, r/2)$ in the case where, for both $k\in[2]$, an optimal choice of $\big((b_k^{(1)})^*, (b_k^{(2)})^*\big)$ has the form $\big((b_k^{(1)})^*, (b_k^{(2)})^*\big) = (b_k/2, b_k/2)$.  For clarity, we define the following functions:
\begin{align*}
\sum_{k\in[2]}\mathbb{E}[h_{mk}(l+1)]\triangleq&f_1(b_1, r_1; l+1)\\
\sum_{k\in[2]}\mathbb{E}[H_{mk, 1}(l+1)]\triangleq&f_2(b_1, r_1; l+1)\\
\sum_{k\in[2]}\mathbb{E}[H_{mk, 2}(l+1)]\triangleq&f_3(b_1, r_1; l+1),
\end{align*}

\begin{lemma}\label{lem: E sum H1 H2 symmetric}
$f_i(b_1, r_1; l+1)$ with $i\in\{1,2,3\}$ are symmetric with  respect to $(b_1, r_1) = (b/2, r/2)$.
\end{lemma}
\begin{proof}
The proof is very similar to the proof of Lemma~\ref{lem: proof symmetric OS}.
\end{proof}

\begin{lemma}\label{lem: proof min br2}
$f_1(b_1, r_1; l+1)$, $f_2(b_1, r_1; l+1)$ and  $f_3(b_1, r_1; l+1)$ have the global minimizer at $(b_1, r_1)=(b/2, r/2)$.
\end{lemma}
\begin{proof}
The proof is completed by the mathematical induction. It is straightforward to verify the base case. We assume Lemma~\ref{lem: proof min br2} holds for all $t\leq l$. Now, we consider $t=l+1$.

\noindent{\bf Step (1)}. From \eqref{eq: upper bound H1}, we have
\begin{align*}
f_2(b_1, r_1; l+1) = 2 + \sum_{k\in[2]}(1-b_k r_k)\mathbb{E}[h_{mk}(l)]. 
\end{align*}
Note that the following function
\begin{align*}
\sum_{k\in[2]}(1-b_k r_k)=2-b_1r_1 - b_2r_2
\end{align*}
has the global minimizer at $(b_1, r_1)=(b/2,  r/2)$. In addition, by assumption, $\sum_{k\in[2]}\mathbb{E}[h_{mk}(l)]$ has the global minimizer at $(b_1, r_1)=(b/2,  r/2)$. From Lemma~\ref{pro: minimum symmetry 2d}, $f_2(b_1, r_1; l+1)$ has the global minimizer at $(b_1, r_1)=(b/2,  r/2)$. 

\noindent{\bf Step (2)}.  From  \eqref{eq: upper bound H2}, we derive:
\begin{align*}
f_3(b_1, r_1; l+1) =& \sum_{k\in[2]}\sum_{n\in[2]}(1-r_k)b_k^{(n)}g_{nk}(l)\\ + & \sum_{k\in[2]}(1-b_k)h_{mk}(l) + 2.
\end{align*}
Recall that $b_k^{(1)} = b_k^{(2)} = b_k/2$, then $\{g_{1k}(t)\}_t$ and $\{g_{2k}\}_t$ are identical, thus 
\begin{align*}
\sum_{k\in[2]}\sum_{n\in[2]}(1-r_k)b_k^{(n)}\mathbb{E}[g_{nk}(l)]=\sum_{k\in[2]}(1-r_k)\mathbb{E}[g_{1k}(l)].
\end{align*}
Therefore,
\begin{align*}
	f_3(b_1, r_1; l+1) =& \sum_{k\in[2]}(1-r_k)\mathbb{E}[g_{1k}(l)]\\ + & \sum_{k\in[2]}(1-b_k)h_{mk}(l) + 2.
\end{align*}
By a very similar proof in Fact~\ref{fact: sum g}, we can prove that $\sum_{k\in[2]}(1-r_k)\mathbb{E}[g_{1k}(l)]$ has the global minimizer at $(b_1, r_1) = (b/2, r/2)$. In addition, by assumption, $\sum_{k\in[2]}\mathbb{E}[h_{mk}(l)]$ has the global minimizer at $(b_1, r_1) = (b/2, r/2)$. Utilizing Lemma~\ref{pro: minimum symmetry 2d} again, $\sum_{k\in[2]}(1-b_k)h_{mk}(l)$  has the global minimizer at $(b_1, r_1) = (b/2, r/2)$.  
Thus, $f_3(b_1, r_1; l+1)$  has the global minimizer at $(b_1, r_1) = (b/2, r/2)$. 

\noindent{\bf Step (3)}. We prove that $f_1(b_1, r_1; t)$ has the global minimizer at $(b_1, r_1)=(b/2, r/2)$ by contradiction. Consider $t=l+1$.
Since $f_2(b_1, r_1; l+1)$ and $f_3(b_1, r_1; l+1)$ have the same global minimizer at $(b_1, r_1)=(b/2, r/2)$, both $f_2(b_1, r_1; l+1)$ and $f_3(b_1, r_1; l+1)$ have a {\it single} symmetric point at $(b_1, r_1)=(b/2, r/2)$. Assume that $f_1(b_1, r_1; l)$ has another global minimizer at $(b', r')$ such that
\begin{align*}
	f_1(b', r'; l+1) < f_1(b/2, r/2; l+1),
\end{align*}
which implies at least one of the following two inequalities hold,
\begin{align*}
	&f_2(b', r'; l+1) < f_2(b/2, r/2; l+1)\\
	&f_3(b', r'; l+1) < f_3(b/2, r/2; l+1).
\end{align*}
It contradicts {\bf Step (1)} or {\bf Step (2)}. 

This completes the proof of Lemma~\ref{lem: proof min br2}.
\end{proof}

From Lemma~\ref{lem: proof min br2}, we complete the proof in {\bf Step 3}.

\noindent{\bf Step 4}. In this step, we consider the case where an optimal choice of $\big((b_1^{(1)})^*, (b_1^{(2)})^*\big)$ takes the form $\big((b_1^{(1)})^*, (b_1^{(2)})^*\big) = (b_1, 0)$, and an optimal choice of $\big((b_2^{(1)})^*, (b_2^{(2)})^*\big)$ takes the form $\big((b_2^{(1)})^*, (b_2^{(2)})^*\big) = (b_2/2, b_2/2)$. We aim to demonstrate that such a case cannot exist.

Assume that in time slot $l$, $r_1$ and $r_2$ have small disturbances $\epsilon_1$ and $\epsilon_2$, respectively, such that $\epsilon_1+\epsilon_2=0$. We next evaluate the disturbances introduced in $\mathbb{E}[h_{mk}(l+1)]$ due to these perturbations.

From \eqref{eq: proof hmkg}, let us denote the disturbance by
\begin{align*}
	&\Delta_{1}(\epsilon_1; b_1, b_2) \triangleq \mathbb{E}[h_{m1}(l+1)|r_1+\epsilon_1, h_{m1}(l)] \\
	&- \mathbb{E}[h_{m1}(l+1)|r_1, h_{m1}(l)] = -b_1\epsilon_1\mathbb{E}[g_{1k}(l)]\big|_{b_1, r_1}.
\end{align*}
Similarly, from \eqref{eq: recursion of h-1}, let us denote the disturbance by
\begin{align*}
	&\Delta_{2}(\epsilon_2; b_1, b_2) \triangleq \mathbb{E}[h_{m2}(l+1)|r_2+\epsilon_2, h_{m2}(l)] \\
	&- \mathbb{E}[h_{m2}(l+1)|r_2, h_{m1}(l)] \\
	&= -b_2\epsilon_2\mathbb{E}[\min\{h_{m2}(l), g_{1k}(l)\}]\big|_{b_2/2, r_2}.
\end{align*}
By the definitions of $h_{mk}(l)$ in \eqref{eq: recursion of h-1}, it follows that $\Delta_1(\epsilon_1; b_1, b_2) + \Delta_2(\epsilon_2; b_1, b_2) \neq 0$ if $b_1 \neq b_2$ and $r_1 \neq r_2$. If $\Delta_{1}(\epsilon_1; b_1, b_2) + \Delta_{2}(\epsilon_2; b_1, b_2) < 0$, applying these disturbances will reduce $\sum_{k\in[2]}\mathbb{E}[h_{mk}(l+1)]$.
Conversely, if $\Delta_{1}(\epsilon_1; b_1, b_2) + \Delta_{2}(\epsilon_2; b_1, b_2) > 0$, we can exchange $\epsilon_1$ and $\epsilon_2$, resulting in the corresponding 
$\Delta_{1}(\epsilon_1; b_1, b_2) + \Delta_{2}(\epsilon_2; b_1, b_2) < 0$, This ensures that applying these altered disturbances will always lead to a decrease in the quantity $\sum_{k\in[2]}\mathbb{E}[h_{mk}(l+1)]$.
When $b_1=b_2$ and $r_1=r_2$, we have $\Delta_{1}(\epsilon_1; b_1, b_2)+\Delta_{2}(\epsilon_2; b_1, b_2)=0$ if $\epsilon_1+\epsilon_2=0$, and it merges into {\bf Step 2} or {\bf Step 3}. 

The same argument applies to perturbations in $b_1$ and $b_2$. Thus, we demonstrate that any $(b_1, b_2, r_1, r_2)$ cannot be a global minimizer unless  $b_1 = b_2$ and $r_1 = r_2$. This implies that the scenario considered in this step does not exist.

From {\bf Step 1} to {\bf Step 4}, we complete the proof of Lemma~\ref{lem: optimal H1H2h OS}.

\section{Transition Probability Matrix}\label{App: steady-state distribution phi}
 Denote $\bar{r}_k = 1-r_k$, $r_0=r/2$, $\bar{r}_0=1-r_0$, and $\bar{\zeta}_{nk} = (1-b_k^{(n)})^2$. Let $\phi = (\phi_1, \phi_2)$ represent the values of truncated age in the current time slot, and $\phi' = (\phi_1', \phi_2')$ be the values of truncated age in the next time slot. Denote
\begin{align*}
&\mathcal{G}_A^{z} =\{\phi'| \phi'=(1, 1)\},\\ &\mathcal{G}_B^{z}=\{\phi'|\phi' = (1, \min\{\phi_2+1,z\})\},\\
&\mathcal{G}_C^{z} =\{\phi'|\phi'=(\min\{\phi_1+1,z\}, 1)\}\\
&\mathcal{G}_D^{z} =\{\phi' | \phi'= (\min\{\phi_1+1,z\}, \min\{\phi_2+1,z\})\},
\end{align*}
and
\begin{align*}
&\mathcal{G}_1^z = \{(\phi_1, \phi_2)|\phi_2 < \phi_1 < z\}\\
&\mathcal{G}_2^z = \{(\phi_1, \phi_2)|\phi_1<\phi_2 < z\}\\
&\mathcal{G}_3^z = \{(\phi_1, \phi_2)|\phi_1=\phi_2\}.
\end{align*}

If  $\phi\in\mathcal{G}_3^z$ and $\phi'\in \mathcal{G}_A^{z}$, $\Pr\{\phi\to\phi'\}=r_0\zeta_{n1}r_0\zeta_{n2}$; if  $\phi\in\mathcal{G}_1^z\cup\mathcal{G}_2^z$ and $\phi'\in \mathcal{G}_A^{z}$ $\Pr\{\phi\to\phi'\}=r_1\zeta_{n1}r_2\zeta_{n2}$. Then, we obtain
\begin{align}\label{eq: transition probability matrix-1}
	{\bf P}_{\phi \phi'}=\left\{
	\begin{aligned}
		&r_0^2\zeta_{n1}\zeta_{n2},\,\,\phi\in\mathcal{G}^z_3,\phi'\in\mathcal{G}_A^{z}\\
		&r_1r_2\zeta_{n1}\zeta_{n2},\,\,\phi\in\mathcal{G}_1^z\cup\mathcal{G}_2^z,\phi'\in\mathcal{G}_A^{z}
	\end{aligned}
	\right.
\end{align}
If  $\phi\in\mathcal{G}_3^z$ and $\phi'\in \mathcal{G}_B^{z}$, $\Pr\{\phi\to\phi'\}=r_0\zeta_{n1}\cdot\big(\bar{r}_0\zeta_{n2} + \bar{\zeta}_{n2}\big)$. If  $\phi\in\mathcal{G}_1^z$ and $\phi'\in \mathcal{G}_B^{z}$, $\Pr\{\phi\to\phi'\}=r_1\zeta_{n1}\cdot(\bar{r}_2\zeta_{n2} + \bar{\zeta}_{n2})$. If  $\phi\in\mathcal{G}_2^z$ and $\phi'\in \mathcal{G}_B^{z}$, $\Pr\{\phi\to\phi'\}=r_2\zeta_{n1}\cdot(\bar{r}_1\zeta_{n2} + \bar{\zeta}_{n2})$. Then, we obtain
\begin{align}\label{eq: transition probability matrix-2}
	{\bf P}_{\phi \phi'}=\left\{
	\begin{aligned}
		&r_0\zeta_{n1}\big(\bar{r}_0\zeta_{n2} + \bar{\zeta}_{n2}\big),\,\,\phi\in\mathcal{G}_3^z, \phi'\in\mathcal{G}_B^{z}\\ 
		&r_1\zeta_{n1}(\bar{r}_2\zeta_{n2} + \bar{\zeta}_{n2}),\,\,\phi\in\mathcal{G}^z_1,\phi'\in\mathcal{G}_B^{z}\\
		&r_1\zeta_{n2}(\bar{r}_2\zeta_{n1} + \bar{\zeta}_{n1}),\,\,\phi\in\mathcal{G}^z_2, \phi'\in\mathcal{G}_B^{z}.
	\end{aligned}
	\right.
\end{align}

If  $\phi\in\mathcal{G}_3^z$ and $\phi'\in \mathcal{G}_C^{z}$, $\Pr\{\phi\to\phi'\}=r_0\zeta_{n2}\cdot\big(\bar{r}_0\zeta_{n1} + \bar{\zeta}_{n1}\big)$. If  $\phi\in\mathcal{G}_1^z$ and $\phi'\in \mathcal{G}_C^{z}$, $\Pr\{\phi\to\phi'\}=r_2\zeta_{n2}\cdot(\bar{r}_1\zeta_{n1} + \bar{\zeta}_{n1})$.If  $\phi\in\mathcal{G}_2^z$ and $\phi'\in \mathcal{G}_C^{z}$, $\Pr\{\phi\to\phi'\}=r_1\zeta_{n2}\cdot(\bar{r}_2\zeta_{n1} + \bar{\zeta}_{n1})$. Then, we obtain
\begin{align}\label{eq: transition probability matrix-3}
	{\bf P}_{\phi \phi'}=\left\{
	\begin{aligned}
		&r_0\zeta_{n2}\big(\bar{r}_0\zeta_{n1} + \bar{\zeta}_{n1}\big),\,\,\phi\in\mathcal{G}_3^z, \phi'\in\mathcal{G}_C^{z}\\ 
		&r_2\zeta_{n2}(\bar{r}_1\zeta_{n1} + \bar{\zeta}_{n1}),\,\,\phi\in\mathcal{G}_1^z, \phi'\in\mathcal{G}_C^{z}\\
		&r_1\zeta_{n2}(\bar{r}_2\zeta_{n1} + \bar{\zeta}_{n1}),\,\,\phi\in\mathcal{G}_2^z, \phi'\in\mathcal{G}_C^{z}
	\end{aligned}
	\right.
\end{align}

If  $\phi\in\mathcal{G}_3^z$ and $\phi'\in \mathcal{G}_D^{z}$, $\Pr\{\phi\to\phi'\}=(\bar{r}_0\zeta_{n1}+\bar{\zeta}_{n1})\cdot\big(\bar{r}_0\zeta_{n2} + \bar{\zeta}_{n2}\big)$. If  $\phi\in\mathcal{G}_1^z$ and $\phi'\in \mathcal{G}_D^{z}$, $\Pr\{\phi\to\phi'\}=(\bar{r}_1\zeta_{n1}+\bar{\zeta}_{n1})\cdot\big(\bar{r}_2\zeta_{n2} + \bar{\zeta}_{n2}\big)$. If  $\phi\in\mathcal{G}_2^z$ and $\phi'\in \mathcal{G}_D^{z}$, $\Pr\{\phi\to\phi'\}=(\bar{r}_2\zeta_{n1}+\bar{\zeta}_{n1})\cdot\big(\bar{r}_1\zeta_{n2} + \bar{\zeta}_{n2}\big)$. 
Then, we obtain
\begin{align}\label{eq: transition probability matrix-4}
{\bf P}_{\phi \phi'}=\left\{
\begin{aligned}
&(\bar{r}_0\zeta_{n1}+\bar{\zeta}_{n1})\big(\bar{r}_0\zeta_{n2} + \bar{\zeta}_{n2}\big),\,\,\phi\in\mathcal{G}_3^z, \phi'\in\mathcal{G}_D^{z}\\ 
&(\bar{r}_1\zeta_{n1}+\bar{\zeta}_{n1})\big(\bar{r}_2\zeta_{n2} + \bar{\zeta}_{n2}\big),\,\,\phi\in\mathcal{G}_1^z, \phi'\in\mathcal{G}_D^{z}\\
&(\bar{r}_2\zeta_{n1}+\bar{\zeta}_{n1})\big(\bar{r}_1\zeta_{n2} + \bar{\zeta}_{n2}\big),\,\,\phi\in\mathcal{G}_2^z, \phi'\in\mathcal{G}_D^{z}.
\end{aligned}
\right.
\end{align}

\section{Proof of Theorem~\ref{thm: upper bounds of Jc N2S}}\label{App: thm: upper bounds of Jc N2S}
We define two {\it independent} stochastic processes, $\{u_{mk}'(t)\}_t$ and $\{v_{mk}'(t)\}_t$. $\{u_{mk}'(t)\}_t$ is defined recursively as
\begin{align*}
	u_{mk}'(t+1) = \left\{
	\begin{aligned}
		&1&& b_k r_2\\
		&u_{mk}'(t)+1&&1-b_k r_2;
	\end{aligned}
	\right.
\end{align*}
and $\{v_{mk}'(t)\}_t$ is recursively defined as 
\begin{align*}
	v_{mk}'(t+1) = \left\{
	\begin{aligned}
		&1&&b_kr_2\\
		&g_{1k}(t)+1&& b_k^{(1)}(1-r_2)\\
		&g_{2k}(t)+1&&b_k^{(2)}(1-r_2)\\
		&u_{mk}'(t)+1&&1-b_k,
	\end{aligned}
	\right.
\end{align*}
where $g_{1k}(t)$ and $g_{2k}(t)$ are defined as in \eqref{eq: nonoblivious recursion g}. By convention, we set the initial values as  $u_{mk}'(0)=v_{mk}'(0)=1$. Following the same reasoning as in Lemma~\ref{lem: upper bounds O2S}, we use mathematical induction to prove that
\begin{align*}
	\mathbb{E}[h_{mk}(t)]\leq \mathbb{E}[u_{mk}'(t)],\quad\mathbb{E}[h_{mk}(t)]\leq \mathbb{E}[v_{mk}'(t)].
\end{align*}
The remainder of the proof proceeds analogously to that of Theorem~\ref{thm: upper bounds of Jc} in Appendix~\ref{App: upper bounds of Jc}.

\section{Proof of Lemma~\ref{lem: optimal H h NS}}\label{App: lem: optimal H h NS}
The entire proof is similar to that of Lemma~\ref{lem: optimal H h OS} in Appendix~\ref{App: lem: optimal H h OS}. First of all, we provide the following notations for brevity,
\begin{align*}
\mathbb{E}[h_{mk}(t)]\triangleq&f_1(b_1^{(1)}, b_2^{(1)}; t)\\
\mathbb{E}[H_{mk, 1}'(t)]\triangleq&f_2(b_1^{(1)}, b_2^{(1)}; t)\\
\mathbb{E}[H_{mk, 2}'(t)]\triangleq&f_3(b_1^{(1)}, b_2^{(1)}; t),
\end{align*}

\noindent{\bf Step 1}. By a similar proof of Lemma~\ref{lem: proof symmetric OS}, we show that $f_i(b_1^{(1)}, b_2^{(1)}; t)$ with $i\in\{1,2,3\}$ are symmetric at $(b_1/2, b_2/2)$. Similar to the counterparts in {\bf Step 1} of Appendix~\ref{App: lem: optimal H h OS}, we show that the global minimizer of $f_i(b_1^{(1)}, b_2^{(1)}; t)$, $i\in\{1,2,3\}$ , occurs at $(b_1^{(1)}, b_2^{(1)})=(b_1, b_2)$ or $(0, 0)$, or $(b_1^{(1)},  b_2^{(1)}) = (b_1/2, b_2/2)$.

\noindent{\bf Step 2}. Similar to {\bf Step 2} in the proof of Lemma~\ref{lem: proof symmetric OS}, through lengthy but straightforward derivations, the initial cases can be verified (we omit these derivations for brevity).

\noindent{\bf Step 3}. In this step, we complete the proof by mathematic induction.  Before proceeding with the induction step, we first present the following lemma. From \eqref{eq: taukt} and \eqref{eq: nonoblivious recursion g}, $r_{nk}(t)$ depends on $b_1^{(n)}$ and $b_2^{(n)}$. Denote
\begin{align*}
q_1(k; t)=&\sum_{n\in[2]}b_k^{(n)}\mathbb{E}[r_{nk}(t)]\\
q_{2}(k; t)=&\sum_{n\in[2]}\zeta_{nk}\mathbb{E}[r_{nk}(t)].
\end{align*}
\begin{lemma}\label{pro: proof q4q5}
For $k\in[2]$, if $b_k$ and $r_k$ with $k\in[2]$ are fixed, then 
\begin{itemize}
\item [](i) $\zeta_{nk}\mathbb{E}[r_{nk}(t)]$ increases with $b_1^{(n)}$, $b_2^{(n)}$, and $r_1$. 
\item [] (ii) $q_{1}(k; t)$ has the global maximizer at $(b_1^{(1)}, b_2^{(1)}) = (b_1/2, b_2/2)$.
\item [] (iii) 
$q_{2}(k; t)$ has the global maximizer at $(b_1^{(1)}, b_2^{(1)}) = (b_1/2, b_2/2)$.
	\end{itemize}
\end{lemma}
\begin{proof}
Based on \eqref{eq: taukt}, we copy the expression of $\mathbb{E}[r_{nk}(t)]$ as follows,
\begin{align*}
\mathbb{E}[r_{nk}(t)]=&\Pr\{G_{n}(t)\in \mathcal{G}_{n1}\}r_k +\Pr\{G_{n}(t)\in \mathcal{G}_{n2}\}(r-r_k)\\
+&\Pr\{G_{n}(t)\in \mathcal{G}_{n3}\}r/2.
\end{align*}
To prove Proposition~\ref{pro: proof q4q5}, we first consider the limiting case where $z\to\infty$, without loss of generality. In fact, when $z$ is bounded, denote 
\begin{align*}
\mathcal{G}_{n3, 1} =& \Big\{G_n(t)|g_{n1}(t) \geq g_{n2}(t)\geq z, t\geq 0\Big\}\\
\mathcal{G}_{n3, 2} =& \Big\{G_n(t)|g_{n2}(t) > g_{n1}(t)\geq z, t\geq 0\Big\}.
\end{align*}
Then, there exists $r_1'$ and $r_2'$, such that $r_2\leq r_2'\leq r_1'\leq r_1$ and $r_1'+r_2'=r$, which makes $\mathbb{E}[r_{nk}(t)]$ can be equivalently expressed as
\begin{align*}
\mathbb{E}[r_{nk}(t)]=&\Pr\{G_{n}(t)\in \mathcal{G}_{n1}\cup\mathcal{G}_{n3, 1}\}r_k' \\
+&\Pr\{G_{n}(t)\in \mathcal{G}_{n2}\cup\mathcal{G}_{n3, 1}\}(r-r_k').
\end{align*}

In the case when  $z\to\infty$, we denote
\begin{align*}
\mathcal{G}_{n1}^*=\mathcal{G}_{n1}\cup \mathcal{G}_{n3, 1},\,\, \mathcal{G}_{n2}^*=\mathcal{G}_{n2}\cup \mathcal{G}_{n3, 2}.
\end{align*}
Then, the process over states $\mathcal{G}_{n1}^*$ and $\mathcal{G}_{n2}^*$ form a Markov chain. We have
\begin{align*}
\Pr\{\mathcal{G}_{n1}^*\to \mathcal{G}_{n2}^*\} =& \zeta_{n1}r_1\big(1-\zeta_{n2}r_2\big)\\
\Pr\{\mathcal{G}_{n1}^*\to \mathcal{G}_{n1}^*\} =& 1-\zeta_{n1}r_1\big(1-\zeta_{n2}r_2\big)\\
\Pr\{\mathcal{G}_{n2}^*\to \mathcal{G}_{n1}^*\} =& \zeta_{n2}r_1\\
\Pr\{\mathcal{G}_{n2}^*\to \mathcal{G}_{n2}^*\} =& 1-\zeta_{n2}r_1,
\end{align*}
which implies
\begin{align*}
	&\Pr\{\mathcal{G}_{n1}^*\}=\frac{\zeta_{n2}r_1}{r_1\zeta_{n1}+r_1\zeta_{n2}-r_1r_2\zeta_{n1}\zeta_{n2}},\\
	&\Pr\{\mathcal{G}_{n2}^*\}=\frac{\zeta_{n1}r_1(1-\zeta_{n2}r_2)}{r_1\zeta_{n1}+r_1\zeta_{n2}-r_1r_2\zeta_{n1}\zeta_{n2}}.
\end{align*}
Therefore, from \eqref{eq: taukt}, we have
\begin{align*}
\mathbb{E}[r_{n1}(t)] =& \frac{r_1^2\zeta_{n2}+r_1r_2\zeta_{n1}-r_2^2r_1\zeta_{n1}\zeta_{n2}}{r_1\zeta_{n1}+r_1\zeta_{n2}-r_1r_2\zeta_{n1}\zeta_{n2}}\\
\mathbb{E}[r_{n2}(t)] =& \frac{r_1r_2\zeta_{n2}+r_1^2\zeta_{n1}-r_1^2r_2\zeta_{n1}\zeta_{n2}}{r_1\zeta_{n1}+r_1\zeta_{n2}-r_1r_2\zeta_{n1}\zeta_{n2}}.
\end{align*}

First, we prove $\zeta_{nk}\mathbb{E}[r_{nk}(t)]$ increases with $b_1^{(n)}$ and $b_2^{(n)}$. We only consider $k=1$; the proof for $k=2$ is similar,
\begin{align}\label{eq: zetan1}
\zeta_{n1}\mathbb{E}[r_{n1}(t)]=\zeta_{n1}\frac{\frac{r_1}{\zeta_{n1}}+\frac{r_2}{\zeta_{n2}} -r_2^2}{\frac{1}{\zeta_{n1}}+\frac{1}{\zeta_{n2}} -r_2}.
\end{align}
Note that $\zeta_{nk}=1-(1-b_k^{(n)})^2$ is an increasing function of $b_k^{(n)}$ since $b_k^{(n)}\leq 1$. Substituting this into \eqref{eq: zetan1}, and note that $1\geq r_1\geq r_2\geq 0$ and $0\leq \zeta_{nk}\leq 1$, we have: $\zeta_{n1}\mathbb{E}[\tau_{n1}(t)]$ increases with $b_1^{(n)}$, $b_2^{(n)}$, and $r_1$. By a similar approach, $\zeta_{n2}\mathbb{E}[\tau_{n2}(t)]$ increases with $b_1^{(n)}$, $b_2^{(n)}$, and $r_1$.

Second, we prove $q_1(k; t)$ has the global maximizer at $(b_1^{(1)}, b_2^{(1)}) = (b_1/2, b_2/2)$. Recall that
\begin{align*}
&q_{4}(k; t)=\sum_{n\in[2]}b_k^{(n)}\mathbb{E}[r_{nk}(t)].
\end{align*}
It suffices to consider
\begin{align*}
b_1^{(n)}\mathbb{E}[r_{n1}(t)]=b_1^{(n)}\frac{\frac{r_1}{\zeta_{n1}}+\frac{r_2}{\zeta_{n2}} -r_2^2}{\frac{1}{\zeta_{n1}}+\frac{1}{\zeta_{n2}} -r_2}
\end{align*}
as an example.
Since $\zeta_{nk}=1-(1-b_k^{(n)})^2$ is concave with respect to $b_k^{(n)}$, and $r_1\geq r_2$, then  $b_1^{(n)}\frac{\frac{r_1^2}{\zeta_{n1}}+\frac{r_1r_2}{\zeta_{n2}} -r_1r_2^2}{\frac{r_1}{\zeta_{n1}}+\frac{r_1}{\zeta_{n2}} -r_1r_2}$ is concave with respect to $b_k^{(n)}$ with $k\in[2]$. By similar analysis for $b_2^{(n)}\mathbb{E}[r_{n2}(t)]$, we derive that $q_{1}(k; t)$ is concave with respect to $b_k^{(n)}$ with $k\in[2]$, which has the global maximizer at $(b_1^{(1)}, b_2^{(1)}) = (b_1/2, b_2/2)$.

Third, we prove $q_2(k; t)$ has the  global maximizer  at $(b_1^{(1)}, b_2^{(1)}) = (b_1/2, b_2/2)$. Recall that
\begin{align*}
&q_{2}(k; t)=\sum_{k\in[2]}\sum_{n\in[2]}\zeta_{n1}\mathbb{E}[r_{n1}(t)].
\end{align*}
It suffices to consider
\begin{align*}
\zeta_{n1}\mathbb{E}[r_{n1}(t)]=\zeta_{n1}\frac{\frac{r_1}{\zeta_{n1}}+\frac{r_2}{\zeta_{n2}} -r_2^2}{\frac{1}{\zeta_{n1}}+\frac{1}{\zeta_{n2}} -r_2}
\end{align*}  
as an example. Since $\zeta_{nk}=1-(1-b_k^{(n)})^2$ is concave with respect to $b_k^{(n)}$, and $r_1\geq r_2$, then $\zeta_{n1}\frac{\frac{r_1}{\zeta_{n1}}+\frac{r_2}{\zeta_{n2}} -r_2^2}{\frac{1}{\zeta_{n1}}+\frac{1}{\zeta_{n2}} -r_2}$ is concave with respect to $b_k^{(n)}$ with $k\in[2]$. By similar analysis for $\zeta_{n2}\mathbb{E}[r_{n2}(t)]$, we derive that $q_2(k; t)$ is concave with respect to $b_k^{(n)}$ with $k\in[2]$, which has the global maximizer at $(b_1^{(1)}, b_2^{(1)}) = (b_1/2, b_2/2)$.
\end{proof}

\noindent{\bf Step (1)}. In this step, we prove the following lemma.
\begin{lemma}\label{pro: decreasing threshold w nonoblivious}
If  $f_i(b_1^{(1)}, b_2^{(1)}; l)$ with $i\in\{1,2,3\}$ has the global minimizer at $(b_1^{(1)}, b_2^{(1)}) = (b_1/2, b_2/2)$, respectively, then   $f_i(b_1^{(1)}, b_2^{(1)}; l+1)$ with $i\in\{1,2,3\}$ has the global minimizer at  $(b_1^{(1)}, b_2^{(1)}) = (b_1/2, b_2/2)$, respectively.
\end{lemma}
\begin{proof}
Firstly,  taking the expectation on both sides of \eqref{eq: upper bound H1 NS}, we have
\begin{align*}
&f_2(b_1^{(1)}, b_2^{(1)}; l+1) \\
&=\mathbb{E}[(1-\sum_{n\in[2]}b_k^{(n)}r_{nk}(l))h_{mk}(l)]+2.
\end{align*}
 Since $1-\sum_{n\in[2]}b_k^{(n)}r_{nk}(l)$ represents the probability of receiving a new packet from CP $k$ in time slot $l$, a larger $1-\sum_{n\in[2]}b_k^{(n)}r_{nk}(l)$ corresponds to a higher probability of receiving a new packet from CP $k$, which in turn results in a smaller $h_{mk}(l+1)$. Thus, $h_{mk}(l+1)$ is negatively correlated with $\sum_{n\in[2]}b_k^{(n)}r_{nk}(l)$, hence is positively correlated with $1-\sum_{n\in[2]}b_k^{(n)}r_{nk}(l)$. Since $h_{mk}(l+1) $ is positively correlated with $h_{mk}(l)$, indicating that $h_{mk}(l)$ maintains a positive correlation with $1-\sum_{n\in[2]}b_k^{(n)}r_{nk}(l)$. 
Therefore, if $\big(1-\sum_{n\in[2]}b_k^{(n)}\mathbb{E}[r_{nk}(l)]\big)$ increases, 
$\sum_{k\in[2]}\mathbb{E}[h_{mk}(l)]$ will be non-decreasing, due to the positive correlation between these two quantities.

From Lemma~\ref{pro: proof q4q5}~(ii), $1-\sum_{n\in[2]}b_k^{(n)}\mathbb{E}[r_{nk}(l)]=1-q_1(k; l)$ has the global minimizer at $(b_1^{(1)}, b_2^{(1)}) = (b_1/2, b_2/2)$. By assumption, $\mathbb{E}[h_{mk}(l)]$ has the global minimizer at $(b_1^{(1)}, b_2^{(1)}) = (b_1/2, b_2/2)$. Since $1-q_1(l; t)$ and $\mathbb{E}[h_{mk}(l)]$ are positively correlated, and both have a  single symmetric point at $(b_1^{(1)}, b_2^{(1)}) = (b_1/2, b_2/2)$. By symmetry, $\mathbb{E}[(1-\sum_{n\in[2]}b_k^{(n)}r_{nk}(l))h_{mk}(l)]$ has the global minimizer at $(b_1^{(1)}, b_2^{(1)}) = (b_1/2, b_2/2)$, which implies $f_2(b_1^{(1)}, b_2^{(1)}; l+1)$ has the global minimizer at $(b_1^{(1)}, b_2^{(1)}) = (b_1/2, b_2/2)$.

Secondly,  taking the expectation on both sides of \eqref{eq: upper bound H2 NS}, we have
\begin{align*}
f_3(b_1^{(1)}, b_2^{(1)}; l+1) =&  \sum_{n\in[2]}\mathbb{E}[b_k^{(n)}(1-r_{nk}(l))g_{nk}(l)]\\
+&\mathbb{E}[(1-b_k)h_{mk}(l)]+2.
\end{align*}
By assumption, $\mathbb{E}[(1-b_k)h_{mk}(l)]$ has the global minimizer at $(b_1^{(1)}, b_2^{(1)}) = (b_1/2, b_2/2)$. Then, we will show
\begin{align*}
\sum_{n\in[2]}\mathbb{E}[b_k^{(n)}(1-r_{nk}(l))g_{nk}(l)]
\end{align*}
has the global minimizer at $(b_1^{(1)}, b_2^{(1)}) = (b_1/2, b_2/2)$. First, we need to prove the following fact.

\begin{fact}\label{eq: non oblivious sum g}
$\sum_{n\in[2]}\mathbb{E}[g_{nk}(t)]$  has the global minimizer at $(b_1^{(1)}, b_2^{(1)})=(b_1/2, b_2/2)$.
\end{fact}
\begin{proof}
We complete the proof by mathematical induction. When $t=0$, $\sum_{n\in[2]}\mathbb{E}[g_{nk}(0)]=2$, which implies $(b_1^{(1)}, b_2^{(1)})=(b_1/2, b_2/2)$ is the global minimizer. Assume that Fact~\ref{eq: non oblivious sum g} holds for all $t\leq l$. Now, we consider $t=l+1$. From \eqref{eq: nonoblivious recursion g}, we have
\begin{align*}
\sum_{n\in[2]}\mathbb{E}[g_{nk}(l+1)]=2 + \sum_{n\in[2]}\mathbb{E}[\big(1-\zeta_{nk}r_{nk}(l)\big)g_{nk}(l)].
\end{align*}
From \eqref{eq: nonoblivious recursion g}, the larger $\zeta_{nk}r_{nk}(l)$, the higher probability such that $g_{nk}(l)$ drops to $1$, so $g_{nk}(l)$ and $\zeta_{nk}r_{nk}(l)$ are negatively correlated, and $g_{nk}(l)$ and $1-\zeta_{nk}r_{nk}(l)$ are positively correlated. 

From Lemma~\ref{pro: proof q4q5}~(iii),  $\sum_{n\in[2]}\mathbb{E}[1-\zeta_{nk}r_{nk}(l)]=1-q_2(k;l)$ has the global minimizer at $(b_1^{(1)}, b_2^{(1)})=(b_1/2, b_2/2)$. By assumption, $\sum_{n\in[2]}\mathbb{E}[g_{nk}(l)]$ has the global minimizer at $(b_1^{(1)}, b_2^{(1)})=(b_1/2, b_2/2)$. 

Utilizing a similar proof of Lemma~\ref{lem: proof symmetric OS}, we can show that both $\sum_{n\in[2]}(1-\zeta_{nk}r_{nk}(l))$ and $\sum_{n\in[2]}g_{nk}(l)$ have a single symmetric point at $(b_1^{(1)}, b_2^{(1)})=(b_1/2, b_2/2)$.  Since both $\sum_{n\in[2]}\mathbb{E}[1-\zeta_{nk}\tau_{nk}(l)]$ and $\sum_{n\in[2]}\mathbb{E}[g_{nk}(l)]$ have the global minimizer at $(b_1^{(1)}, b_2^{(1)})=(b_1/2, b_2/2)$.  Then, by symmetry, $\mathbb{E}[\big(1-\zeta_{nk}\tau_{nk}(l)]\big)g_{nk}(l)]$ has the global minimizer at $(b_1^{(1)}, b_2^{(1)})=(b_1/2, b_2/2)$, which implies $\sum_{n\in[2]}\mathbb{E}[g_{nk}(l+1)]$ has the global minimizer at $(b_1^{(1)}, b_2^{(1)})=(b_1/2, b_2/2)$.
\end{proof}

\noindent Now, we go back to the proof. From  Lemma~\ref{pro: proof q4q5}~(ii), 
\begin{align*}
\mathbb{E}[\sum_{n\in[2]}b_k^{(n)}(1-r_{nk}(l))] = b_k - q_1(k; l)
\end{align*}
has the global minimizer at $(b_1^{(1)}, b_2^{(1)})=(b_1/2, b_2/2)$. Note that  $\sum_{n\in[2]}b_k^{(n)}\big(1-r_{nk}(l)\big)$ represents the probability of the event that CU $m$ sends requests for CPs to eRRHs, but all eRRHs retrive cached packets to serve the requests. Thus, the larger $\sum_{n\in[2]}b_k^{(n)}(1-r_{nk}(l))$, the less likely for eRRHs to receive new content from CPs. This implies $\sum_{n\in[2]}b_k^{(n)}(1-r_{nk}(l))$ and $\sum_{n\in[2]}g_{nk}(l)$ are positively correlated. From Fact~\ref{eq: non oblivious sum g}, $\sum_{n\in[2]}\mathbb{E}[g_{nk}(l)]$ has the global minimizer at $(b_1^{(1)}, b_2^{(1)})=(b_1/2, b_2/2)$. And $(b_1^{(1)}, b_2^{(1)})=(b_1/2, b_2/2)$ is the single symmetric point for both  $\sum_{n\in[2]}b_k^{(n)}(1-r_{nk}(l))$ and $\sum_{n\in[2]}g_{nk}(l)$ have a single symmetric point at $(b_1^{(1)}, b_2^{(1)})=(b_1/2, b_2/2)$. Then, by symmetry, $f_3(b_1^{(1)}, b_2^{(1)}; l+1)$ has the global minimizer at $(b_1^{(1)}, b_2^{(1)})=(b_1/2, b_2/2)$.

Thirdly, utilizing a similar contradiction method in {\bf Step (3)} of the proof of Lemma~\ref{pro: decreasing threshold w} in Appendix~\ref{App: lem: optimal H h OS}. We can prove $f_1(b_1^{(1)}, b_2^{(1)}; l+1)$ has the global minimizer at $(b_1^{(1)}, b_2^{(1)})=(b_1/2, b_2/2)$. Finally, we complete the proof.
\end{proof}

\noindent{\bf Step (2)}. Next, we define 
\begin{align*}
w'(r_k, r_{3-k}; l+1) &=  \sup\Big\{b_k: \mathbb{E}[h_{mk}(l+1)]\text{ can not be}\\
&\text{ minimized at }(b_1^{(1)}, b_2^{(1)}) = (b_1/2, b_2/2)\Big\}.
\end{align*}
By a similar contradiction proof to that of Lemma~\ref{pro: b<w is not a minimum point}, we can show: For $\forall b_k < w'(r_k, r_{3-k}; l+1)$, $(b_1^{(1)}, b_2^{(1)})=(b_1/2, b_2/2)$ can not be the global minimizer of $\mathbb{E}[h_{mk}(l+1)]$.  

From {\bf Step (1)} and {\bf Step (2)}, we completes the mathmatical induction for Lemma~\ref{App: lem: optimal H h NS}.  Finally, under stationary randomized policies,  both $h_{mk}(l)$ and $g_{nk}(l)$ converge to asymptotically stationary processes. Hence, the initial values $h_{mk}(0)$ and $g_{nk}(0)$ have negligible influence on the long-term behavior of $h_{mk}(l)$ and $g_{nk}(l)$, respectively. Thus, $h_{mk}(l+1) = h_{mk}(l)$ almost surely when $l\to\infty$, which implies $\lim_{l\to\infty}w'(r_k, r_{3-k}; l)$ converges.

\section{Proof of Lemma~\ref{lem: optimal H1H2h NS}}\label{App: optimal H1H2h NS}
This proof consists of two parts.  In the first part, we  demonstrate that $(b_1, b_2) = (b/2, b/2)$ is the global minimizer of $\sum_{k\in[2]}\mathbb{E}[h_{mk}(t)]$, assuming $r_1$ and $r_2$ are fixed. In the second part, we show that $(r_1, r_2) = (r, 0)$ when $r\leq1$, and $(r_1, r_2) = (1, r-1)$ when $r > 1$, constitutes the global minimizer of $\sum_{k\in[2]}\mathbb{E}[h_{mk}(t)]$ for any $(b_1, b_2) = (b/2, b/2)$ with fixed $b$.

\noindent{\bf Part I}. we  demonstrate that $(b_1, b_2) = (b/2, b/2)$ is the global minimizer of $\sum_{k\in[2]}\mathbb{E}[h_{mk}(t)]$, assuming $r_1$ and $r_2$ are fixed. 

This part is similar to the proof of Lemma~\ref{lem: optimal H1H2h OS} in Appendix~\ref{App: optimal H1H2h OS}. We complete the proof by mathematical induction. 

\noindent{\bf Step (1)}. This step is similar to {\bf Step 1} in Appendix~\ref{App: optimal H1H2h OS}. We check initial cases, and omit the calculation process for brevity.

\noindent{\bf Step (2)}. We prove that $\sum_{k\in[2]}\mathbb{E}[h_{mk}(l+1)]$ has the global minimizer $(b_1, b_2) = (b/2, b/2)$ under the condition that, for both $k\in[2]$, an optimal choice of $\big((b_k^{(1)})^*, (b_k^{(2)})^*\big)$ takes the form $\big((b_k^{(1)})^*, (b_k^{(2)})^*\big) = (b_k, 0)$. 
\begin{equation}\label{eq: proof hmkg-1}
	\begin{aligned}
		h_{mk}(l+1)= \left\{
		\begin{aligned}
			&g_{1k}(l)&& b_{k}\\
			&h_{mk}(l)+1&&1-b_k,
		\end{aligned}
		\right.
	\end{aligned}
\end{equation}
where $g_{1k}(t)$ is given in \eqref{eq: nonoblivious recursion g}. Taking the expectation on both sides of \eqref{eq: proof hmkg-1}, we have
\begin{align*}
\mathbb{E}[h_{mk}(l+1)] =& \frac{1}{b_k} + \mathbb{E}[g_{1k}(l)] - 1.
\end{align*}
Following a similar argument in Fact~\ref{eq: non oblivious sum g} in Appendix~\ref{App: lem: optimal H h NS}, we show that $\sum_{k\in[2]}\mathbb{E}[g_{1k}(l)]$ has the global minimizer at $(b_1, b_2)=(b/2, b/2)$. By convexity, it is straightforward to verify that $\sum_{k\in[2]}\frac{1}{b_k}$ has the global minimizer at  $(b_1, b_2)=(b/2, b/2)$. Thus, in this case, 
$\sum_{k\in[2]}\mathbb{E}[h_{mk}(l+1)]$ has the global minimizer at $(b_1, b_2) = (b/2, b/2)$

\noindent{\bf Step (3)}. In this step, we prove that $\sum_{k\in[2]}\mathbb{E}[h_{mk}(l+1)]$ has the global minimizer at $(b_1, b_2) = (b/2, b/2)$ under the condition that, for both $k\in[2]$, an optimal choice of $\big((b_k^{(1)})^*, (b_k^{(2)})^*\big)$ takes the form $\big((b_k^{(1)})^*, (b_k^{(2)})^*\big) = (b_k/2, b_k/2)$.   Utilizing a similar proof of Lemma~\ref{pro: decreasing threshold w nonoblivious} in Appendix~\ref{App: lem: optimal H h NS}, we can show that $\sum_{k\in[2]}\mathbb{E}[h_{mk}(t)]$ has the global minimizer at $(b_1, b_2)= (b/2, b/2)$. 

\noindent{\bf Step (4)}. Utilizing a similar proof in {\bf Step 4} in Appendix~\ref{App: optimal H1H2h OS}, we can demonstrate that the following case cannot exist: where an optimal choice of $\big((b_1^{(1)})^*, (b_1^{(2)})^*\big)$ takes the form $\big((b_1^{(1)})^*, (b_1^{(2)})^*\big) = (b_1, 0)$, and an optimal choice of $\big((b_2^{(1)})^*, (b_2^{(2)})^*\big)$ takes the form $\big((b_2^{(1)})^*, (b_2^{(2)})^*\big) = (b_2/2, b_2/2)$.

\noindent{\bf Part II}. We prove that $(r_1, r_2) = (r, 0)$ if $r\leq1$ or $(r_1, r_2) = (1, r-1)$ if $r > 1$ is the global minimizer of $\sum_{k\in[2]}\mathbb{E}[h_{mk}(t)]$ for any $(b_1, b_2) = (b/2, b/2)$ with fixed $b$.

As we discussed in the proof of Lemma~\ref{pro: decreasing threshold w nonoblivious}, $\sum_{n\in[2]}b_k^{(n)}r_{nk}(l)$, which implies a negative correlated with $h_{mk}(l)$. Therefore,  $\sum_{k\in[2]}\sum_{n\in[2]}b_k^{(n)}r_{nk}(l)$ is negatively correlated with $\sum_{k\in[2]}h_{mk}(l)$. 

From {\bf Part I}, for any $r_1$ and $r_2$,  an optimal choice of $\big((b_k^{(1)})^*, (b_k^{(2)})^*\big)$ takes the form either $\big((b_k^{(1)})^*, (b_k^{(2)})^*\big) = (b/2, 0)$ or $(b/4, b/4)$. In each case, we can compute that 
\begin{align*}
\sum_{k\in[2]}\sum_{n\in[2]}b_k^{(n)}r_{nk}(l) = \frac{br}{2}
\end{align*}
is a constant. 

Under the case where, for $k\in[2]$, an optimal choice of $\big((b_k^{(1)})^*, (b_k^{(2)})^*\big)$ takes the form $\big((b_k^{(1)})^*, (b_k^{(2)})^*\big) = (b/2, 0)$,  $h_{mk}(l)$ follows the recursion \eqref{eq: proof hmkg-1}. In this case, we observe that $\sum_{k\in[2]}h_{mk}(l)$ is positively correlated with $\sum_{k\in[2]}\mathbb{E}[g_{1k}(l)]$. As discussed in the proof of Fact~\ref{eq: non oblivious sum g} in Appendix~\ref{App: lem: optimal H h NS}, $g_{nk}(l)$ is negatively correlated with $\zeta_{nk}r_{nk}(l)$, leading to a negative correlation between $\sum_{k\in[2]}h_{mk}(l)$ and $\sum_{k\in[2]}\zeta_{nk}r_{nk}(l)$. In the case where, for $k\in[2]$, an optimal choice of $\big((b_k^{(1)})^*, (b_k^{(2)})^*\big)$ takes the form $\big((b_k^{(1)})^*, (b_k^{(2)})^*\big) = (b/4, b/4)$. Note that   $\sum_{k\in[2]}h_{mk}(l)$ is positively correlated with $\sum_{k\in[2]}\sum_{n\in[2]}g_{nk}(l)$,  resulting in a negative correlation between $\sum_{k\in[2]}h_{mk}(l)$ and $\sum_{k\in[2]}\zeta_{nk}r_{nk}(l)$. Therefore, if $\sum_{k\in[2]}\mathbb{E}[\zeta_{nk}r_{nk}(l)]$ increases, $\sum_{k\in[2]}\mathbb{E}[h_{mk}(l)]$ will be non-increasing.

From Lemma~\ref{pro: proof q4q5}~(i) in Appendix~\ref{App: lem: optimal H h NS}, $\zeta_{nk}\mathbb{E}[r_{nk}(t)]$ increases with $r_1$. Note that $0\leq r_1\leq 1$ Thus, when $r\leq 1$, $\sum_{k\in[2]}\mathbb{E}[\zeta_{nk}r_{nk}(l)]$ attains its maximum when $r_1=r$, and when $r>1$,  $\sum_{k\in[2]}\mathbb{E}[\zeta_{nk}r_{nk}(l)]$ attains its maximum when $r_1=1$, thereby ensuring that $\sum_{k\in[2]}\mathbb{E}[h_{mk}(l)]$ attains its minimum. Therefore, the optimal allocation is $(r_1, r_2) = (r, 0)$ when $r \leq 1$, and $(r_1, r_2) = (1, r - 1)$ when $r > 1$, which minimizes $\sum_{k\in[2]}\mathbb{E}[h_{mk}(t)]$ for any fixed $b$.

\section{Generalizations}\label{App: Generalizations}

In this section, we extend our results by exploring two generalizations:
\begin{itemize} 
\item [](i) Collision channels — where the communication links from CUs to eRRHs may experience collisions (see Appendix~\ref{subApp: Collision Channels}). 
\item [](ii) Larger networks — where the system is scaled to accommodate an arbitrary number of $M$ CUs, $N$ eRRHs, and $K$ CPs (see Appendix~\ref{subApp: Larger Networks}). 
\end{itemize}
The analytical framework and proof techniques developed in the main text remain applicable to these generalized settings. As such, when a result involves proof methods that closely parallel those already presented, we omit restating the full details for brevity and to avoid redundancy.

\subsection{Collision Channels}\label{subApp: Collision Channels}
In this section, we relax the assumption of interference-free transmissions. The updated assumption is as follows: (i) when two CUs simultaneously request the same CP from the same eRRH, a collision occurs; (ii) however, if they request different CPs from the same eRRH, no collision occurs. The second condition can be justified through effective resource management. For instance, the eRRH can implement Frequency Division Multiple Access (FDMA), allocating distinct frequency bands to different CPs, thereby avoiding collisions even when multiple CUs are served concurrently.

The key change is that, due to collisions, a request for CP $k$ is successfully received at  an eRRH only if it is issued by a single CU. For analytical tractability, we impose the constraint 
\begin{align}\label{eq: restric condition}
	\sum_{n\in[N]}b_k^{(n)}=1.
\end{align}
This assumption can be justified as follows: if $\sum_{n\in[N]}b_k^{(n)}$ is large, we can reduce the length of the original slot, whereas if it is small, we can extend the slot accordingly. While this assumption is somewhat restrictive, it is introduced solely to establish a theoretically tractable framework. 

According to \eqref{eq: restric condition}, the evolutions of AoI at the eRRHs, as described in \eqref{eq: recursion of edge node AoI0}, is updated as follows:
\begin{align}\label{eq: recursion of edge node AoI0, collisions}
g_{nk}(t+1) = \left\{
\begin{aligned}
&1&\gamma_{nk}(t)\sum_{m\in[M]}\beta_{mk}^{(n)}(t)=1\\
&g_{nk}(t)+1&\text{otherwise}.
\end{aligned}
\right.
\end{align}
Similarly, the AoI on CUs' side, as described in \eqref{eq: recursion of user AoI}, is updated as follows:
\begin{align}\label{eq: recursion of user AoI, collisions}
	&h_{mk}(t+1) =  \left\{
	\begin{aligned}
		&1&&\text{if }\tilde{\beta}_{mk}(t)\gamma_{nk}(t)=1\\
		&\tilde{h}_{mk}^{(n)}(t)+1&&\text{if }\tilde{\beta}_{mk}^{(n)}(t)\big(1-\gamma_{nk}(t)\big)=1\\
		&h_{mk}(t)+1&&\text{otherwise}
	\end{aligned}
	\right.
\end{align}
with 
\begin{align*}
&\tilde{\beta}_{mk}(t) = \sum_{n}\beta_{mk}^{(n)}(t)\prod_{j\neq m}\big(1-\beta_{jk}^{(n)}(t)\big)\\
&\tilde{\beta}_{mk}^{(n)}(t) = \beta_{mk}^{(n)}(t)\prod_{j\neq m}\big(1-\beta_{jk}^{(n)}(t)\big)\\
&\tilde{h}_{mk}^{(n)}(t) = \min\{h_{mk}(t), g_{nk}(t)\}.
\end{align*}
By convention, we set $h_{mk}(0)=g_{nk}(0)=1$. 
For clarity, we adopt the following notation throughout the rest of this section:
\begin{align}\label{eq: zetac}
\zeta_{nk}^{c}=b_k^{(n)}(1-b_k^{(n)}).
\end{align}

\subsubsection{Optimal OSR Strategies}

Firstly, we derive the closed form of $L_e$ in the following proposition.
\begin{proposition}\label{pro: closed form Je, oblivious, C}
Let $\zeta_{nk}^{c}$ be in \eqref{eq: zetac}. The closed-form expression for the average AoI of eRRHs is provided by 
\begin{align}\label{eq: closed form Je, oblivious, C}
L_e=\frac{1}{2}\sum_{n\in[2]}\sum_{k\in[k]}\frac{1}{r_k\zeta_{nk}^c}.
\end{align}
\end{proposition}
\begin{proof}
The proof is very similar to that of Theorem~\ref{thm: closed form Je} given in Appendix~\ref{App: closed form Je}. The key difference lies in the substitution of $\zeta_{nk}$ with $2\zeta_{nk}^c$.
\end{proof}

Next, we aim to obtain an optimal OSR strategy  $\pi^*$ for $J_c$ in \eqref{eq: CUs average AoI}.

\begin{proposition}\label{pro: b r optimal, Oblivious, C}
Let \eqref{eq: restric condition} hold and $r$ be fixed. The optimal policy $\pi^*$ is given by
\begin{align}\label{eq: optimal OSR-1, Oblivious, C}
\pi^* = \{b_{mk}^{(n)}=\frac{1}{4}, r_k=\frac{r}{2}\}.
\end{align}
\end{proposition}
\begin{proof}
Similar to  \eqref{eq: upper bound H1} and \eqref{eq: upper bound H2}, we define two stochastic processes,
\begin{align}\label{eq: upper bound H1, Oblivious, C}
H_{mk, 1}(t+1) = \left\{
\begin{aligned}
&1&& \sum_{n=1}^{2}\zeta_{nk}^c r_k\\
&h_{mk}(t)+1&&1-\sum_{n=1}^{2}\zeta_{nk}^c r_k;
\end{aligned}
\right.
\end{align}
and
\begin{align}\label{eq: upper bound H2, Oblivious, C}
H_{mk, 2}(t+1) = \left\{
\begin{aligned}
			&1&& \sum_{n=1}^{2}\zeta_{nk}^c r_k\\
			&g_{1k}(t)+1&& \zeta_{1k}^c(1-r_k)\\
			&g_{2k}(t)+1&&\zeta_{2k}^c(1-r_k)\\
			&h_{mk}(t)+1&&1 - \sum_{n=1}^{2}\zeta_{nk}^c,
		\end{aligned}
		\right.
	\end{align}
	with $H_{mk, 1}(0) = H_{mk, 2}(0)=1$. From \eqref{eq: recursion of user AoI, collisions}, \eqref{eq: upper bound H1, Oblivious, C}, and \eqref{eq: upper bound H2, Oblivious, C}, 
	\begin{align*}
	h_{mk}(t) \overset{d}{=} \min\{H_{mk, 1}(t), H_{mk, 2}(t)\}
	\end{align*}
	where $\overset{d}{=}$ represents equality {\it in distribution}. 

By replacing \eqref{eq: upper bound H1} and \eqref{eq: upper bound H2} with \eqref{eq: upper bound H1, Oblivious, C} and \eqref{eq: upper bound H2, Oblivious, C}, respectively, and following a similar argument to that of Lemma~\ref{lem: optimal H h OS} in Appendix~\ref{App: lem: optimal H h OS}, we analyze the case where $b_1 = b_2 = 1$. Using the same approach, it can be shown that when $t = 2$, the function $\mathbb{E}[h_{mk}(2)]$ is convex with respect to $b_k^{(1)}$ for all $n, k \in [2]$. This implies that the symmetric point $(b_k^{(1)}, b_k^{(2)}) = (1/2, 1/2)$ is the global minimizer of $\mathbb{E}[h_{mk}(2)]$. Furthermore, by applying Lemma~\ref{pro: decreasing threshold w} in Step 3 of Appendix~\ref{App: lem: optimal H h OS}, the same point remains a minimizer for $\mathbb{E}[h_{mk}(t)]$ for all $t \geq 3$.

The proof follows a similar structure to that of Theorem~\ref{thm: b r optimal}.  

\end{proof}

\subsubsection{Optimal NSR Policies}

Utilizing a similar framework as in Section~\ref{sec: optimal NSR policies}, we can derive the closed-form expression for $L_e$ under non-oblivious policies.  By replacing $\zeta_{nk}$ with $2\zeta_{nk}^c$, we obtain $\pi_{\phi}^c$ in \eqref{eq: steady-state distribution phi} and $F_{nk}^c(\phi)$  in \eqref{eq: systems of Fnk-1}.
\begin{proposition}\label{pro: expectation of Xnk, non-oblivious, C}
The closed-form expression for the average AoI of eRRHs is provided by 
\begin{align}\label{eq: closed form Je, non-oblivious, c}
L_e=\frac{1}{4}\sum_{n\in[2]}\sum_{k\in[2]}\sum_{\phi\in\mathcal{G}_{n}^z}F_{nk}^c(\phi)\pi_\phi^c,
\end{align}
where $\pi_\phi^c$ is provided by \eqref{eq: steady-state distribution phi} and $F_{nk}^c(\phi)$ is provided by \eqref{eq: systems of Fnk-1}.
\end{proposition}
\begin{proof}
The proof closely follows the structure of Theorem~\ref{thm: expectation of Xnk}. The key difference lies in the substitution of $\zeta_{nk}$ with $2\zeta_{nk}^c$.
\end{proof}

\begin{proposition}\label{thm: b r optimal NS, non-oblivious, c}
Let \eqref{eq: restric condition} hold and $r$ be fixed. The optimal policy $\pi^*$ is given by
\begin{align}\label{eq: optimal NSR-1 collision}
	\pi^* = \Big\{
	b_{mk}^{(n)}=\frac{1}{4}, (r_1, r_2)\text{ is given in \eqref{eq: optimal r}}
	\Big\}.
\end{align}
\end{proposition}
\begin{proof}
Let $r_{nk}(t)$ be in \eqref{eq: taukt}, and let $\zeta_{nk}^c$ be as given in \eqref{eq: zetac}. Following the formulations in \eqref{eq: upper bound H1 NS} and \eqref{eq: upper bound H2 NS}, we introduce two stochastic processes:
\begin{align}\label{eq: upper bound H1 NS, non-oblivious, c}
H_{mk, 1}'(t+1) = \left\{
\begin{aligned}
&1&& \sum_{n\in[2]}\zeta_{nk}^cr_{nk}(t)\\
&h_{mk}(t)+1&&1-\sum_{n\in[2]}\zeta_{nk}^cr_{nk}(t)
\end{aligned}
\right.
\end{align}
and
\begin{align}\label{eq: upper bound H2 NS, non-oblivious, c}
H_{mk, 2}'(t+1) = \left\{
\begin{aligned}
&1&&\sum_{n\in[2]}\zeta_{nk}^cr_{nk}(t)\\
&g_{1k}(t)+1&& \zeta_{1k}^c\big(1-r_{1k}(t)\big)\\
&g_{2k}(t)+1&&\zeta_{2k}^c\big(1-r_{2k}(t)\big)\\
&h_{mk}(t)+1&&1-\sum_{n\in[2]}\zeta_{nk}^c,
\end{aligned}
\right.
\end{align}
From \eqref{eq: recursion of user AoI, collisions}, \eqref{eq: upper bound H1 NS, non-oblivious, c}, and \eqref{eq: upper bound H2 NS, non-oblivious, c}, we have
\begin{align*}
h_{mk}(t) \overset{d}{=} \min\{H_{mk, 1}'(t), H_{mk, 2}'(t)\},
\end{align*}
where $\overset{d}{=}$ represents equality {\it in distribution}. 

By replacing  \eqref{eq: upper bound H1 NS} and \eqref{eq: upper bound H2 NS} with \eqref{eq: upper bound H1 NS, non-oblivious, c} and \eqref{eq: upper bound H2 NS, non-oblivious, c}, respectively, and following a similar argument to that of Lemma~\ref{lem: optimal H h NS} in Appendix~\ref{App: lem: optimal H h NS}, we analyze the case where $b_1 = b_2 = 1$. Using the same approach, it can be shown that when $t = 2$, the function $\mathbb{E}[h_{mk}(2)]$ is convex with respect to $b_k^{(1)}$ for all $n, k \in [2]$. This implies that the symmetric point $(b_k^{(1)}, b_k^{(2)}) = (1/2, 1/2)$ is the global minimizer of $\mathbb{E}[h_{mk}(2)]$. Furthermore, by applying Lemma~\ref{pro: decreasing threshold w nonoblivious} in Step 3 of Appendix~\ref{App: lem: optimal H h NS}, the same point remains a minimizer for $\mathbb{E}[h_{mk}(t)]$ for all $t \geq 3$.
	
The remaining proof closely follows that of Theorem~\ref{thm: b r optimal NS}. 
\end{proof}

\subsubsection{Simulations}

We verify our findings through simulations. Let $M=N=K=2$. For non-oblivious policies, we let pre-determined threshold $z=3$. 

\begin{figure}[htbp]
	\centering
net	\includegraphics[height=5cm, width=6.5cm]{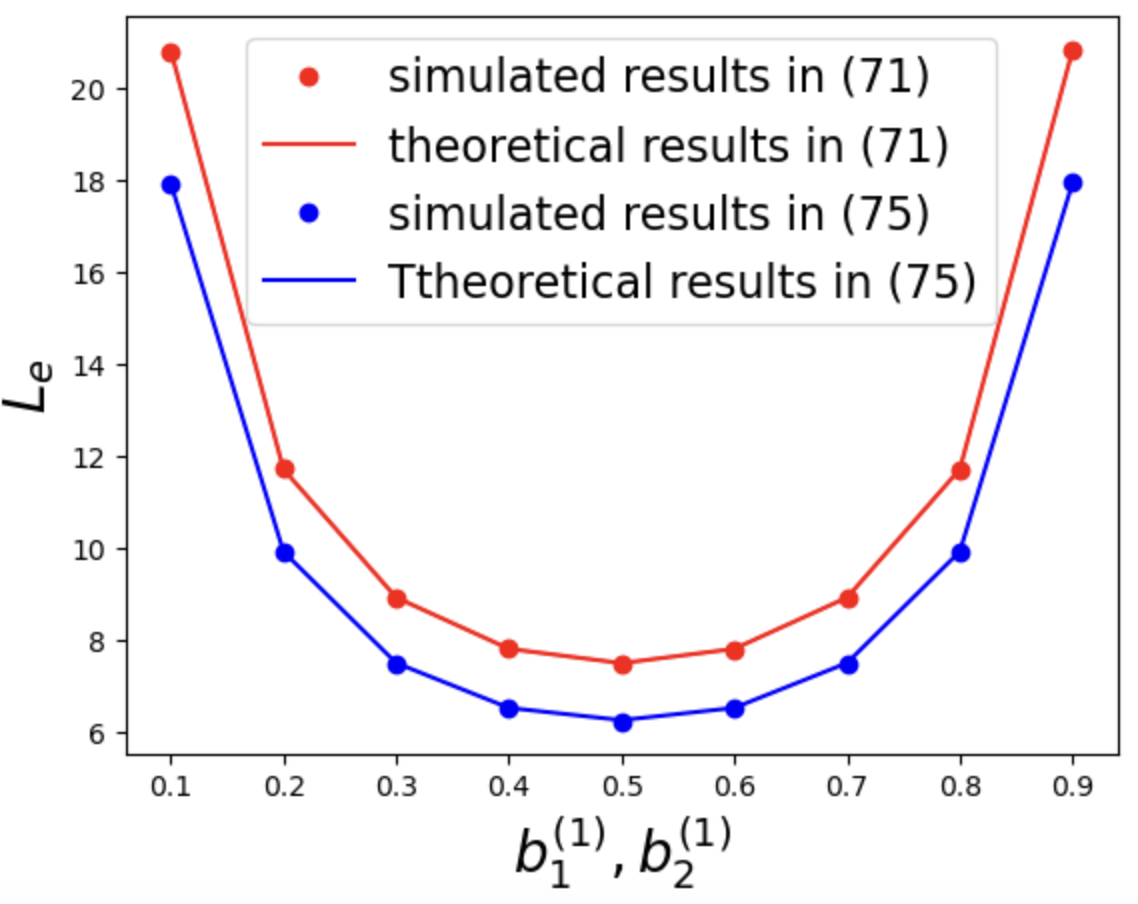}
	\caption{The theoretical and numerical $L_e$ under both oblivious and non-oblivious policies.}
	\label{Le_c}
\end{figure} 
In Fig.~\ref{Le_c},  we let $r_1=0.4$, $r_2=0.2$, $b_1=b_2=1$, and $z=3$ for non-oblivious policies, while varying $b_1^{(1)}=b_2^{(1)}$ within the range $[0.1, 0.9]$. The closed-form expressions for $L_e$ align perfectly with the simulation results, confirming the accuracy of Proposition~\ref{pro: closed form Je, oblivious, C} and Proposition~\ref{pro: expectation of Xnk, non-oblivious, C}. 

\begin{figure}[htbp]
	\centering
	\includegraphics[height=4.5cm, width=8cm]{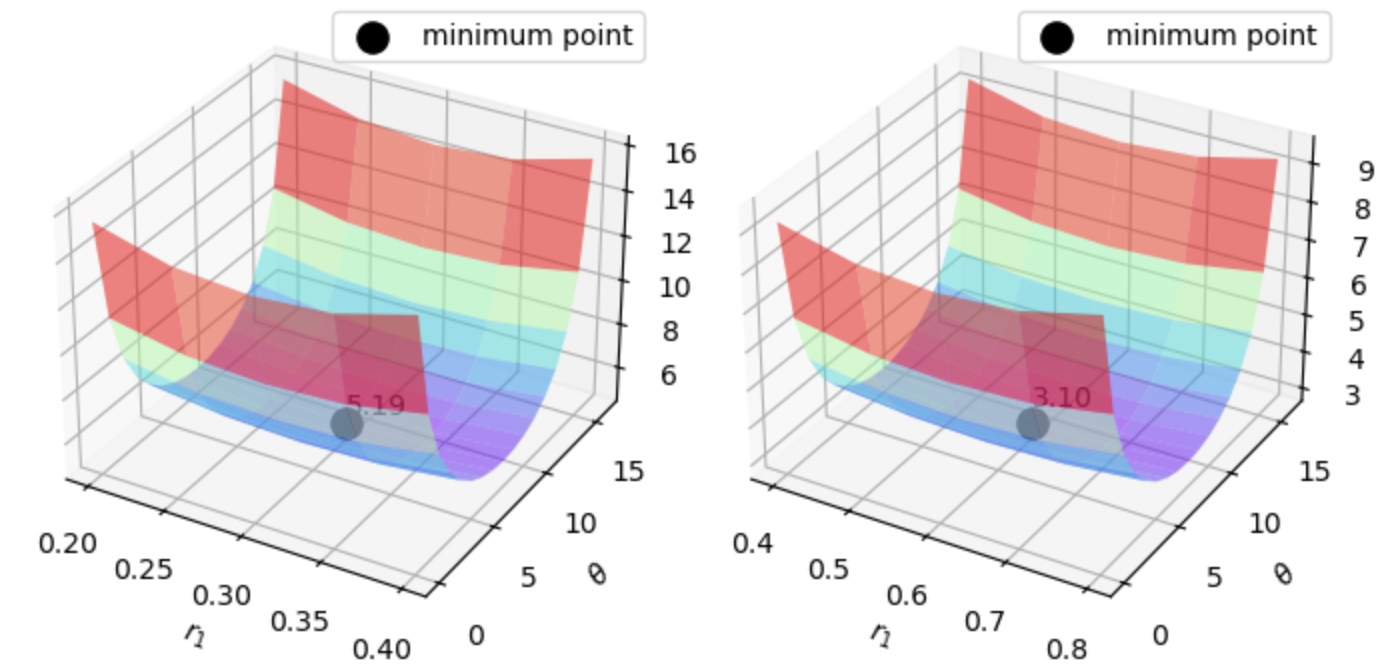}
	\caption{Performance of $J_c$ under oblivious policies for $r=0.6$ (left) and $b=1.2$ (right).}
	\label{network2_c_O}
\end{figure} 

\begin{figure}[htbp]
	\centering
	\includegraphics[height=4.5cm, width=8cm]{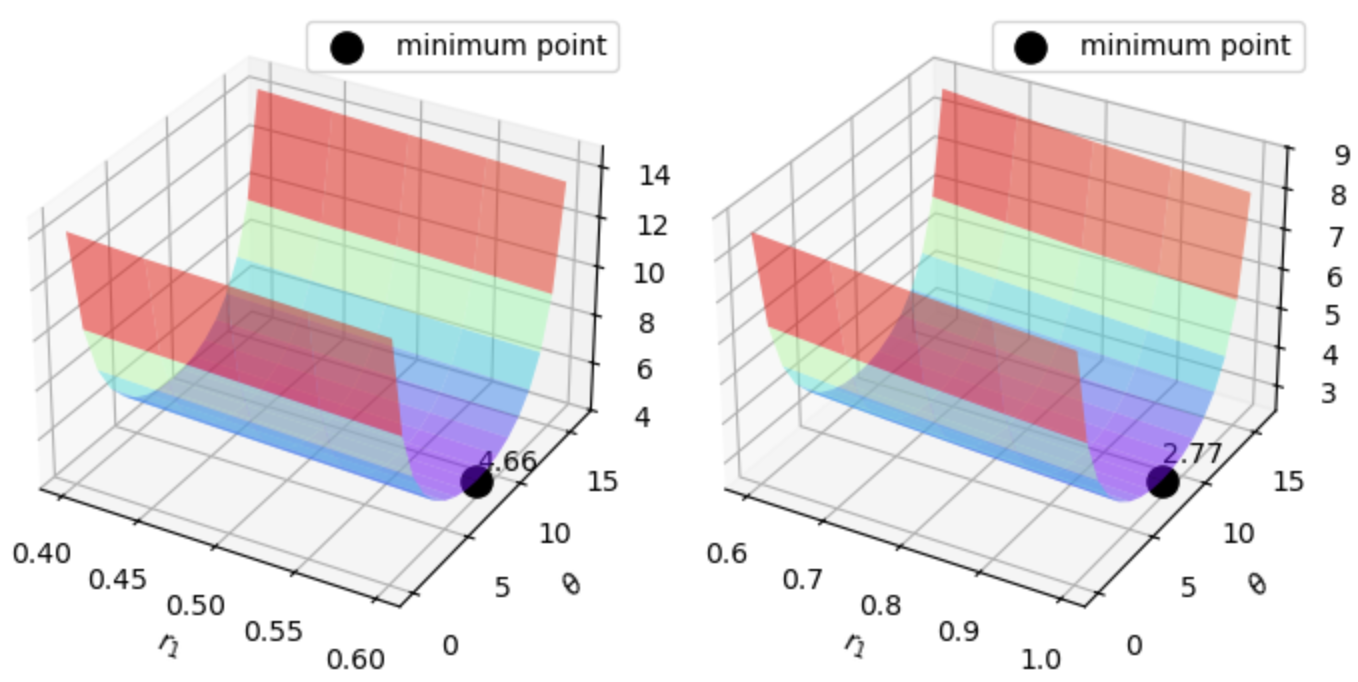}
	\caption{Performance of $J_c$ under non-oblivious policies for $r=0.6$ (left) and $r=1.2$ (right).}
	\label{network2_c_N}
\end{figure} 

We verify the optimality of both the oblivious and non-oblivious policies as stated in Proposition~\ref{pro: b r optimal, Oblivious, C} and Proposition~\ref{thm: b r optimal NS, non-oblivious, c}, respectively, using the results presented in Fig.~\ref{network2_c_O} and Fig.~\ref{network2_c_N}. In these figures:
(i) the $x$-axis shows $r_1$;
(ii) the $y$-axis shows the index $\theta\in[17]$, which determines the value of $b_k^{(1)}$ via the relation $b_k^{(1)} = \frac{b_k}{17}\cdot\theta = \frac{\theta}{17}$ (with $b_1 = b_2 = 1$ fixed);
(iii) the $z$-axis shows the resulting value of $J_c$.
We consider two values for the total rate, $r \in \{0.6, 1.2\}$. 
From Fig.~\ref{network2_c_O}, under oblivious policies, $J_c$ achieves its minimum at $(b_1^{(1)}, b_1^{(2)}, b_2^{(1)}, b_2^{(2)}) = (1/4, 1/4, 1/4, 1/4)$ and $r_1=r_2=r/2$, which aligns with the optimal solution described in Proposition~\ref{pro: b r optimal, Oblivious, C}. From Fig.~\ref{network2_c_N}, under non-oblivious policies, $J_c$ achieves its minimum at $(b_1^{(1)}, b_1^{(2)}, b_2^{(1)}, b_2^{(2)}) = (1/4, 1/4, 1/4, 1/4)$ and $r_1=\min\{r, 1\}, r_2=r-\min\{r, 1\}$, which aligns with the optimal solution described in Proposition~\ref{thm: b r optimal NS, non-oblivious, c}.

\subsection{Larger Networks}\label{subApp: Larger Networks}
In this section, we extend the framework to support larger networks by considering a general F-RAN architecture comprising $M$ CUs, $M$ eRRHs, and $K$ CPs. To maintain analytical tractability under this extended setting, we focus {\it exclusively} on oblivious policies. It is important to note that the recursive formulations for $g_{nk}(t)$ and $h_{mk}(t)$ remain unchanged and are still given by \eqref{eq: recursion of edge node AoI0} and \eqref{eq: recursion of user AoI}, respectively.  For convenience, we define:
\begin{align}\label{eq: zetalnk}
\zeta^l_{nk}=r_k\big(1-(1-b_k^{(n)})^N\big).
\end{align}
\subsubsection{Optimal OSR Strategies}\label{subsec: Optimal Oblivious Stationary Strategies}

Firstly, we derive the closed form of $L_e$ in the following proposition.
\begin{proposition}\label{pro: closed form Je, oblivious, L}
Let $\zeta_{nk}^l$ be in \eqref{eq: zetalnk}.  The closed-form expression for the average AoI of eRRHs is provided by 
\begin{align}\label{eq: closed form Je, oblivious, L}
L_e=\frac{1}{NK}\sum_{n=1}^{N}\sum_{k=1}^{K}\frac{1}{r_k\zeta^l_{nk}}.
\end{align}
\end{proposition}
\begin{proof}
Based on \eqref{eq: recursion of edge node AoI0},  the recursion of $g_{nk}(t)$ in general networks is given by
\begin{align}\label{eq: general g o2s}
	g_{nk}(t+1) = \left\{
	\begin{aligned}
		&1&& \text{w.p. }r_k\zeta^l_{nk}\\
		&g_{nk}(t)+1&&\text{w.p. }1-r_k\zeta^l_{nk},
	\end{aligned}
	\right.
\end{align}

The remaining proof closely follows that of Theorem~\ref{thm: closed form Je} in Appendix~\ref{App: closed form Je}, with one key distincion. Specifically, we replace \eqref{eq: zeta} and \eqref{eq: oblivious recursion g-1} with \eqref{eq: zetalnk} and \eqref{eq: general g o2s}, respectively.
\end{proof}

Next, we aim to obtain an optimal policy $\pi^*$ for $J_c$ in \eqref{eq: CUs average AoI}.
\begin{proposition}\label{pro: b r optimal, Oblivious, L}
Let $\zeta_{nk}^l$ be in \eqref{eq: zetalnk}.	Let $b, r$ be fixed, $(b, r)\in\mathcal{F}$, and $ m,n,k\in[2]$.  Let $\mathcal{K}_j$ denote any $N-j$ elements of $[N]$, i.e., $\mathcal{K}_j \subset[N]$ and $|\mathcal{K}_j| = N-j$. There exists a sequence of thresholds (depending on $r$): $w_{1}(r) > w_{2}(r)\cdots > w_{N-1}(r)$, such that if $b>w_{1}(r)$, an optimal policy $\pi^*$ is given by
	\begin{align}\label{eq: optimal OSR-1, Oblivious, L}
		\pi^* = \{b_{mk}^{(n)}=\frac{b}{MN}, r_k=\frac{r}{K}\};
	\end{align}
	if $w_{j+1}(r) < b\leq w_{j}(r)$, an optimal policy is given by
	\begin{align}\label{eq: optimal OSR-2, Oblivious, L}
		\pi^* = \Big\{&b_{mk}^{(i)}=\frac{b}{M (N-j)},\,\, i\in \mathcal{K}_j,\nonumber\\
		&b_{mk}^{(i)}= 0,\,\, i\in[N]\backslash \mathcal{K}_j,\,\, r_k=\frac{r}{K}
		\Big\};
	\end{align}
and if $b\leq w_{N-1}(r)$, an optimal policy is given by
\begin{align}\label{eq: optimal OSR-3, Oblivious, L}
\pi^* = \Big\{&b_{mk}^{(i)} = \frac{b}{M},\,\, i\in\mathcal{K}_{N-1}\nonumber\\
&b_{mk}^{(i)} = 0,\,\, i\in[N]\backslash\mathcal{K}_{N-1},\,\,r_k=\frac{r}{K}
\Big\}.
\end{align}
\end{proposition}
\begin{remark}
Proposition~\ref{pro: b r optimal, Oblivious, L} extends the insight of Theorem~\ref{thm: b r optimal}, revealing how the optimal request distribution adapts to varying demand and communication resources. When demand is high and communication resources are abundant, the optimal strategy evenly distributes requests among all eRRHs. However, as communication resources or demand decrease, the optimal approach shifts towards consolidating requests to a subset of eRRHs. Moreover, the greater the reduction in resources or demand, the smaller the subset of eRRHs that should be utilized. The values of the thresholds depend on the parameter $r$. Simulations demonstrate that, in some cases, portions of the thresholds may fall outside the range of $b$, i.e., $[0, M]$.
\end{remark}
\begin{proof}
Let $\mathcal{K}_j$ denote any $N-j$ elements of $[N]$, i.e., $\mathcal{K}_j \subset[N]$ and $|\mathcal{K}_j| = N-j$. For any fixed sequence $\{\mathcal{K}_1, \mathcal{K}_2,\cdots,\mathcal{K}_{N-1}\}$, similar to  \eqref{eq: upper bound H1} and \eqref{eq: upper bound H2}, we define the following $N$ stochastic processes,
\begin{align}\label{eq: upper bound H1, Oblivious, L}
	H_{mk, 1}(t+1) = \left\{
	\begin{aligned}
		&1&& b_k r_k\\
		&h_{mk}(t)+1&&1-b_k r_k,
	\end{aligned}
	\right.
\end{align}
and for $2\leq j\leq N$,
\begin{align}\label{eq: upper bound H2, Oblivious, L}
H_{mk, j}(t+1) = \left\{
\begin{aligned}
&1&&b_kr_k\\
&g_{i k}(t)+1&& b_k^{(i)}(1-r_k),i\in \mathcal{K}_{N-j}\\
&h_{mk}(t)+1&&1-b_k+\sum_{i\notin \mathcal{K}_{N-j}}b_k^{(i)}r_k,
	\end{aligned}
	\right.
\end{align}
with $H_{mk, j}(0) =1$, $j\in\{1,2,\cdots, N\}$. Similar to \eqref{eq: hH1H2}, From \eqref{eq: recursion of user AoI}, \eqref{eq: upper bound H1, Oblivious, L}, and \eqref{eq: upper bound H2, Oblivious, L}, 
\begin{align}\label{eq: hH1H2, Oblivious, L}
h_{mk}(t) \overset{d}{=} \min_{j\in\{1,2,\cdots,N\}}H_{mk, j}(t)
\end{align}
where $\overset{d}{=}$ represents equality {\it in distribution}.

The proof follows a structure similar to that of Theorem~\ref{thm: b r optimal}. The key difference lies in the substitution of \eqref{eq: upper bound H1} and \eqref{eq: upper bound H2} are replaced by \eqref{eq: upper bound H1, Oblivious, L} and \eqref{eq: upper bound H2, Oblivious, L}, respectively. Furthermore, the proof framework from Theorem~\ref{thm: b r optimal} must be applied iteratively $N$ times.
\end{proof}
\subsubsection{Simulations}

In this section, we verify our theoretical findings on F-RANs with $M=N=K=3$. In Fig.~\ref{Le_l},  we fix $r_1=0.4$, $r_2=0.3$, $r_3=0.2$, and  $b_k=0.8$, while varying $b_k^{(1)}=b_k^{(2)}$ within the range $(0, 0.4)$ for $k\in\{1,2,3\}$. The closed-form expressions for $L_e$ align perfectly with the simulation results, confirming the accuracy of Proposition~\ref{pro: closed form Je, oblivious, L}. 

\begin{figure}[htbp]
	\centering
	\includegraphics[height=5cm, width=6cm]{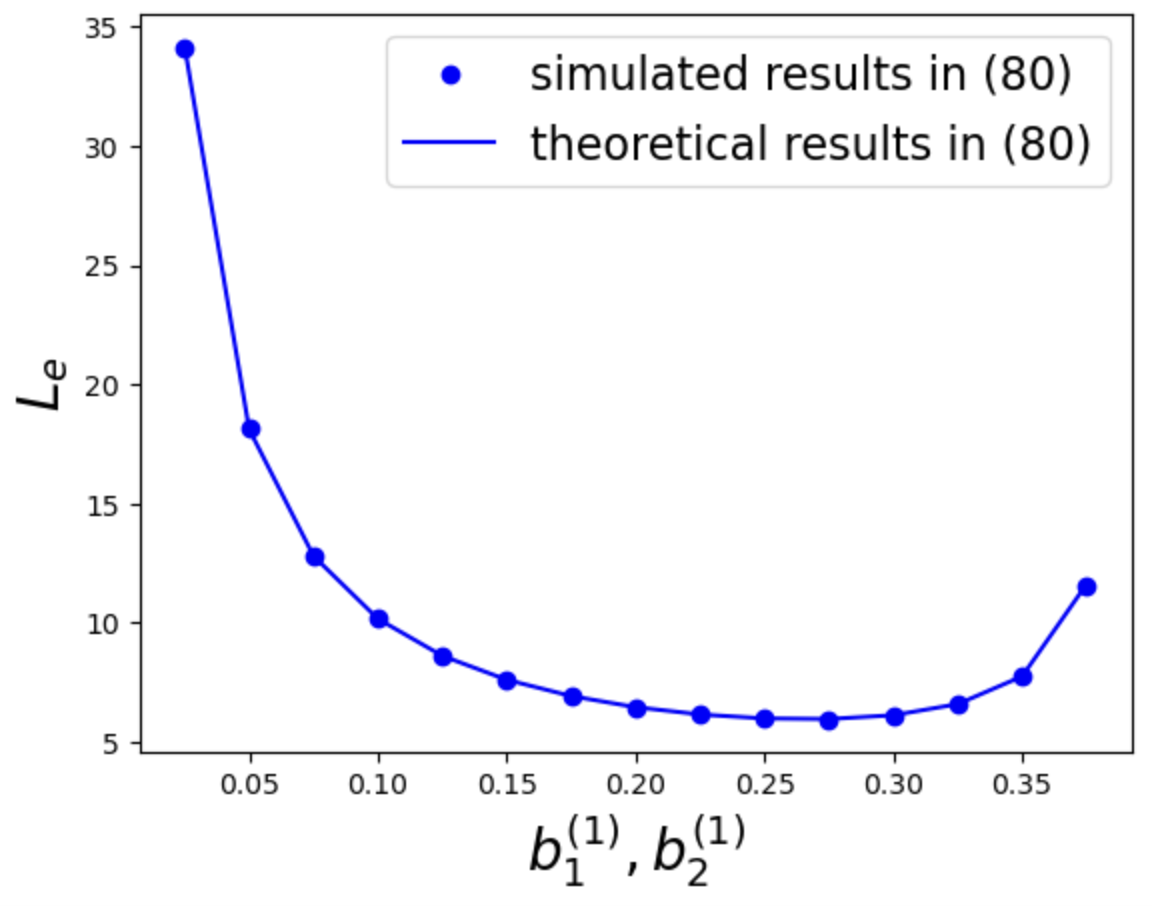}
	\caption{The theoretical and numerical $L_e$ under oblivious policies.}
	\label{Le_l}
\end{figure}

We verify the optimality of the oblivious policies as stated in Proposition~\ref{pro: b r optimal, Oblivious, L}, using the results shown in Fig.~\ref{network3_min1}, Fig.~\ref{network3_min2}, and Fig.~\ref{network3_min3}. In these figures:
(i) The $x$-axis corresponds to $r_1$.
(ii) The $y$-axis represents the index $\theta\in[21]$, which determines the value of $b_k^{(1)}$ via the relation $b_k^{(1)} = \frac{b_k}{20}\cdot\theta$. Since $b_k^{(n)}$ for $n\in\{1,2,3\}$ varies within $[0, b_k]$.  For simplicity, we set $b_k^{(2)}=b_k^{(3)} = \frac{b_k - b_k^{(1)}}{2}$.
(iii) The $z$-axis shows the resulting value of $J_c$.
We consider a fixed $r$ and three values of $b$, i.e., $r=0.9$ and $b\in\{0.3, 1.8, 2.7\}$.

\begin{figure}[t]
	\centering
	\subfigure[$(b, r)=(0.3, 0.9)$]{%
		\includegraphics[width=0.48\linewidth]{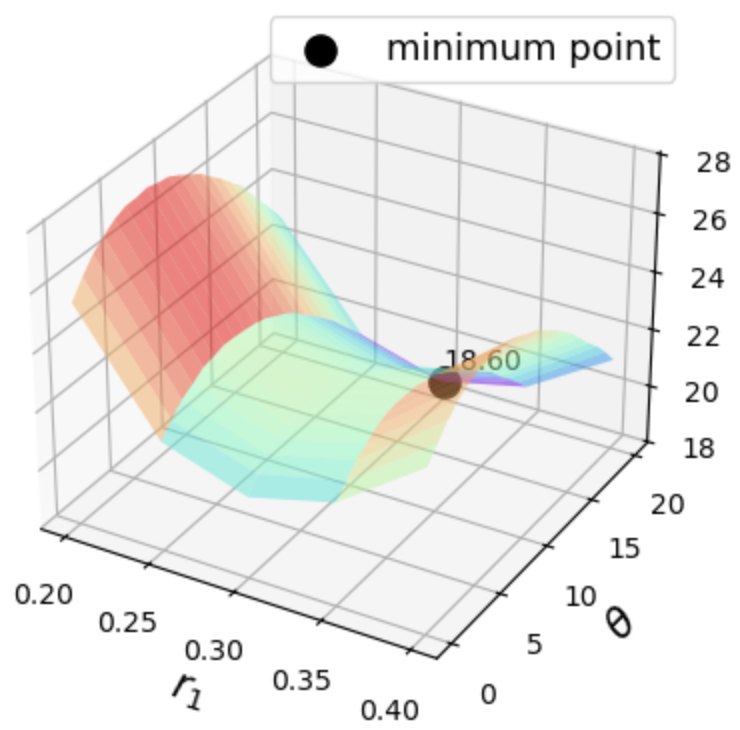}
		\label{network3_min1}}
	\subfigure[$(b, r)=(1.8, 0.9)$]{%
		\includegraphics[width=0.48\linewidth]{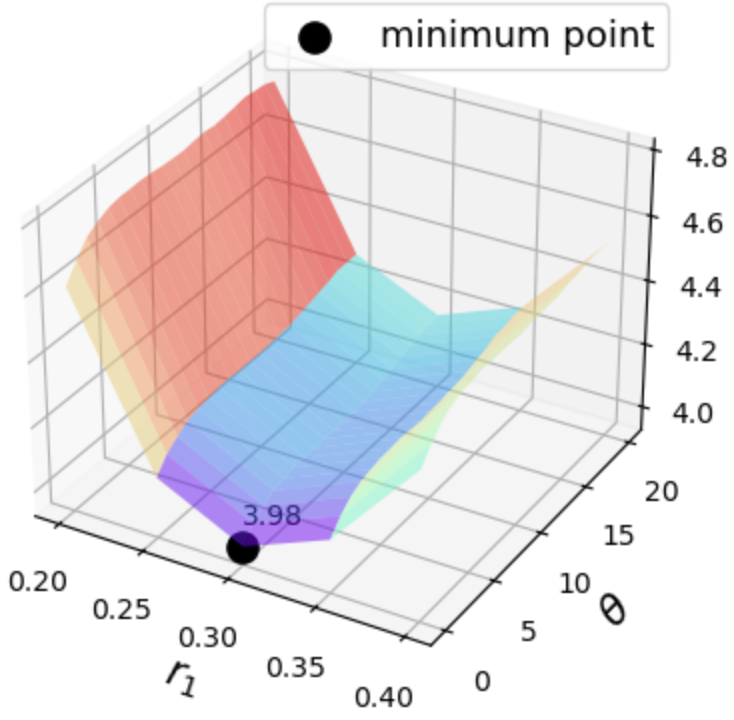}
		\label{network3_min2}}
	\subfigure[$(b, r)=(2.7, 0.9)$]{%
		\includegraphics[width=0.48\linewidth]{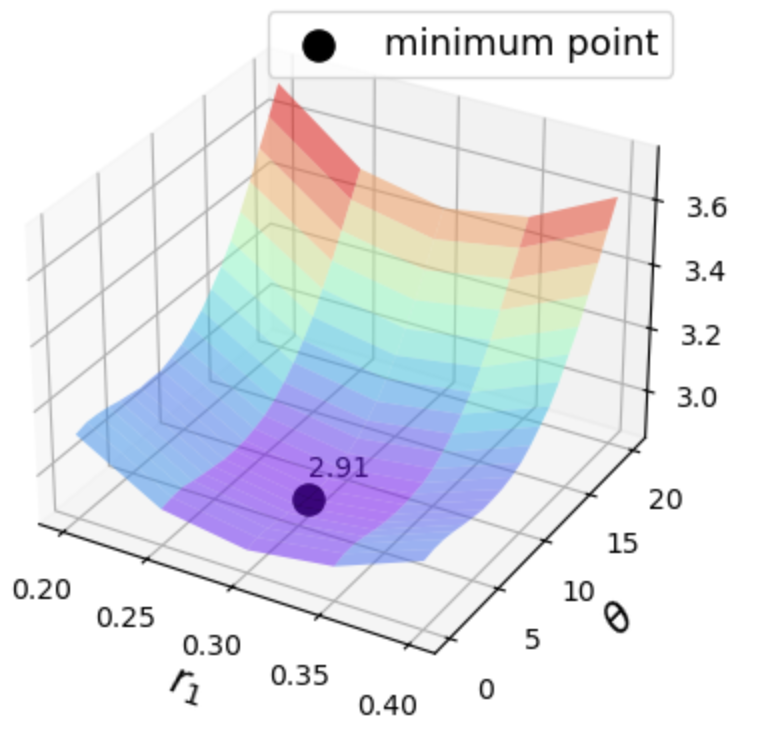}
		\label{network3_min3}}
	\caption{Performance of $J_c$ under oblivious policies with different $(b, r)$.}
	\label{fig:three_subfigs}
\end{figure}

From Fig.~\ref{network3_min1}, when $b$ is relatively small ($=0.3$), $J_c$ achieves its minimum at $(r_1, r_2, r_3, b_1,b_2,b_3) = (0.3, 0.3, 0.3, 0.1, 0.1, 0.1)$ and $(b_k^{(1)}, b_k^{(2)}, b_k^{(3)}) = (0.1, 0, 0)$, which aligns with the optimal solution described in \eqref{eq: optimal OSR-3, Oblivious, L} in Proposition~\ref{pro: b r optimal, Oblivious, L}. From Fig.~\ref{network3_min2}, when $b$ is moderate ($=1.8$), $J_c$ achieves its minimum at $(r_1, r_2, r_3, b_1,b_2,b_3) = (0.3, 0.3, 0.3, 0.6, 0.6, 0.6)$ and $(b_k^{(1)}, b_k^{(2)}, b_k^{(3)}) = (0, 0.3, 0.3)$, which aligns with the optimal solution described in \eqref{eq: optimal OSR-2, Oblivious, L} in Proposition~\ref{pro: b r optimal, Oblivious, L}. From Fig.~\ref{network3_min3}, when $b$ is relatively large ($=2.7$), $J_c$ achieves its minimum at $(r_1, r_2, r_3, b_1,b_2,b_3) = (0.3, 0.3, 0.3, 0.9, 0.9, 0.9)$ and $(b_k^{(1)}, b_k^{(2)}, b_k^{(3)}) = (0.3, 0.3, 0.3)$, which aligns with the optimal solution described in \eqref{eq: optimal OSR-1, Oblivious, L} in Proposition~\ref{pro: b r optimal, Oblivious, L}.

\end{document}